\documentclass{article} \usepackage[arxiv]{optional}
\usepackage{mathtools}
\usepackage{amsmath}
\usepackage{amsthm,bbm,amssymb}
\usepackage{color}
\usepackage{graphicx}
\usepackage[nodisplayskipstretch]{setspace}
\usepackage{esint}
\usepackage{soul}
\mathtoolsset{showonlyrefs,showmanualtags}
\usepackage{caption}
\usepackage{subcaption}
\newtheorem{Th}{Theorem}[section]
\newtheorem{Lemma}[Th]{Lemma}
\newtheorem{Cor}[Th]{Corollary}
\newtheorem{Prop}[Th]{Proposition}
\newtheorem{Asm}[Th]{Assumption}

\usepackage[ruled]{algorithm2e}
\SetAlCapHSkip{0pt} %
\usepackage{graphicx,graphics}
\usepackage{xspace}

\newcommand{\Bernoulli}{\mathrm{Bernoulli}}
\newcommand{\Uniform}{\mathrm{Uniform}}
\newcommand{\InvGamma}{\mathrm{InvGamma}}

\providecommand{\Diag}{\mathop\mathrm{Diag}}
\newcommand{\eqd}{\stackrel{d}{=}}

\newcommand{\toas}{\stackrel{\textup{a.s.}}{\to}}
\newcommand{\toL}{\stackrel{L^1}{\to}}
\newcommand{\toasL}{\stackrel{\textup{a.s.},\,L^1}{\to}}
\def\norm#1{\|{#1}\|} %
\newcommand{\twonorm}[1]{\norm{#1}_2} %
\newcommand{\opnorm}[1]{\norm{#1}_{\mathrm{op}}} %
\newcommand{\cubname}{CUB\xspace} %
\newcommand{\cub}[1][p]{\textup{\cubname}_{#1}}%
\newcommand{\defeq}{\triangleq}
\newcommand\mydots{\makebox[1em][c]{.\hfil.\hfil.}}

\newenvironment{talign*}
 {\csname align*\endcsname}
 {\endalign}
\newenvironment{talign}
 {\csname align\endcsname}
 {\endalign}

\def\spacingset#1{\renewcommand{\baselinestretch}%
{#1}\small\normalsize} \spacingset{1}

\opt{jasa}{
\usepackage[numbers]{natbib}
\bibliographystyle{chicago}

\usepackage{graphicx,psfrag,epsf}
\usepackage{enumerate}
\addtolength{\oddsidemargin}{-.5in}%
\addtolength{\evensidemargin}{-.5in}%
\addtolength{\textwidth}{1in}%
\addtolength{\textheight}{-.3in}%
\addtolength{\topmargin}{-.8in}%
\usepackage{url} 
\usepackage[unicode=false,bookmarks=false,bookmarksnumbered=false,bookmarksopen=false,breaklinks=false,pdfborder={0 0 0},pdfborderstyle={},backref=false,colorlinks=false]{hyperref}
\hypersetup{pdftitle={Bound Wasserstein},pdfauthor={Biswas Mackey},linkcolor=RoyalBlue,citecolor=RoyalBlue,urlcolor=RoyalBlue}
\usepackage[dvipsnames,svgnames,x11names,hyperref]{xcolor}
\usepackage{tikz}
}

\opt{arxiv}{
\usepackage{lmodern}
\usepackage{geometry}
\usepackage[utf8]{inputenc} %
\usepackage[T1]{fontenc}    %
\usepackage{textcomp}
\usepackage{booktabs}       %
\usepackage{amsfonts}       %
\usepackage{nicefrac}       %
\usepackage{microtype}      %

\usepackage{natbib}
\bibliographystyle{abbrvnat}
\renewcommand{\baselinestretch}{1.1}

\usepackage{url} 
\usepackage[unicode=true,bookmarks=true,bookmarksnumbered=false,bookmarksopen=false,breaklinks=false,pdfborder={0 0 0},pdfborderstyle={},backref=page,colorlinks=true]{hyperref}
\hypersetup{pdftitle={Bound Wasserstein},pdfauthor={Biswas Mackey},linkcolor=RoyalBlue,citecolor=RoyalBlue,urlcolor=RoyalBlue}
\usepackage[dvipsnames,svgnames,x11names,hyperref]{xcolor}
\usepackage{tikz}
}

\newcommand{\blind}{1}

\begin{document}
\if1\blind
{
  \title{
  \opt{jasa}{\bf} 
  Bounding Wasserstein distance with couplings}
  \author{Niloy Biswas \\
    Harvard University \\
    niloy\_biswas@g.harvard.edu \\
    \and
    Lester Mackey \\
    Microsoft Research New England \\
    lmackey@microsoft.com}
  \maketitle
} \fi

\if0\blind
{
  \bigskip
  \bigskip
  \bigskip
  \begin{center}
    {\LARGE \opt{jasa}{\bf} 
    Bounding Wasserstein distance with couplings}
\end{center}
  \medskip
} \fi

\bigskip
\begin{abstract}
Markov chain Monte Carlo (MCMC) provides asymptotically consistent estimates of 
intractable posterior expectations as the number of iterations tends to infinity. 
However, in large data applications, MCMC can be computationally expensive per iteration. 
This has catalyzed interest in approximating MCMC in a manner that improves computational speed per iteration 
but does not produce asymptotically consistent estimates.
In this article, we propose estimators based on couplings of Markov chains to assess the 
quality of such asymptotically biased sampling methods. The estimators give empirical upper 
bounds of the Wasserstein distance between the limiting distribution of the asymptotically 
biased sampling method and the original target distribution of interest. 
We establish theoretical guarantees for our upper bounds and show that our estimators 
can remain effective in high dimensions. 
We apply our quality measures to stochastic gradient MCMC, variational Bayes, 
and Laplace approximations for tall data and to approximate MCMC for Bayesian logistic regression in $4500$ dimensions and Bayesian linear regression in $50000$ dimensions.
\end{abstract}

\opt{jasa}{
\noindent%
\vfill

\newpage
\spacingset{1.45} %
}

\section{Introduction} \label{section:introduction}

\subsection{Quality of asymptotically biased Monte Carlo methods}
Markov chain Monte Carlo (MCMC) methods are commonly used for the
approximation of intractable integrals arising in Bayesian statistics, probabilistic inference, 
machine learning, and other fields \citep{brooks2000handbook}. 
They are based on a transition kernel $K_1$ which is invariant 
with respect to a target distribution of interest $P$. 
MCMC methods are asymptotically unbiased 
in that they generate Markov chains with marginal distributions that asymptotically 
converge to $P$ as the number of iterations tend to infinity.
However, in modern applications with a large number of data points or high dimensions,
evaluating the transition kernel $K_1$ at each iteration can incur high computation cost.
This has catalyzed the use of asymptotically biased sampling methods such as approximate MCMC and variational inference.
Approximate MCMC \citep[e.g.,][]{welling2011bayesianICML, bardenet2017onJMLR, narisetty2019skinnyJASA, 
johndrow2020scalableJMLR} is based on a transition kernel $K_2$ 
which is an approximation of $K_1$ with low computation cost; these approximate Markov chains typically converge to a distribution $Q$ that differs from the target $P$.
Variational inference \citep[e.g.,]{blei2017JASAvariational} alternatively uses optimization to inexactly approximate $P$ with a surrogate distribution $Q$.

Assessing the quality of such asymptotically biased samplers is of great interest for researchers who develop new approximate inference methods.
Standard MCMC diagnostic tests \cite[e.g.,][]{johnson1998coupling, biswas2019estimating, vats2020revisitingSS, vehtari2020rankBA}
are not directly suitable for such settings as they do not account for asymptotic bias. 
Researchers often resort to comparing summary statistics or marginal univariate 
traceplots of samples from such methods with samples from an asymptotically unbiased Markov chain. However, such marginal traceplots and summary statistics may fail to capture higher order moments and
dependencies between different components. Moreover, in high-dimensional settings, 
visualizing all marginal traceplots may not even be feasible.
In this manuscript, we develop generic upper bound estimates of the Wasserstein distance, 
an appealing measure of distance between distributions discussed in Sec.~\ref{subsection:wass}. 
Our estimates are then applied to assess the quality of asymptotically biased samplers.

\subsection{Couplings and Wasserstein distances} \label{subsection:wass}
Consider a complete, separable metric space $(\mathcal{X}, c)$ where $c$ is a metric. 
For each $p \geq 1$, let $\mathcal{P}_p(\mathcal{X})$ denote the 
set of all probability measures $P$ on $(\mathcal{X}, c)$ which have finite moments of order $p$, i.e., for which $\int_\mathcal{X} c(x_0,x)^p dP(x) < \infty$ for some $x_0 \in \mathcal{X}$.
Then the $p$-Wasserstein distance is a metric on $\mathcal{P}_p(\mathcal{X})$, defined for
any probability measures $P$ and $Q$ in $\mathcal{P}_p(\mathcal{X})$ as 
\begin{align}\label{eq:Wp_defn}
\mathcal{W}_p(P, Q) = 
( \inf_{ \gamma \in \Gamma( P, Q )} \textstyle\int_{\mathcal{X} \times \mathcal{X}} c(x, y)^p d\gamma(x,y) )^{1/p}
\end{align}
where $\Gamma(P, Q)$ is the set of probability measures on $\mathcal{X} \times \mathcal{X}$
with marginal measures $P$ and $Q$ respectively. 
Any probability measure in $\Gamma(P, Q)$ is called a coupling of $P$ and $Q$, 
and any coupling which attains the infimum in \eqref{eq:Wp_defn}
is called $p$-Wasserstein optimal.

The Wasserstein distance has many advantageous properties. Here we note those most relevant for this work 
and refer to \citet{villani2008optimal}
for more details. 
First, it allows comparison between mutually singular distributions 
that may have disjoint supports, unlike common alternatives like the total variation distance, Kullback–Leibler (KL) divergence and R\'enyi’s $\alpha$-divergences \citep{vanErven2014renyiIEEE}. 
Moreover, it captures geometric properties induced by the metric $c$ and differences between moments of distributions. 
For example when $\mathcal{X}=\mathbb{R}^d$ and $c(x,y) = \| x-y \|_p = (\sum_{j=1}^d |x_i-y_i|^p)^{1/p}$, Jensen's inequality and the triangle inequality imply 
\begin{talign} \label{eq:pwass_moment}
\max \big\{ \| \mathbb{E}[ X-Y] \|_p , | \mathbb{E}[ \| X \|_p^p]^{1/p}\! - \! \mathbb{E}[ \| Y \|_p^p]^{1/p} | \big\} 
\leq \mathbb{E}_{\gamma^*(P, Q)} [ \| X-Y \|_p^p ]^{1/p} = \mathcal{W}_p(P, Q)
\end{talign}
for any $P, Q \in \mathcal{P}_p(\mathcal{X})$ 
and random variables $(X,Y)$ jointly distributed according to a $p$-Wasserstein optimal coupling  $\gamma^*(P, Q)$.
Equation \eqref{eq:pwass_moment} shows that $p$-Wasserstein distances can control the 
difference between moments of order $p$. 
Indeed, \citet[Thm.~3.4, Rem.~3.5, and Prop.~3.6]{Huggins2020validatedAISTATS} showed that explicit bounds on Wasserstein distances translate into explicit guarantees for a variety of downstream inferential tasks including mean estimation, covariance estimation, numerical integration of Lipschitz functions, and prediction accuracy.
Meanwhile, these guarantees are \emph{not} implied by a small KL or $\alpha$-divergence \citep{Huggins2020validatedAISTATS}.

Popular approaches to estimating $\mathcal{W}_p(P, Q)$ involve drawing 
independent samples from $P$ and $Q$ and then computing the Wasserstein 
distance between the corresponding empirical distributions.
Such approaches produce estimates that are consistent
as the number of samples tend to infinity but can suffer from the curse of dimensionality
and give loose upper bounds of $\mathcal{W}_p(P, Q)$ when the number of 
samples does not increase exponentially with dimension 
\citep[e.g.,][]{weed2019sharpBERNOULLI}.
They also incurs prohibitive computational costs which scale at 
a cubic rate with the number of samples \citep{orlin1988fasterSTOC}.
Entropy-regularized variants of the Wasserstein distance 
such as Sinkhorn distances \citep{cuturi2013sinkhorn} offer computational costs which 
scale at a quadratic rate with the number of samples but produce estimates 
that are not consistent \citep{altschuler2017nearNEURIPS}.

This manuscript develops consistent upper bound estimates for Wasserstein distances. 
The developed algorithms and estimators are then used to assess the quality of approximate MCMC and certain variational inference methods. 
Specifically, we use couplings of Markov chains to estimate upper bounds on the Wasserstein distance between the limiting distribution of the asymptotically biased sampling method and the original target distribution of interest. 
As we cover in Sec.~\ref{subsection:comparison}, 
our work provides an appealing alternative to estimates based on empirical Wasserstein distances and Sinkhorn distances and to 
the upper bound estimates of \citet{Huggins2020validatedAISTATS}, which are based on
worst-case divergence bounds and rely on efficient importance sampling.
In addition, our upper bound estimates provably improve upon those of \citet{dobson2021usingSIAM}
which rely on challenging contraction-constant estimation.

In related work, measures of asymptotic bias based on Stein discrepancies  
have been developed, which do not require sampling from 
the target distribution of interest.
For example, \citet{gorham2019measuringAoAP} 
established a near-linear relationship between Stein discrepancies and standard Wasserstein distances, 
but the constants in these results rely on specific knowledge of the gradient of the log target density
that must be derived for each new target distribution. 
Our upper bound estimates of the Wasserstein distance apply to any distributions that can be targeted with Markov chains and do not require any additional distributional knowledge.

\subsection{Our contributions}
We introduce new tools for method developers to assess the quality of 
their approximate inference procedures. 
Our primary contributions are summarized below. 

In Sec.~\ref{section:sample_quality}, we first introduce algorithms for coupling two Markov chains with distinct stationary distributions. 
Our approach generalizes recent efforts to couple Markov chains with identical transition kernels 
\citep[see, e.g.,][]{glynn2014exact, heng2019unbiased, 
middleton2019unbiasedPMLR, jacob2020unbiasedJRSSB,biswas2019estimating,
biswas2021coupled}.
We then introduce estimators based on our coupled chains that consistently upper bound the Wasserstein distance between their stationary distributions. This enables us to assess the asymptotic bias of approximate MCMC methods and certain variational inference procedures. 

Sec.~\ref{section:properties} provides a theoretical analysis of our upper bound estimates.
We first establish the consistency and unbiasedness of our upper bound estimates
and then derive interpretable analytic upper bounds on our  estimates in terms of the mixing rate of one chain and the closeness of the two transition kernels.
These analytic bounds provide sufficient conditions for our upper estimates to be informative in high dimensions.

In Sec.~\ref{section:applications}, we demonstrate the favorable empirical performance of our upper bound estimates 
on modern applications. We first consider datasets with a large number of data points to assess 
the quality of stochastic gradient MCMC, variational Bayes, and Laplace approximations for Bayesian logistic regression. 
We then consider high-dimensional datasets to assess the quality of 
approximate MCMC for high-dimensional linear regression with continuous shrinkage priors 
($d \approx 50000$) and high-dimensional logistic regression with spike-and-slab priors
($d \approx 4500$). 
Finally, we discuss our results and directions for future work in Sec.~\ref{section:discuss}.
Open-source R code 
recreating all experiments in this paper can be found at
\url{github.com/niloyb/BoundWasserstein}.

\section{Bounding Wasserstein distance with couplings} \label{section:sample_quality}
Given distributions $P$ and $Q$ in $\mathcal{P}_p(\mathcal{X})$
for some $p \geq 1$, we wish to estimate upper bounds on $\mathcal{W}_p(P,Q)$. 
Our estimates are based on Markov chains $(X_t)_{t \geq 0}$ and $(Y_t)_{t \geq 0}$
with marginal transition kernels $K_1$ and $K_2$ invariant 
for $P$ and $Q$ respectively. 
Specifically, we construct a Markovian kernel $\bar{K}$ on the joint space $\mathcal{X} \times \mathcal{X}$ 
such that for all $x,y \in \mathcal{X}$, 
\begin{talign} \label{eq:joint_kernel}
\bar{K}\big( (x, y), (\cdot, \mathcal{X}) \big) = K_1(x, \cdot) \text{ and } \bar{K}\big( (x, y), (\mathcal{X}, \cdot) \big) = K_2(y, \cdot).
\end{talign}
Given the kernel $\bar{K}$, we generate a coupled Markov chain $(X_t, Y_t)_{t \geq 0}$ using 
Alg.~\ref{algo:coupled_chain_general}, 
a generalization of 
existing coupling constructions 
\citep{johnson1998coupling, glynn2014exact, heng2019unbiased, middleton2019unbiasedPMLR, 
jacob2020unbiasedJRSSB,biswas2019estimating, biswas2021coupled}.
While prior work focused on $K_1=K_2$ and $X_t \eqd Y_t$ to establish convergence to a single stationary distribution $P$, our work uses distinct kernels $K_1$ and $K_2$ to bound the distance between distinct stationary distributions $P$ and $Q$.
Algorithms to sample from $\bar{K}$ are covered in Sec.~\ref{subsection:algo_k_bar}. 

\begin{algorithm}[!htb]
\small
\DontPrintSemicolon
\KwIn{Initial distribution $\bar{I}_0$ on $\mathcal{X} \times \mathcal{X}$, joint kernel $\bar{K}$, number of iterations $T$} 
\textbf{Initialize:} Sample $(X_0, Y_0) \sim \bar{I}_0$ \;
\lFor{$t=1,\mydots,T-1$}{
    Sample $(X_{t+1}, Y_{t+1}) | (X_{t}, Y_{t}) \sim \bar{K}\big( (X_{t}, Y_{t}), \cdot \big)$
    }
 \Return Markov chain $(X_{t}, Y_{t})_{t = 0}^T$
 \caption{Coupled Markov chain Monte Carlo for bounding Wasserstein distances}
 \label{algo:coupled_chain_general}
\end{algorithm}

For a Markov chain $(X_{t}, Y_{t})_{t \geq 0}$ from 
Alg.~\ref{algo:coupled_chain_general}, 
suppose the marginal distributions of $X_{t}$ and $Y_{t}$ converge in $p$-Wasserstein distance
to $P$ and $Q$ respectively as $t$ tends to infinity. Informally, the coupling representation of the 
Wasserstein distance implies
$\mathcal{W}_p(P, Q)^p \leq  \underset{S \rightarrow \infty, T-S\rightarrow\infty}{\liminf} \sum_{t=S+1}^T \frac{\mathbb{E}[c(X_t, Y_t)^p]}{T-S} $.
This motivates our coupling upper bound (\cubname) estimate
\begin{talign} \label{eq:W_ublimit1}
\cub\defeq  ( \frac{1}{I(T-S)} \sum_{i=1}^I \sum_{t=S+1}^T c(X^{(i)}_t, Y^{(i)}_t)^p )^{1/p},
\end{talign}
where $(X^{(i)}_{t}, Y^{(i)}_{t})_{t = 0}^T$ are 
sampled using Alg.~\ref{algo:coupled_chain_general}
independently for each $i$, with burn-in $S \geq 0$ and trajectory length $T>S$.
We prove the consistency of this and related upper bound estimators in Sec.~\ref{section:properties}. 
We now consider the empirical performance of this estimator on two stylized 
examples, working with the Euclidean metric $c(x,y)=\|x-y\|_2$ on $\mathbb{R}^d$.

\subsection{Upper bound on Wasserstein distance} \label{subsection:wass_ub}
We consider the performance of $\cub[2]$ \eqref{eq:W_ublimit1} 
for two Gaussian distributions on $\mathbb{R}^d$, given by
\begin{talign} \label{eq:stylized_example_targets}
P = \mathcal{N}(0, \Sigma) \text{ where } \Sigma_{i,j} = 0.5^{|i-j|} \text{ for } 1 \leq i,j \leq d \text{ and } Q = \mathcal{N}(0, I_d).
\end{talign}
Here we use the marginal kernels $K_1$ and $K_2$ of the Metropolis--adjusted Langevin algorithm (MALA)
with step sizes $\sigma_P=\sigma_Q=0.5 d^{-1/6}$ targeting $P$ and $Q$ respectively,
following existing guidance for step size choice \citep{roberts1998optimalJRSSB}.
The joint kernel $\bar{K}$ is based on a common random numbers 
(CRN, also called ``synchronous'') coupling of both the proposal step and the accept-reject step of the MALA algorithm, as detailed in Alg.~\ref{algo:mala_crn} of App.~\ref{appendices:algos}. 
Each chain is initialized with independent draws of 
$X^{(i)}_0 \sim P$ and $Y^{(i)}_0 \sim Q$, and the choice of initialization 
is covered in Sec.~\ref{section:properties}.
Throughout, we will also compare to an \emph{independent coupling} obtained by sampling the $(X_t)_{t\geq 0}$ and $(Y_t)_{t\geq 0}$ chains independently using the $K_1$ and $K_2$ kernels respectively.

Fig.~\ref{fig:stylized_example_mvn_combined} (Left) compares several upper bound estimates of $\mathcal{W}_2(P,Q)$ for dimension $d=100$. 
The solid line 
(\raisebox{-3pt}{\tikz{
\fill[lightgray!50] (0,0) rectangle (0.6,0.45);
\draw[black,solid,line width=2pt](0,0.225) -- (0.6,0.225)
}}) 
is $\cub[2]$ 
based on $I=5$ independent chains, burn-in $S=0$ and 
varying trajectory length $1\leq T \leq 1000$, and
the grey error bands represent 95\% confidence intervals 
arising from Monte Carlo error. 
As the marginal chains are initialized at their respective
stationary distributions, here $\cub[2]$ produces valid upper bounds for all 
trajectory lengths $T$ with zero burn-in $S=0$. The values of $T$ and $I$
are chosen based on upper bound estimates and error bands of initial runs, 
and this choice is further discussed in Sec.~\ref{subsection:algo_k_bar}.
The dotted line
(\raisebox{-3pt}{\tikz{
\fill[lightgray!50] (0,0) rectangle (0.6,0.45);
\draw[black,dotted,line width=2pt](0,0.225) -- (0.6,0.225)
}})
plots the independent coupling upper bound
$\underset{Y \sim Q,X \sim P}{\mathbb{E}}[\|X-Y\|^2_2]^{1/2}=(2d)^{1/2}$ with $X$ and $Y$ independent.
The dot-dashed line 
(\raisebox{-3pt}{\tikz{
\fill[lightgray!50] (0,0) rectangle (0.6,0.45);
\draw[black,dash pattern={on 6pt off 3pt on 2pt off 3pt},line width=2pt](0,0.225) -- (0.6,0.225)
}})
plots an estimate based on empirical Wasserstein distances, 
given by $\sum_{i=1}^I \mathcal{W}_2(\hat{P}^{(i)}_T,\hat{Q}^{(i)}_T)/ I$ where each $\hat{P}^{(i)}_T$ and $\hat{Q}^{(i)}_T$ are the empirical distributions of 
$T=1000$ points sampled independently from $P$ and $Q$ respectively and 
$\mathcal{W}_2(\hat{P}^{(i)}_T,\hat{Q}^{(i)}_T)$ is calculated exactly by solving a linear program \citep[][see also App.~\ref{appendices:empirical_wass}]{orlin1988fasterSTOC}. In Sec.~\ref{subsection:comparison} we examine the upper- and lower-bounding properties of this common Wasserstein distance estimate and observe that its convergence can be slow in high dimensions due to substantial bias. 
Finally, the dashed line 
(\raisebox{-3pt}{\tikz{
\fill[lightgray!50] (0,0) rectangle (0.6,0.45);
\draw[black,dashed,line width=2pt](0,0.225) -- (0.6,0.225)
}})
shows the true Wasserstein distance $\mathcal{W}_2(P,Q)$, which is known for this stylized example \citep[see, e.g.,][Rem.~2.23]{Peyre2019computationalFTML} and is
given by the coupling $\underset{Y \sim Q,X=\Sigma^{1/2}Y \sim P}{\mathbb{E}}[\|X-Y\|^2_2]^{1/2}$ where $\Sigma^{1/2}$ is the positive matrix square root of $\Sigma$. 
At initialization ($T=0$) $\cub[2]$ matches the equivalent 
independent coupling bound. For greater trajectory lengths $T$, 
$\cub[2]$ offers a significant improvement over  the independent bound 
and the popular empirical Wasserstein estimate.

Fig.~\ref{fig:stylized_example_mvn_combined} (right) considers  $\mathcal{W}_2(P,Q)$ for higher dimensions.
The solid line now plots $\cub[2]$
based on $I=5$, $S=0$, and $T=1000$. 
Fig.~\ref{fig:stylized_example_mvn_combined} (right) highlights that, unlike the independent and empirical Wasserstein estimates, $\cub[2]$ 
offers bounds that remain informative even in higher dimensions. Such 
dimension-free properties of our upper bounds are 
investigated in Sec.~\ref{section:properties}. 
Sec.~\ref{subsection:comparison} provides a further comparison of 
our \cubname bounds with empirical Wasserstein and Sinkhorn distances, which can have prohibitive computational cost for larger sample 
sizes and suffer from the curse of dimensionality. 

\begin{figure}[!t]
\captionsetup[subfigure]{font=footnotesize,labelfont=footnotesize}
    \centering
    \includegraphics[width=\textwidth]{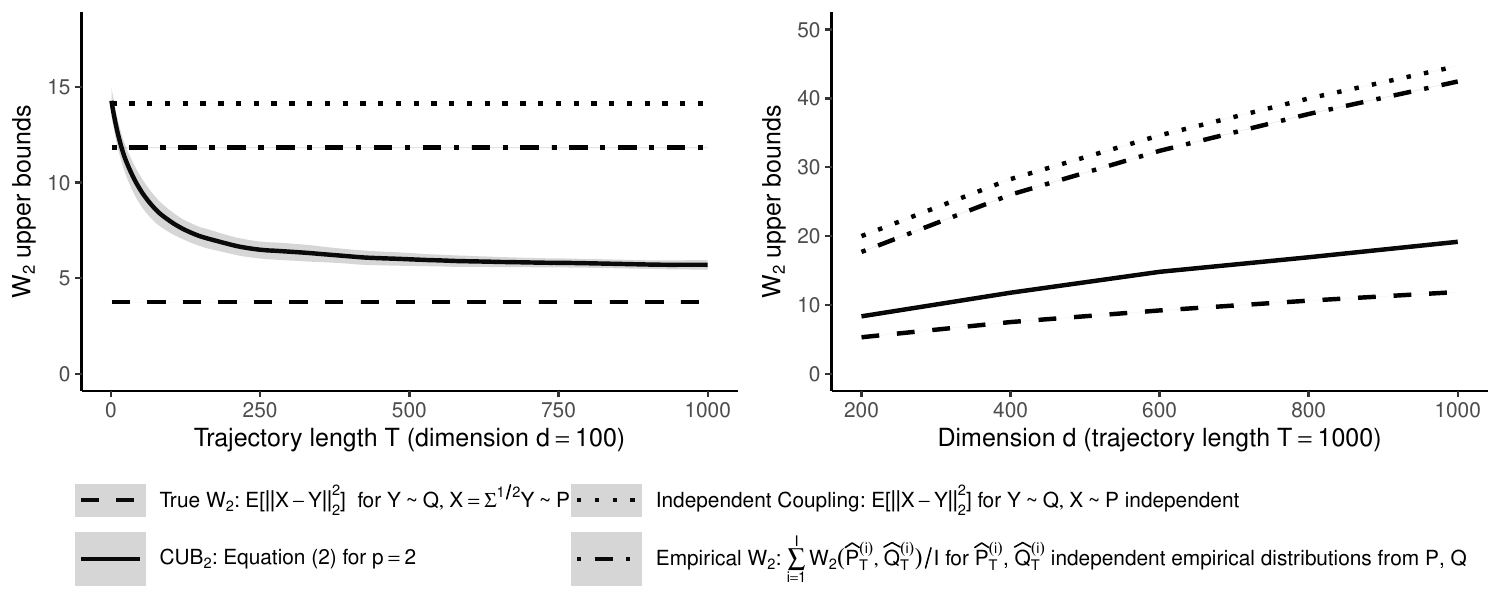}
        \caption{Upper bound estimates for $\mathcal{W}_2(P,Q)$ with $P = \mathcal{N}(0, \Sigma) \text{ where } \Sigma_{i,j} = 0.5^{|i-j|} \text{ for } 1 \leq i,j \leq d$, $Q= \mathcal{N}(0, I_d)$, and metric $c(x,y) = \twonorm{x-y}$. See Sec.~\ref{subsection:wass_ub}. 
    }
    \label{fig:stylized_example_mvn_combined}
\end{figure}

\subsection{Bias of approximate MCMC methods} \label{subsection:bias_ub}
The unadjusted Langevin algorithm (ULA)
is a popular 
approximate MCMC counterpart to MALA. It has the same proposal step as 
MALA but now all proposed states are accepted. 
The lack of a Metropolis--Hastings accept-reject step 
leads to ULA having a lower computation 
costs per iteration than MALA, which is beneficial for
applications with large datasets \cite[e.g.,][]{nemeth2021stochasticJASA}. 
On the other hand, ULA is asymptotically biased
\citep{durmus2019highBERNOULLI}. 
In this section, we consider upper bounds 
of the Wasserstein distance between the limiting distribution of
ULA and the original target distribution of interest on a stylized example. 

\begin{figure}[!htb]
\captionsetup[subfigure]{font=footnotesize,labelfont=footnotesize}
    \centering
    \begin{subfigure}[b]{0.65\textwidth}
        \includegraphics[width=\textwidth]{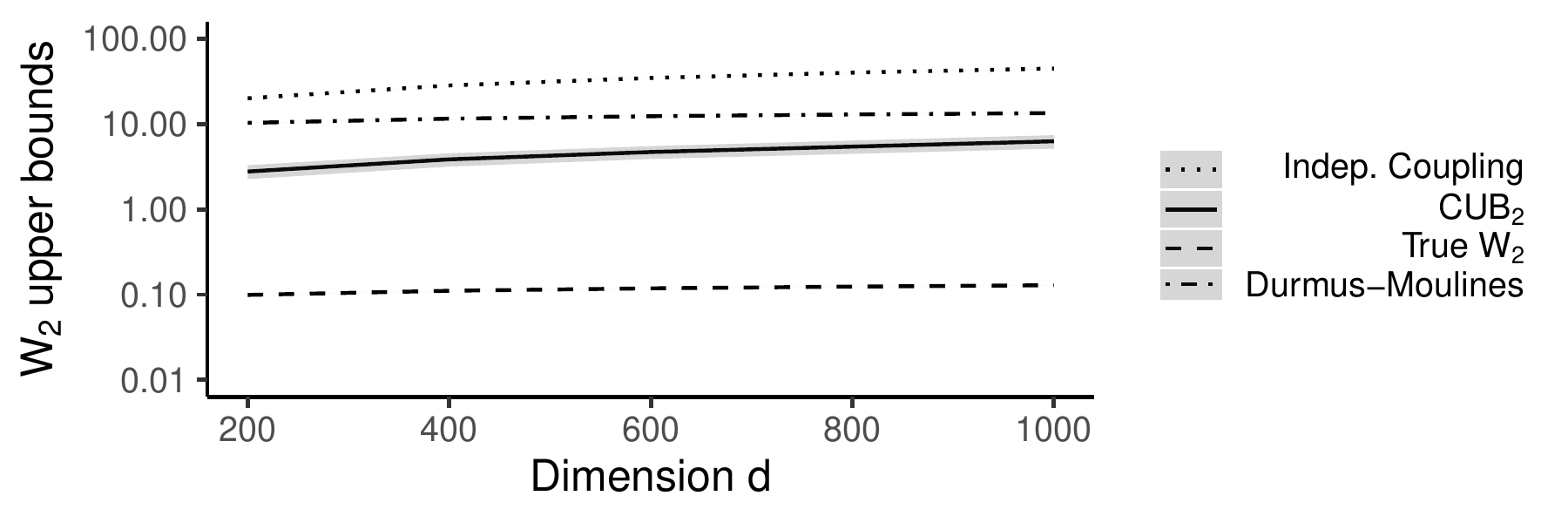}
    \end{subfigure}
    \caption{Upper bound estimates for the $2$-Wasserstein distance with $c(x,y) = \twonorm{x-y}$ between the 
    limiting distributions of ULA and MALA targeting $P = \mathcal{N}(0, \Sigma) \text{ where } \Sigma_{i,j} = 0.5^{|i-j|} \text{ for } 1 \leq i,j \leq d$
    on $\mathbb{R}^d$. See Sec.~\ref{subsection:bias_ub}. }
    \label{fig:stylized_example_ula_mala}
\end{figure}
Fig.~\ref{fig:stylized_example_ula_mala} shows the performance of $\cub[2]$ \eqref{eq:W_ublimit1} when the marginal kernels 
$K_1$ and $K_2$ are based, respectively, on the MALA and ULA 
Markov chains targeting the distribution $\mathcal{N}(0, \Sigma)$ on $\mathbb{R}^d$ defined in
\eqref{eq:stylized_example_targets}. 
The MALA kernel $K_1$ produces an \textit{exact} Markov chain which
is $\mathcal{N}(0, \Sigma)$ invariant, and the ULA kernel $K_2$ produces 
an \textit{approximate} Markov chain which is not $\mathcal{N}(0, \Sigma)$
invariant. The joint kernel $\bar{K}$ is based 
on a CRN coupling of the proposal steps of MALA and ULA, and
is given in Alg.~\ref{algo:mala_ula_crn} of App.~\ref{appendices:algos}. 
We again use a step size of $\sigma_P=\sigma_Q=0.5d^{-1/6}$ for both  marginal chains (following existing guidance for step size choice \citep{roberts1998optimalJRSSB})
and initialize $X^{(i)}_0 \sim \mathcal{N}(0, I_d)$ and $Y^{(i)}_0 \sim \mathcal{N}(0, I_d)$ 
independently for each coupled chain $i$. Let $P_t$ and $Q_t$ denote the marginal distribution of
 $X^{(i)}_t$ and $Y^{(i)}_t$ respectively. 
We show in App.~\ref{appendices:stylized_ula_mala} that 
$P_t \overset{t \rightarrow \infty}{\Rightarrow} P \defeq \mathcal{N}(0, \Sigma)$, $Q_t = \mathcal{N} \big(0, \sigma_Q^2 \sum_{j=0}^{t-1} B^{2j}\big)$,
and $Q_t \overset{t \rightarrow \infty}{\Rightarrow} Q \defeq \mathcal{N}(0, \sigma_Q^2 (I_d - B^2)^{-1})$,
where $B \defeq ( I_d - (\sigma_Q^2/2) \Sigma^{-1} )$ and the weak convergence of $Q_t$ to $Q$ holds for $\sigma_Q$ 
sufficiently small.

Fig.~\ref{fig:stylized_example_ula_mala} compares several approaches to bounding the asymptotic $\mathcal{W}_2(P,Q)$ bias of ULA. 
The solid line (\raisebox{-3pt}{\tikz{
\fill[lightgray!50] (0,0) rectangle (0.6,0.45);
\draw[black,solid,line width=2pt](0,0.225) -- (0.6,0.225)
}}) displays our coupling upper bound estimate.
For each dimension $d$, it is calculated using $\cub[2]$
\eqref{eq:W_ublimit1} with 
$I=10$, $S=1000$, and $T=3000$. 
The dashed line 
(\raisebox{-3pt}{\tikz{
\fill[lightgray!50] (0,0) rectangle (0.6,0.45);
\draw[black,dashed,line width=2pt](0,0.225) -- (0.6,0.225)
}})
shows the true asymptotic bias $\mathcal{W}_2(P,Q)$ and 
the dotted 
(\raisebox{-3pt}{\tikz{
\fill[lightgray!50] (0,0) rectangle (0.6,0.45);
\draw[black,dotted,line width=2pt](0,0.225) -- (0.6,0.225)
}})
line shows the independent coupling upper bound, both 
of which can be computed exactly in this example.  
The dot-dashed line 
(\raisebox{-3pt}{\tikz{
\fill[lightgray!50] (0,0) rectangle (0.6,0.45);
\draw[black,dash pattern={on 2pt off 3pt on 7pt off 3pt},line width=2pt](0,0.225) -- (0.6,0.225)
}})
plots the analytic ULA bias upper bounds of \citet[Cor.~9]{durmus2019highBERNOULLI} 
(see App.~\ref{appendices:stylized_ula_mala} for more details). 
The tailored Durmus-Moulines bounds are significantly tighter than the convenient independent coupling bound, but $\cub[2]$ is tighter still, offering significantly improved estimates for all dimensions.

\section{Properties and Implementation} \label{section:properties}
In this section we establish the consistency of the estimators in Sec.~\ref{section:sample_quality},
describe how to sample from the joint kernel $\bar{K}$ in Alg.~\ref{algo:coupled_chain_general},
investigate the theoretical properties of our upper bounds, and compare to alternative approaches. All
proofs are in App.~\ref{appendices:proofs}.

\subsection{Consistency of coupling upper bounds} \label{subsection:consistency}
We begin by establishing the consistency of coupling upper bound estimators.
Our first result bounds the Wasserstein distance between coupled chains in terms of an instantaneous CUB estimator related to the time-averaged estimator in \eqref{eq:W_ublimit1}.

\begin{Prop}[Consistency of instantaneous \cubname]
\label{prop:Wp_UB_t}
Let $(X^{(i)}_t, Y^{(i)}_t)_{t \geq 0}$ for $i=1,\mydots,I$  denote coupled chains generated independently from  Algorithm 
\ref{algo:coupled_chain_general} 
with marginal distributions $X^{(i)}_t\sim P_t$ and $Y^{(i)}_t\sim Q_t$ at time $t$.
For each $t\geq 0$, define the \emph{instantaneous \cubname} estimator 
\begin{talign} \label{eq:W_ublimit2}
\cub[p,t] \defeq \Big( \frac{1}{I} \sum_{i=1}^I c(X^{(i)}_t, Y^{(i)}_t)^p \Big)^{1/p}.
\end{talign}
If $P_s$ and $Q_s$ have finite moments of order $p$ for all $s \leq t$, then $\cub[p,t]$ has finite moments of order $p$, and, as $I\to\infty$,
\begin{talign} \label{eq:Wp_UB_t}
\cub[p,t]^p \toasL \mathbb{E}[ \cub[p,t]^p ] \geq \mathcal{W}_p(P_t, Q_t)^p.
\end{talign}
\end{Prop}

Our next result shows that the estimator $\cub[p]$ \eqref{eq:W_ublimit1} consistently bounds the Wasserstein distance between time-averaged marginal distributions. 
\begin{Cor}[Consistency of CUB for time-averaged marginals]
\label{cor:Wp_UB_average_t}
Under the assumptions and notation of Prop.~\ref{prop:Wp_UB_t}, 
consider the estimator $\cub[p]$ \eqref{eq:W_ublimit1} with 
any number of independent chains $I \geq 0$, and trajectories with 
burn-in $S \geq 1$ and length $T \geq S$. 
Then $\cub[p]$ has finite moments of order $p$, and as $I\to\infty$,
\begin{talign}\cub[p]^p \toasL \mathbb{E}[ \cub[p]^p ] \geq \mathcal{W}_p ( \frac{1}{T-S} \sum_{t=S+1}^T P_t, \frac{1}{T-S} \sum_{t=S+1}^T Q_t )^p.
\end{talign}
\end{Cor}

An important implication of Cor.~\ref{cor:Wp_UB_average_t} is that $\cub[p]$ \eqref{eq:W_ublimit1} consistently bounds the Wasserstein distance between stationary distributions whenever its chains are marginally initialized at stationarity.

\begin{Cor}[Consistency of \cubname with stationary initialization]
\label{cor:Wp_UB}
Under the assumptions and notation of Prop.~\ref{prop:Wp_UB_t}, 
suppose kernels $K_1$ and $K_2$ have stationary distributions $P$ and $Q$ respectively,
where $P$ and $Q$ have finite moments of order $p$. Suppose we initialize 
$(X_0, Y_0) \sim \bar{I}_0$ such that $X_0 \sim P$ and $Y_0 \sim Q$ marginally.
Then for any number of independent chains 
$I \geq 0$, trajectories with burn-in $S \geq 1$ and length $T \geq S$, 
the estimator $\cub$ \eqref{eq:W_ublimit1} 
has finite moments of order $p$, and as $I\to\infty$,
\begin{talign} \label{eq:Wp_UB}
\cub[p]^p \ \toasL \mathbb{E}[ \cub^p ] \geq \mathcal{W}_p(P, Q)^p.
\end{talign}
\end{Cor}

We may not always be able to initialize using the marginal stationary distributions $P$ and $Q$. 
To obtain upper bounds on $\mathcal{W}_p(P,Q)$ without starting at the marginal stationary distributions
$P$ and $Q$, we make an assumption related to convergence of the Markov chain marginals
$(P_t)_{t \geq 0}$ and $(Q_t)_{t \geq 0}$.

\begin{Asm}[Convergence of marginal chains]
\label{asm:marginal_conv_Wp} 
As $t \rightarrow \infty$, $P_t$ and $Q_t$ converge in $p$-Wasserstein distance respectively to $P$ and $Q$ with finite moments of order $p$.
\end{Asm}

\begin{Prop}[Consistency when chain marginals converge] 
\label{prop:Wp_UB_marginal_conv}
Under Assump.~\ref{asm:marginal_conv_Wp} and the assumptions and notation of Prop.~\ref{prop:Wp_UB_t}, for all $\epsilon>0$ 
there exists $S \geq 1$ such that for all $T \geq S$, 
the estimator $\cub$ \eqref{eq:W_ublimit1} 
has finite moments of order $p$, and as $I\to\infty$,
\begin{talign} \label{eq:w_ub_burnin}
\cub[p]^p \ \toasL \mathbb{E} \big[ \cub[p]^p \big] \geq \mathcal{W}_p(P, Q)^p - \epsilon.
\end{talign}
\end{Prop}
Prop.~\ref{prop:Wp_UB_marginal_conv} establishes that $\cub[p]$ with any initialization $(X_0, Y_0) \sim \bar{I}_0$ consistently bounds $\mathcal{W}_p(P, Q)$ as $I$ and $S$ grow.  
In practice, we can use standard MCMC burn-in diagnostics to select an appropriate burn-in level for our marginal chains of interest
\cite[e.g.,][]{johnson1998coupling, biswas2019estimating, vats2020revisitingSS, vehtari2020rankBA}.
Alternatively, for $p=1$, we can avoid burn-in removal and instead directly correct 
our bound for non-stationarity using the recent $L$-lag coupling approach of \citet{biswas2019estimating}
(see App.~\ref{appendices:llag} for details).

We emphasize that the results of this section hold for any coupled chain sampled 
using Alg.~\ref{algo:coupled_chain_general} with joint kernel $\bar{K}$ satisfying
\eqref{eq:joint_kernel}. For example, this includes both the CRN coupled chains 
and the independently coupled chains from Sec.~\ref{section:sample_quality}, where 
the CRN coupled chains produced more informative upper bounds empirically as shown 
in Figures \ref{fig:stylized_example_mvn_combined} and \ref{fig:stylized_example_ula_mala}. 
We now consider how to sample from $\bar{K}$ and investigate when 
our upper bounds are informative.

\subsection{Algorithms to sample from the coupled kernel $\bar{K}$} \label{subsection:algo_k_bar}
In this section, we develop algorithms to sample from the joint kernel $\bar{K}$ such that the 
estimators from Sec.~\ref{subsection:consistency} can produce informative upper bounds. 
Our construction 
decomposes the overall coupling into two convenient coupling steps based on 
a same-chain coupling kernel $\Gamma_1$ on $\mathcal{X} \times \mathcal{X}$ and 
a perturbative coupling kernel $\Gamma_\Delta$ on $\mathcal{X}$:
\begin{enumerate}
\item $\Gamma_1$ is a Markovian coupling of the kernel $K_1$ with itself: 
for all $x,\tilde{x} \in \mathcal{X}$, $\Gamma_1(x,\tilde{x})$ is a coupling of the distributions $K_1(x,\cdot)$ and $K_1(\tilde{x},\cdot)$. 
\item $\Gamma_\Delta$ is coupling of kernels $K_1$ and $K_2$ from the same point:
for all $z \in \mathcal{X}$, $\Gamma_\Delta(z)$ is a coupling of the distributions $K_1(z,\cdot)$ and 
$K_2(z,\cdot)$. 
\end{enumerate}
This decomposition allows us to exploit the extensive and growing literature on same-chain coupling kernels and their properties (see Section~\ref{subsection:theory}) and to analyze the targeting of two distinct stationary distributions as a simple perturbation to well-studied same-chain couplings.
For  example, when $K_1$ is a Metropolis--Hastings kernel, $\Gamma_1$ can be a CRN coupling of both 
the proposal step and the accept-reject step.
Indeed, we often make use of CRN couplings as a default choice in this work due to their broad applicability and straightforward implementation. 
When the Metropolis--Hastings proposal is based on a spherically symmetric distribution such as
a Gaussian---as in random walk Metropolis--Hastings or the momentum component in 
Hamiltonian Monte Carlo (HMC)---$\Gamma_1$ can be a reflection coupling of the 
proposal step and a CRN coupling of 
the accept-reject step 
\citep[e.g.][]{bou-rabee2020couplingAOAP, oleary2021maximal}.
The kernel $\Gamma_\Delta$ characterizes the perturbation
between the marginal kernels $K_1$ and $K_2$. 
For example, when $K_1$ and $K_2$ are MALA and ULA kernels respectively, 
$\Gamma_\Delta$ can be a CRN coupling of the proposal step. This leads to identical proposals when MALA and 
ULA have the same step size, but the MALA chain will have a further accept-reject step while the ULA chain will always accept the proposal. 
We discuss the choice of $\Gamma_1$ and $\Gamma_\Delta$ further in Sec.~\ref{subsection:theory}. 
Given $\Gamma_1$ and $\Gamma_\Delta$, we sample from the joint kernel $\bar{K}$
using Alg.~\ref{algo:coupled_kernel_single_step}. 
\begin{algorithm}[!h]
\small
\DontPrintSemicolon
\KwIn{Chain states $X_{t-1}$ and $Y_{t-1}$, kernels $K_1$ and $K_2$, coupled kernels $\Gamma_1$ and $\Gamma_\Delta$} 
Sample $(X_t, Z_t, Y_t) | X_{t-1}, Y_{t-1}$ such that  
$(X_t, Z_t) \sim \Gamma_1(X_{t-1}, Y_{t-1})$, $(Z_t,Y_t) \sim \Gamma_\Delta(Y_{t-1})$\;
\Return $(X_{t}, Y_{t})$
\caption{Joint kernel $\bar{K}$ 
which couples the marginal kernels $K_1$ and $K_2$}
\label{algo:coupled_kernel_single_step}
\end{algorithm}

Alg.~\ref{algo:coupled_kernel_single_step} gives the conditional marginal distributions 
$X_t | X_{t-1}, Y_{t-1} \sim K_1(X_{t-1}, \cdot)$, $Z_t | X_{t-1}, Y_{t-1} \sim K_1(Y_{t-1}, \cdot)$,
$Y_t | X_{t-1}, Y_{t-1} \sim K_2(Y_{t-1})$ 
so that $\bar{K}$ satisfies \eqref{eq:joint_kernel}. 
Often Alg.~\ref{algo:coupled_kernel_single_step} can be implemented 
without explicitly sampling $Z_t$.
As an example, consider 
when $K_1$ and $K_2$ are MALA and ULA 
kernels with step sizes $\sigma_P$ and $\sigma_Q$, target distributions $P$ and $Q$, and $\Gamma_1$ and $\Gamma_\Delta$ are CRN coupled kernels. Given 
$(X_{t-1}, Y_{t-1})$, we sample $\epsilon_{CRN} \sim \mathcal{N}(0, I_d)$ 
and calculate the proposals 
$X^* = X_{t-1} + (\sigma^2_P/2) \nabla \log P(X_{t-1}) + \sigma_P \epsilon_{CRN}$, 
$Z^* = Y_{t-1} + (\sigma^2_P/2) \nabla \log P(Y_{t-1}) + \sigma_P \epsilon_{CRN}$, and
$Y^* = Y_{t-1} + (\sigma^2_Q/2) \nabla \log Q(Y_{t-1}) + \sigma_Q \epsilon_{CRN}$.
Then we accept or reject proposals $X^*$ and $Z^*$ based on a Metropolis--Hastings
correction with a common random number $U_{CRN} \sim \Uniform(0,1)$ to obtain 
$X_t$ equal to $X^*$ or $X_{t-1}$, $Z_t$ equal to $Z^*$ or $Y_{t-1}$, and always 
accept $Y^*$ to obtain $Y_t=Y^*$. 
Notably, $Z_t$ need not be explicitly sampled to perform this update of $(X_t, Y_t)$.
This CRN coupling of MALA and ULA is included in Alg.~\ref{algo:mala_ula_crn} of 
App.~\ref{appendices:algos}. 
App.~\ref{appendices:algos} also details general CRN and reflection
couplings between two Metropolis--Hastings kernels. 

We now cover implementation practicalities and potential limitations.
\begin{figure}[!t]
\captionsetup[subfigure]{font=footnotesize, labelfont=footnotesize}
    \centering
    \begin{subfigure}[b]{0.32\textwidth}
        \includegraphics[width=\textwidth]{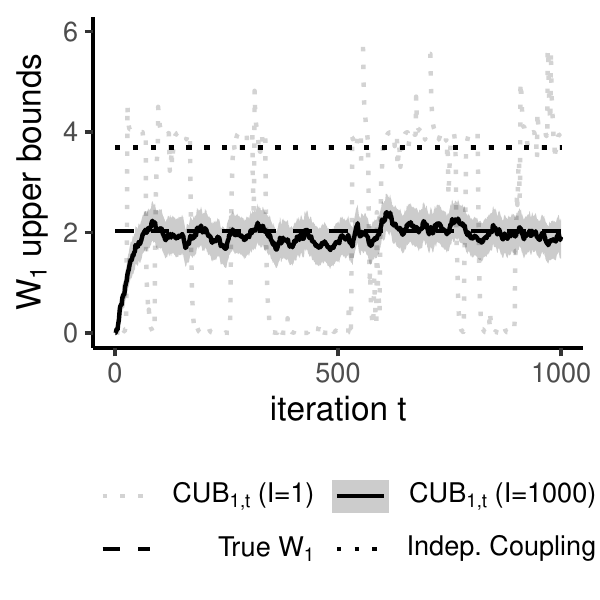}
        \caption{Impact of multiple trajectories on CRN coupling bounds.}
        \label{fig:bimodal1a}
    \end{subfigure}
    \begin{subfigure}[b]{0.32\textwidth}
        \includegraphics[width=\textwidth]{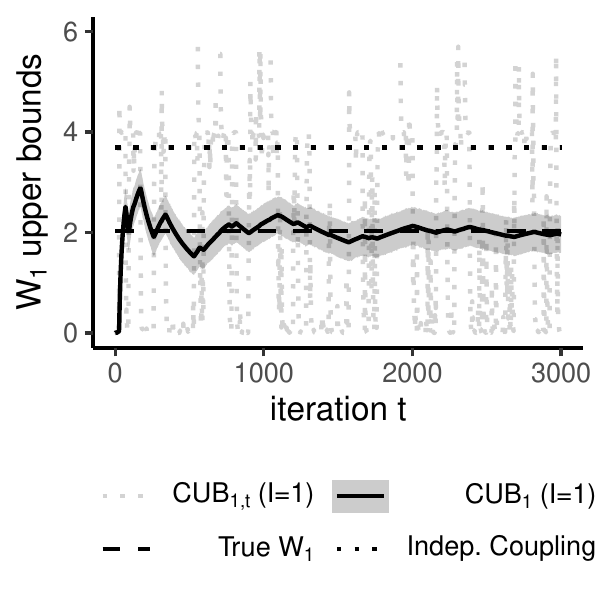}
        \caption{Impact of ergodic averaging on CRN coupling bounds.}
        \label{fig:bimodal1b}
    \end{subfigure}
    ~ %
    \begin{subfigure}[b]{0.32\textwidth}
        \includegraphics[width=\textwidth]{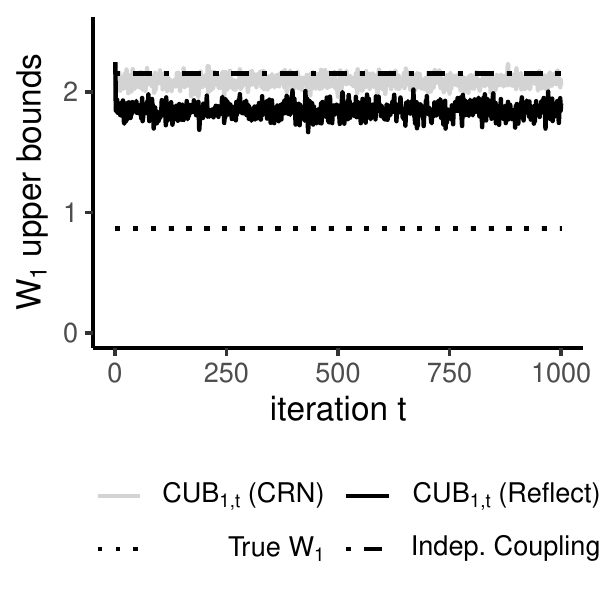} 
        \caption{Impact of coupling choice on bound quality.}
        \label{fig:crn_reflection}
    \end{subfigure}
    \caption{
    Impact of multiple trajectories, ergodic averaging, and coupling choice on coupling bound quality for the $1$-Wasserstein distance with $c(x,y)=\|x-y\|_2$. See Sec.~\ref{subsection:algo_k_bar}.
    }
    \label{fig:implementation}
\end{figure}

\newcommand{\para}[1]{\textbf{#1}\ \ }
\para{Number of coupled chains and chain length to simulate.} We first highlight the value of averaging over time and over independent coupled chains when producing  upper bound estimates. 
Figures \ref{fig:bimodal1a} and \ref{fig:bimodal1b} examine the performance of the $\cub[1]$ \eqref{eq:W_ublimit1} and instantaneous $\cub[1,t]$ \eqref{eq:W_ublimit2} estimators  when bounding the $1$-Wasserstein distance with $c(x,y)=\|x-y\|_2$ between 
$%
P = \frac{1}{2}\mathcal{N}(1_d, I_d) + \frac{1}{2}\mathcal{N}(-1_d, I_d) 
\text{ and } 
Q = \mathcal{N}(1_d, I_d)
\text{ with }
d=4
$ %
so that one of the marginal target distributions is bimodal with well-separated modes. 
We simulate the coupled chains $(X^{(i)}_t, Y^{(i)}_t)_{t \geq 0}$ independently for each $i$ using 
Alg.~\ref{algo:coupled_chain_general}, where the joint kernel $\bar{K}$ is based on a CRN coupling 
of MALA kernels $K_1$ and $K_2$ targeting distributions $P$ and $Q$ respectively. The MALA kernels
have a common step size $d^{-1/6}$ (following existing guidance for step size choice \citep{roberts1998optimalJRSSB}), and we initialize $X^{(i)}_0 = \mathrm{1}_d$ and 
$Y^{(i)}_0 = \mathrm{1}_d$ such that both marginal chains start at the common mode. 
Fig.~\ref{fig:bimodal1a} isolates the impact of averaging over multiple chains when computing the  $\cub[1,t]$ estimate \eqref{eq:W_ublimit2}. 
The grey dotted line shows the single trajectory $(c(X_t^{(1)}, Y_t^{(1)}))_{t=1}^{1000}$ and
the black solid line shows the averaged trajectory $(\bar{c}(X_t, Y_t))_{t=1}^{1000}$
where $\bar{c}(X_t, Y_t) \defeq \sum_{i=1}^{I} c(X_t^{(i)}, Y_t^{(i)})/I$ for $I=100$ independent chains.
The grey dotted line alternates between values close to $0$ or $4$, corresponding to 
when the marginal chains from a single trajectory are both near the common mode ($1_d$) or 
near different modes ($-\mathrm{1}_d$ and $\mathrm{1}_d$) respectively. This illustrates that instantaneous 
upper bound estimator $\cub[1,t]$  \eqref{eq:W_ublimit2} based on only a single trajectory of short chain length can have high variance. 
For multiple independent coupled chains, the averaged trajectory 
has lower variance and higher precision as shown by the grey confidence bands and the black solid line which 
remains close to the true $\mathcal{W}_1(P,Q)$ distance (shown by black dotted line). 
Conveniently, these multiple chains can be simulated in 
parallel.  
Also even for upper bound estimates
based on a single chain, the $\cub[1]$ estimator with $I=1$ 
and a sufficiently large chain length $T$ can produce estimates with low variance, as shown 
by the grey confidence bands and the black solid line in Fig.~\ref{fig:bimodal1b}. 
The optimal choice between number of independent coupled chains and chain length, 
given a certain coupled kernel $\bar{K}$ and a fixed number of parallel processors is an open 
area for further investigation. 
\citet{jacob2020unbiasedJRSSB} contains related motivating discussions for unbiased estimation
with couplings. 

\para{Choice of coupled kernel.}
Secondly, we highlight the importance of the choice of the coupled kernel $\bar{K}$. 
Fig.~\ref{fig:crn_reflection} examines the performance of the $\cub[1]$ \eqref{eq:W_ublimit1} 
estimator when bounding the $1$-Wasserstein distance with $c(x,y)=\|x-y\|_2$ between 
$%
P = \frac{1}{2}\mathcal{N}(2, 1) + \frac{1}{2}\mathcal{N}(-2, 1) 
\text{ and } 
Q = \frac{1}{2}\mathcal{N}(1, 1) + \frac{1}{2}\mathcal{N}(-1, 1),
$ %
so that now both the marginal target distributions are bimodal. Under this setup, 
we simulated coupled chains based on both a CRN coupling and a reflection coupling
of MALA kernels $K_1$ and $K_2$ targeting distributions $P$ and $Q$ respectively. 
The MALA kernels have a common step size $2$, and we initialize such that each 
$X^{(i)}_0 \sim P$ and $Y^{(i)}_0 \sim Q$ are independent. 
In Fig.~\ref{fig:crn_reflection}, the grey and black solid lines show averaged 
trajectories from $I = 1000$ independent coupled chains based on CRN and reflection coupling
respectively. It highlights that reflection coupling gives tighter upper bounds compared to 
CRN for this example. In general, the choice of coupling can have an impact on the 
tightness of our upper bounds. We emphasize that any choice of such couplings still 
produces consistent upper bounds (as shown in Sec.~\ref{subsection:consistency}). In 
practice, one can simulate different coupling algorithms to empirically assess which 
choice produces the tightest upper bounds and even select the smallest of multiple coupling bounds. 
Finally, Fig.~\ref{fig:crn_reflection} highlights that our upper bounds may not always be
very close to the true Wasserstein distance
when the marginal Markov chains have slow mixing rates or 
when the coupling of the marginal transition kernels is not close to optimal. 
Alternative coupling algorithms and tailored
Wasserstein distance upper bounds between mixtures of distributions 
could give further improvements for this example.

\subsection{Interpretable upper bounds for \cubname} \label{subsection:theory}
So far we have established that \cubname \eqref{eq:W_ublimit1} consistently upper bounds
 Wasserstein distances (Sec.~\ref{subsection:consistency})
and developed algorithms to compute \cubname in practice (Sec.~\ref{subsection:algo_k_bar}).
We next derive 
upper bounds on the size of \cubname 
to provide interpretable sufficient conditions under which \cubname is guaranteed to be small.
We emphasize that it is possible for \cubname to be significantly smaller than these interpretable bounds and for \cubname to be small even when the assumptions of the interpretable bounds are not met.
Hence, when bounding Wasserstein distances in practice, we would not recommend computing these intepretable bounds but rather computing the even tighter \cubname Wasserstein bound directly.

Our analysis is based on Markov chain perturbation theory for $\mathcal{W}_1$
\citep{pillai2015ergodicity, johndrow2018error, rudolf2018}, 
which we generalize to $\mathcal{W}_p$ 
for all $p \geq 1$. This is a useful extension, as $\mathcal{W}_2$ in particular is believed to better reflect geometric features and adapt to geometric structure than  $\mathcal{W}_1$  \citep[Rem.~6.6]{villani2008optimal}.
We also discuss examples where the $\mathcal{W}_p$ upper bounds do not 
explicitly depend on the state space dimension and are stable up to a coupling 
of the one-step marginal kernels.

To establish our $\cub$ upper bounds, we assume that the Markovian coupling 
$\Gamma_1$ in Alg.~\ref{algo:coupled_kernel_single_step} 
gives uniform contraction in Wasserstein distance. 
Recall that $\Gamma_1$ is a coupling of the marginal kernel $K_1$ with itself, so 
Assump.~\ref{asm:gammaP_uniform_contract} concerns only the single kernel $K_1$ targeting the single stationary distribution $P$.

\begin{Asm}[Uniform contraction]
\label{asm:gammaP_uniform_contract} 
There exists $\rho \in (0,1)$ such that for all $X_t, \tilde{X}_t \in \mathcal{X}$
and $(X_{t+1}, \tilde{X}_{t+1}) | (X_t, \tilde{X}_t) \sim \Gamma_1(X_t, \tilde{X}_t)$,
$\mathbb{E}[ c(X_{t+1}, \tilde{X}_{t+1})^p | X_t, \tilde{X}_t ]^{1/p} \leq \rho c(X_t, \tilde{X}_t)$.
\end{Asm}

Assump.~\ref{asm:gammaP_uniform_contract} is stronger than the convergence assumption of the 
marginal chain corresponding to kernel $K_1$ (Assump.~\ref{asm:marginal_conv_Wp} for $(P_t)_{t \geq 0}$). 
For many popular MCMC algorithms, Assump.~\ref{asm:gammaP_uniform_contract} has been
established under certain metrics $c$ and coupled kernels $\Gamma_1$ to 
give contraction rates $\rho$ that do not explicitly depend on the dimension of the state space $\mathcal{X}$. 
This includes MALA \citep{eberle2014errorAoAP} 
and HMC \citep{bou-rabee2020couplingAOAP}.
When the target distributions are log-concave, these algorithms satisfy 
Assump.~\ref{asm:gammaP_uniform_contract} with 
$c(x,y)=\| x-y \|_2$ 
and the coupled kernel $\Gamma_1$ based on a CRN coupling. 
For target distributions satisfying a weaker distant dissipativity condition 
\citep{eberle2016reflectionPTRF, gorham2019measuringAoAP} (including, for example, multimodal distributions 
with Gaussian tails), 
these algorithms satisfy Assump.~\ref{asm:gammaP_uniform_contract}  
with $\Gamma_1$ based on a combination of CRN and reflection coupling 
and a metric $\tilde{c}$
satisfying $r \tilde{c}(x,y) \leq \| x-y \|_2 \leq R\, \tilde{c}(x,y)$ for some $0 < r \leq R < \infty$.

Furthermore, we can weaken Assump.~\ref{asm:gammaP_uniform_contract} 
to a geometric ergodicity condition as in \citep{rudolf2018}, where for some
constants $C \geq 1$, $\rho \in (0,1)$,
and for all $L \geq 1$, $\mathbb{E}[ c(X_{t+L}, Y_{t+L})^p | X_t, Y_t ]^{1/p} \leq C \rho^L c(X_t, Y_t)$
for $(X_{t+L}, Y_{t+L}) | (X_t, Y_t) \sim \Gamma^L_P(X_t, Y_t)$
where $\Gamma^L_P(X_t, Y_t)$ denotes a coupling of $L$-steps of the kernel $K_1$
marginally starting from states $X_t$ and $Y_t$. 
Our analysis then is based on the construction of a multi-step coupling kernel.
This may be of independent interest and is
included in App.~\ref{appendices:multi_step} for completeness.

Under Assump.~\ref{asm:gammaP_uniform_contract}, we can upper bound the distance from our
coupled chains explicitly in terms of the initial distribution $\bar{I_0}$, contraction constant $\rho$, and
coupled kernel $\Gamma_\Delta$ corresponding to perturbations between the marginal kernels $K_1$ and $K_2$. 

\begin{Th}[\cubname upper bound]
\label{thm:Wp_UB_single_step}
Let $(X_t, Y_t)_{t \geq 0}$ denote a coupled Markov chain
generated using Alg.~\ref{algo:coupled_chain_general} 
with initial distribution $\bar{I}_0$ and joint kernel $\bar{K}$ 
from Alg.~\ref{algo:coupled_kernel_single_step}. 
Suppose the coupled kernel $\Gamma_1$ satisfies 
Assump.~\ref{asm:gammaP_uniform_contract} for some $\rho \in (0,1)$. Then
\begin{talign}
\mathbb{E}[ \cub[p,t]^p ]^{1/p} = \mathbb{E} [ c(X_{t}, Y_{t})^p ]^{1/p} 
&\leq \rho^t \mathbb{E} [ c(X_0, Y_0)^p ]^{1/p} + \sum_{i=1}^t \rho^{t-i} \mathbb{E} [ \Delta_p(Y_{i-1}) ]^{1/p}
\end{talign}
for all $t \geq 0$, where $(X_0, Y_0) \sim \bar{I}_0$ and $\Delta_p(z) : = \mathbb{E}[ c(X,Y)^p|z]$ 
for $(X,Y)|z \sim \Gamma_\Delta(z)$.
\end{Th}

For $\cub[p,t]$ based on a metric $c$, one obtains an analogous bound if Assump.~\ref{asm:gammaP_uniform_contract} instead holds for a dominating metric $\tilde{c}$, i.e., for $\tilde{c}$  satisfying
$c(x,y) \leq R\, \tilde{c}(x,y) $ for some constant $R \in (0,\infty)$. Then 
$ \mathbb{E} [ c(X_{t}, Y_{t})^p ]^{1/p} \leq R\, \mathbb{E} [ \tilde{c}(X_{t}, Y_{t})^p ]^{1/p}$.
Also, when the marginal distributions $(Q_t)_{t \geq 0}$ converge, we can obtain a simpler
expression for the upper bound. 

\begin{Cor}[\cubname upper bound under marginal convergence]
\label{cor:Wp_UB_single_step}
Under the notation and assumptions of Thm.~\ref{thm:Wp_UB_single_step}, 
suppose that the marginal distributions $Q_t$ converge in $p$-Wasserstein 
distance to some distribution $Q$ as 
$t \rightarrow\infty$. 
Then for each $\epsilon >0$, there exists $S \geq 1$ such that for all $t \geq S$,
\begin{talign}
\mathbb{E}[ \cub[p,t]^p ]^{1/p} = \mathbb{E} [ c(X_{t}, Y_{t})^p ]^{1/p} \leq 
\rho^t \mathbb{E} [ c(X_0, Y_0)^p ]^{1/p} + (1 - \rho^t) \frac{\mathbb{E}[\Delta_p(Y^*)]^{1/p}}{1-\rho} + \epsilon.
\end{talign}
where $(X_0, Y_0) \sim \bar{I}_0$, $\Delta_p(z) \defeq \mathbb{E}[ c(X,Y)^p|z]$ for $(X,Y) \sim \Gamma_\Delta(z)$, and $Y^* \sim Q$.
\end{Cor}

Cor.~\ref{cor:Wp_UB_single_step} gives 
$
\mathcal{W}_p(P,Q) \leq 
\liminf_{t \rightarrow \infty} \mathbb{E}[ \cub[p,t]^p ]^{1/p} \leq 
\mathbb{E}[\Delta_p(Y^*)]^{1/p} / (1-\rho),
$ 
implying that \cubname estimators 
may give informative empirical upper bounds 
when the expected perturbation $\mathbb{E}[\Delta_p(Y^*)]$ for $Y^* \sim Q$ 
is small. Further if the contraction rate $\rho$ does not explicitly depend on the dimension,
then our upper bounds do not increase unfavorably with dimension and remain informative 
in high dimensional settings. Hence Cor.~\ref{cor:Wp_UB_single_step} 
provides interpretable sufficient conditions for \cubname to be dimension-free,  as in 
Figs.~\ref{fig:stylized_example_mvn_combined} and \ref{fig:stylized_example_ula_mala}. 

Our next result covers the case in which the marginals $(Q_t)_{t \geq 0}$ do not converge to any limiting distribution in $p$-Wasserstein distance.  
In this case, our upper bound is in terms of perturbations between the marginal kernels 
weighted by a Lyapunov function of $K_2$. 

\begin{Prop}[\cubname upper bound weighted by a Lyapunov function]
\label{prop:Wp_UB_single_step_lyapunov}
Under the notation and assumptions of Thm.~\ref{thm:Wp_UB_single_step}, 
let $V: \mathcal{X} \rightarrow [0, \infty)$ 
satisfy 
$\mathbb{E}[V(Y_{t+1})^p|Y_t=z] \leq \gamma V(z)^p + L$
for some fixed constants
$\gamma \in [0,1)$ and $L \in [0,\infty)$ and all $z \in \mathcal{X}$. Define
$\delta \defeq \sup_{z \in \mathcal{X}} \big( \frac{\Delta_p(z)
}{1 + V(z)^p} \big)^{1/p}$ and 
$\kappa \defeq \big(1+\max \big\{ \mathbb{E}[V(Y_{0})^p], \frac{L}{1-\gamma} \big\} \big)^{1/p}$,
where $\Delta_p(z) \defeq \mathbb{E}[ c(X,Y)^p|z]$ for $(X,Y) \sim \Gamma_\Delta(z)$. Then for all $t \geq 0$,
\begin{talign}
\mathbb{E}[ \cub[p,t]^p ]^{1/p} &= 
\mathbb{E} [ c(X_t, Y_t)^p ]^{1/p} \leq \rho^t \mathbb{E} [ c(X_0, Y_0)^p ]^{1/p} + (1-\rho^t)\frac{\delta \kappa}{1-\rho}.
\end{talign}
\end{Prop}

In the case $p=1$, Prop.~\ref{prop:Wp_UB_single_step_lyapunov}
recovers Thm.~3.1 of \citet{rudolf2018}. 
For such result to be informative, we require functions $V$ 
such that $\delta \kappa$ is small. 
An application of these results to three simple examples based on 
MALA, ULA, and stochastic gradient Langevin dynamics (SGLD) \citep{welling2011bayesianICML}  
chains is given in App.~\ref{appendices:theory_examples}.

\subsection{Comparison with alternative Wasserstein bounds} \label{subsection:comparison}
In this section, we compare our coupling-based Wasserstein bounds with alternatives. 

\para{Empirical Wasserstein and Sinkhorn distances.} 
A common approach to estimating $\mathcal{W}_p(P, Q)$ is to draw independent samples from $P$ and $Q$ and then exactly compute the $\mathcal{W}_p$ distance between the empirical distributions.
This is precisely the empirical Wasserstein estimate that appeared in Fig.~\ref{fig:stylized_example_mvn_combined}. As our next proposition, proved in App.~\ref{sec:empirical_wasserstein_proof}, demonstrates, this empirical Wasserstein approach consistently upper bounds $\mathcal{W}_p(P, Q)$. 
\newcommand{\iid}{\textrm{i.i.d.}\xspace}
\newcommand{\dist}{\sim}
\newcommand{\distiid}{\overset{\textrm{\tiny\iid}}{\dist}}
\begin{Prop}[Empirical Wasserstein distance bounds]
\label{prop:empirical_wasserstein}
For $P$ and $Q$ in $\mathcal{P}_p(\mathcal{X})$,
let $\hat{P}_n$, $\tilde{P}_n$, $\hat{Q}_n$, and $\tilde{Q}_n$ denote empirical distributions of the samples
$(X_i)_{i=1}^n$, $(\tilde{X}_i)_{i=1}^n$, $(Y_i)_{i=1}^n$, and $(\tilde{Y}_i)_{i=1}^n$
respectively, where $X_i, \tilde{X}_i \distiid P$ and, independently, $Y_i, \tilde{Y}_i \distiid Q$ for all $i=1,\mydots,n$.
Then, $\mathcal{W}_p(\hat{P}_n, \hat{Q}_n) \stackrel{\textup{a.s.}}{\to} \mathcal{W}_p(P, Q)$
as $n \rightarrow \infty$, and
\begin{talign}
0 \leq 
\mathbb{E}[ \mathcal{W}_p(\hat{P}_n, \hat{Q}_n)^p]^{1/p} - \mathcal{W}_p(P, Q) 
\leq 
\mathbb{E} [ \mathcal{W}_p(\hat{P}_n, \tilde{P}_n)^p ]^{1/p} +
\mathbb{E} [ \mathcal{W}_p (\hat{Q}_n, \tilde{Q}_n)^p ]^{1/p}.
\end{talign}
\end{Prop}

However, there are two downsides to the empirical Wasserstein approach.
The first is statistical.  
The difference between $\mathbb{E} \big[ \mathcal{W}_p(\hat{P}_n, \hat{Q}_n)^p \big]^{1/p}$ and $\mathcal{W}_p(P, Q)$ can be quite large and decay very slowly in $n$.  
For example, for some $d$-dimensional target distributions, 
$\mathbb{E}[\mathcal{W}_p(\hat{P}_n, \hat{Q}_n)]$ converges to $\mathcal{W}_p(P, Q)$ at rate $\Omega(n^{-1/d})$ when $d>2p$ \citep{weed2019sharpBERNOULLI}.
This can lead to the empirical Wasserstein distance giving loose upper bounds on $\mathcal{W}_p(P, Q)$
when the number of samples does not increase exponentially with dimension. 
The example in Fig.~\ref{fig:stylized_example_mvn_combined} illustrates this curse of dimensionality, where the estimator $\cub$ \eqref{eq:W_ublimit1} with CRN coupling gives tighter upper bounds of $\mathcal{W}_p(P, Q)$ than the empirical Wasserstein estimates. 

The second downside is computational.
Calculating $\mathcal{W}_p(\hat{P}_n, \hat{Q}_n)$ amounts to solving an 
uncapacitated minimum cost flow problem with $\mathcal{O}(n^3 \log n)$ 
computational cost \citep{orlin1988fasterSTOC}, 
prohibitive cost for large sample sizes. A popular alternative is to compute an entropy-regularized Wasserstein distance instead using the Sinkhorn algorithm \citep{cuturi2013sinkhorn}.
A larger value of the regularization parameter $\lambda > 0$ leads to faster computation but also introduces an additional bias that can compromise bound accuracy.
A smaller $\lambda$ leads to more expensive $\mathcal{O}(n^2/(\lambda \epsilon))$ 
computation time for $\epsilon$-accurate solutions \citep{altschuler2017nearNEURIPS}
and potential instability of the Sinkhorn algorithm in practice. 
See App.~\ref{appendices:sinkhorn} for simulations illustrating these issues.

In comparison, our coupling estimators run in time linear in the sample size $n$ and do not require the solution of any expensive optimization problems. 
On the other hand, empirical Wasserstein estimates will eventually converge to the true Wasserstein distance given sufficiently (perhaps exponentially) large sample sizes, so the empirical Wasserstein approach can lead to tighter bounds if one has a substantial computational budget.

\para{The approach of Huggins et al.}
\citet{Huggins2020validatedAISTATS} derive upper bounds for Euclidean  
Wasserstein distances 
in terms of KL or 
$\alpha$-divergences. 
To estimate their upper bounds of $\mathcal{W}_p(P,Q)$ 
for $P$ and $Q$ in $\mathcal{P}_p(\mathbb{R}^d)$
and $P$ absolutely continuous with respect to $Q$, 
\citeauthor{Huggins2020validatedAISTATS} propose importance sampling based estimates which require 
samples from $Q$, evaluations of the normalized density of $Q$, and evaluations of the unnormalized 
density of $P$.  
Fig.~\ref{fig:huggins_dobson_comparison} (left) plots the performance
of the $\mathcal{W}_2$ bounds of \citeauthor{Huggins2020validatedAISTATS} for the example in Sec.~\ref{subsection:wass_ub}. 
The dot-dashed line represents the mean of $I=20$ independent  \citeauthor{Huggins2020validatedAISTATS}  importance-sampling estimators, each with $2T=3000$ samples
from $Q$. 
The $\cub[2]$ estimator plotted for comparison uses $I$ independent CRN coupled chains with trajectory length $T$ and burnin $S=500$.  
In this example, the \citeauthor{Huggins2020validatedAISTATS} bounds   are significantly looser than both our CRN coupling bound and the independent coupling upper bound.
Furthermore, the \citeauthor{Huggins2020validatedAISTATS} estimates exhibit an increasing variance in higher dimensions, as shown by the large grey error bands.
One advantage of the \citeauthor{Huggins2020validatedAISTATS} estimates over $\cub[2]$ 
is that samples from $P$ are not required. 
On the other hand, unlike the \citeauthor{Huggins2020validatedAISTATS} estimates, $\cub$ remains applicable even when the 
density of $Q$ cannot be evaluated.
This case arises for many approximate MCMC algorithms such as ULA in Sec.~\ref{subsection:bias_ub}, the 
stochastic gradient-based samplers in Sec.~\ref{subsection:tall_data}, and the
matrix approximation-based sampler in Sec.~\ref{subsection:half_t}. 

\begin{figure}[!t]
\captionsetup[subfigure]{font=footnotesize,labelfont=footnotesize}
    \centering
    \begin{subfigure}[b]{0.43\textwidth}
        \includegraphics[width=\textwidth]{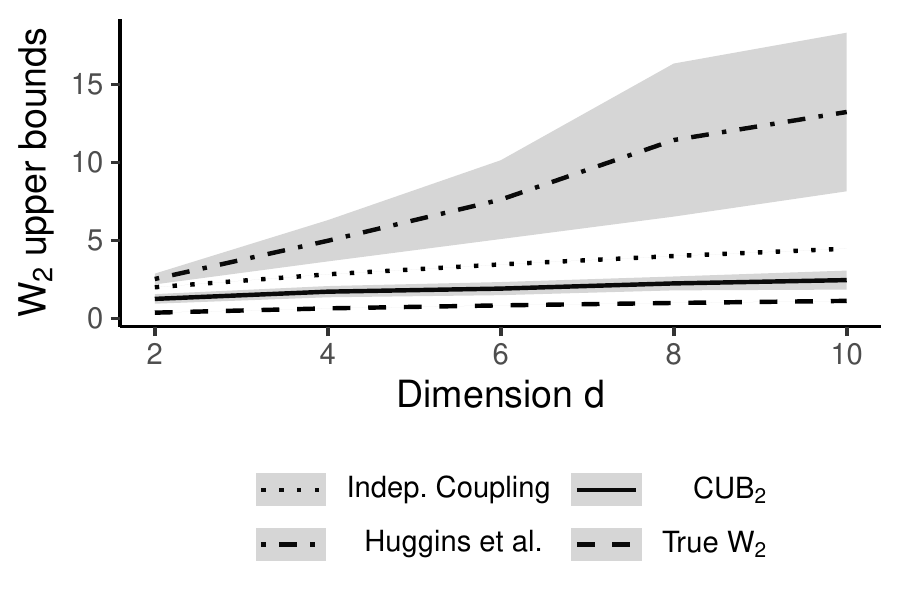}
    \end{subfigure}
    \begin{subfigure}[b]{0.43\textwidth}
        \includegraphics[width=\textwidth]{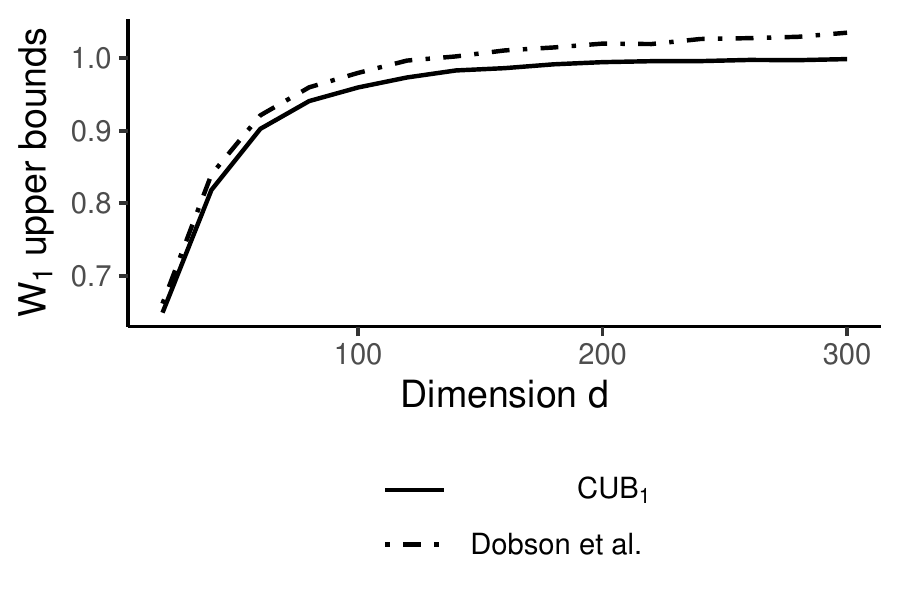} 
    \end{subfigure}
    \caption{\textbf{(Left)} Upper bound estimates for $\mathcal{W}_2$ with $c(x,y) = \twonorm{x-y}$ between 
    $P = \mathcal{N}(0, \Sigma)$ and $Q=\mathcal{N}(0, I_d)$ 
    for
    $[\Sigma]_{i,j}=0.5^{|i-j|}$ for $1 \leq i, j \leq d$. 
    The \citet{Huggins2020validatedAISTATS} bound is looser than $\cub[2]$ and has larger variance as the dimension grows.
    See Sec.~\ref{subsection:comparison} for more details. 
    \textbf{(Right)} Upper bound estimates for 
    $\mathcal{W}_1$ with 
    $c(x,y) = \min \{1, \|x-y\|_2 \}$ between 
    ULA and MALA chains targeting $P = \mathcal{N}(0, \Sigma)$. 
    In line with Prop.~\ref{prop:dobson_bound}, the $\cub[1]$ \eqref{eq:W_ublimit1} estimate is tighter than the \citet{dobson2021usingSIAM} bound employing the same CRN coupling.
    See Sec.~\ref{subsection:comparison} for more details.}
    \label{fig:huggins_dobson_comparison}
\end{figure}

\para{The approach of Dobson et al.} 
\citet{dobson2021usingSIAM} apply couplings to assess the quality of 
numerical approximation of stochastic differential equations.  
Specifically, they focus on the 1-Wasserstein distance with the capped metric 
$c(x,y) = \min \{1, \|x-y\|_2 \}$ on $\mathbb{R}^d$ and derive  
upper bounds in terms of the contraction constant of one of the marginal chains
which are then estimated using couplings. 
Our next result, proved in App.~\ref{appendices:dobson},
shows that $\mathbb{E}[\cub[1]]$ with the same 
coupling provides a tighter upper bound than the proposal of \cite{dobson2021usingSIAM}.

\begin{Prop}[\cubname lower bounds Dobson et al.]
\label{prop:dobson_bound}
Consider the $1$-Wasserstein distance with metric 
$c(x,y) = \min \{1, \|x-y\|_2 \}$ on $\mathbb{R}^d$. 
Then, for any coupling and sufficiently large burn-in, $\mathbb{E}[\cub[1]]$  
\eqref{eq:W_ublimit1} lower bounds the estimated upper bound of \cite{dobson2021usingSIAM}.
\end{Prop}

Fig.~\ref{fig:huggins_dobson_comparison} (right) plots the $1$-Wasserstein 
upper bounds of \citeauthor{dobson2021usingSIAM} and $\cub[1]$ 
for the example in Sec.~\ref{subsection:bias_ub}
with the capped metric $c(x,y) = \min \{1, \|x-y\|_2 \}$ on $\mathbb{R}^d$.
We use $I=100$ independent coupled chains 
with trajectory length $T=3000$ and burnin $S=1000$ to estimate both the 
upper bounds of \citeauthor{dobson2021usingSIAM} and $\cub[1]$. 
The figure shows that, in line with Prop.~\ref{prop:dobson_bound}, the upper bounds of 
\citeauthor{dobson2021usingSIAM} are looser than $\cub[1]$.

\section{Applications} \label{section:applications}
We now illustrate the value of our methods for three practical applications.
We focus on the $2$-Wasserstein distance with $c(x,y) = \twonorm{x-y}$ on $\mathbb{R}^d$,  
which by \eqref{eq:pwass_moment} controls first and second order moments 
and captures geometric features induced by the Euclidean norm $\twonorm{\cdot}$. 
In this case a tractably estimated lower bound on the Wasserstein distance is also available.
For any $P, Q \in \mathcal{P}_2(\mathbb{R}^d)$,  let $P_i$ and $Q_i$ denote the 
marginal distributions of the $i^{th}$ component of $P$ and $Q$ respectively.  
Let $\mathcal{N}_P$ and $\mathcal{N}_Q$ denote Gaussian distributions on $\mathbb{R}^d$ 
with the same means and covariance matrices as $P$ and $Q$ respectively. Then,
\begin{talign} \label{eq:W2L2LB} 
\max \Big\{ \sum_{i=1}^d \mathcal{W}_2(P_i, Q_i)^2 \ , \ \mathcal{W}_2(\mathcal{N}_P,  \mathcal{N}_Q)^2 \Big\} 
&\leq  \mathcal{W}_2(P, Q)^2.
\end{talign}
Here,  $\sum_{i=1}^d \mathcal{W}_2(P_i, Q_i)^2 \leq \mathcal{W}_2(P, Q)^2$ follows 
from the coupling representation of $\mathcal{W}_2(P, Q)$,  
and $\mathcal{W}_2(\mathcal{N}_P,  \mathcal{N}_Q) \leq  \mathcal{W}_2(P, Q) $ is the lower bound of 
\citet[][Thm.~2.1]{gelbrich1990onMN}. 
Each one-dimensional Wasserstein distance 
$\mathcal{W}_2^2(P_i, Q_i)$ admits a convenient representation for estimation,
given by $\int_0^1 ( F^{-1}_{P_i}(u) - F^{-1}_{Q_i}(u) )^2 du$
where $F^{-1}_{P_i}$ and $F^{-1}_{Q_i}$ are the inverse cumulative distribution functions 
of $P_i$ and $Q_i$ respectively, while 
$\mathcal{W}_2(\mathcal{N}_P,  \mathcal{N}_Q)$ has the closed form 
$\big(\| \mu_P - \mu_Q \|_2^2 + \textup{Trace} \big(\Sigma_P + \Sigma_Q - 2(\Sigma_P^{1/2}\Sigma_Q \Sigma_P^{1/2})^{1/2}\big)\big)^{1/2}$
in terms of the means $\mu_P, \mu_Q$ and covariances $\Sigma_P, \Sigma_Q$ of $P$ and $Q$ \citep[Rem.~2.23]{Peyre2019computationalFTML}. 
Since the true Wasserstein distances are unknown in our applications to follow, we will assess the tightness of our coupling-based upper bounds by estimating the lower bound  \eqref{eq:W2L2LB}.
Details of all the datasets, algorithms, and specific estimator parameters used in
this section can be found in App.~\ref{appendices:applications}.

\subsection{Approximate MCMC and variational inference for tall data} \label{subsection:tall_data}
Our first application concerns Bayesian inference for \textit{tall} datasets \citep{bardenet2017onJMLR}, 
where the number of observations $n$ is large compared to the dimension $d$.  
In such settings, exact MCMC can be computationally expensive with $\Omega(n)$ cost per iteration. 
This computational bottleneck 
and the prevalence of tall datasets 
has catalyzed much interest in approximate MCMC 
and variational approximation based algorithms. 
Approximate MCMC algorithms include ULA and
stochastic gradient MCMC 
(see \citep{nemeth2021stochasticJASA} for a review) 
such as SGLD \citep{welling2011bayesianICML}.
Popular variational approximation methods include Laplace approximation 
\citep[e.g.,][]{tierney1986accurateJASA} and 
variational Bayes 
(VB, see \citep{blei2017JASAvariational} for a review).

In this section, we assess the quality of these sampling algorithms. 
We consider ULA,  SGLD,  Laplace approximation, and mean field VB applied to Bayesian logistic regression with a Gaussian prior
for the Pima diabetes dataset \citep{smith1988using} 
and the DS1 life sciences dataset \citep{ds1dataset}. 
For each sampling algorithm, 
Fig.~\ref{fig:tall_data} plots $\cub[2]$ 
\eqref{eq:W_ublimit1} upper bounds and $\mathcal{W}_2$ lower bounds estimated using \eqref{eq:W2L2LB}.  
We simulate the coupled chains $(X^{(i)}_t,Y^{(i)}_t)_{t \geq 0}$ independently for each $i$, where each 
$(X^{(i)}_t)_{t \geq 0}$ is a MALA chain targeting the posterior $P$ and 
each $(Y^{(i)}_t)_{t \geq 0}$ is linked to an approximate MCMC or a variational procedure. 
In particular,  
we consider $(Y^{(i)}_t)_{t \geq 0}$ to be an ULA chain, 
SGLD chains based on sub-sampling $10\%$ and $50\%$ of the observations, 
a MALA chain targeting $\mathcal{N}(\mu_{L}, \Sigma_{L})$ where 
$\mu_{L} \in \mathbb{R}^d$ and $\Sigma_{L} \in \mathbb{R}^{d \times d}$ 
are from a Laplace approximation of $P$, and 
a MALA chain targeting $\mathcal{N}(\mu_{VB}, \Sigma_{VB})$ where 
$\mu_{VB} \in \mathbb{R}^d$ and $\Sigma_{VB} \in \Diag(\mathbb{R}^{d \times d})$ 
are from a Gaussian mean field VB approximation of $P$. 
In each case, we use a CRN coupling between the 
marginal kernels of $(X^{(i)}_t)_{t \geq 0}$ and $(Y^{(i)}_t)_{t \geq 0}$. 
App.~\ref{appendices:tall_data} contains details 
about the datasets, algorithms and 
estimator parameters used.

Fig.~\ref{fig:tall_data} shows that 
Laplace approximation has the smallest asymptotic bias 
for both datasets.
This promising Laplace performance can be linked to 
posterior concentration and accuracy of the corresponding Bernstein-von Mises approximation 
\citep{bardenet2017onJMLR, chopin2017leaveSS}. 
Our bounds also show how the Metropolis--Hastings correction and stochastic gradients 
affect the quality of ULA and SGLD.
Overall, this application illustrates the effectiveness of our proposed quality 
measures for comparing approximate inference algorithms in the tall data setting. 

\begin{figure}[!t]
\captionsetup[subfigure]{font=footnotesize, labelfont=footnotesize}
    \centering
    \includegraphics[width=0.8\textwidth]{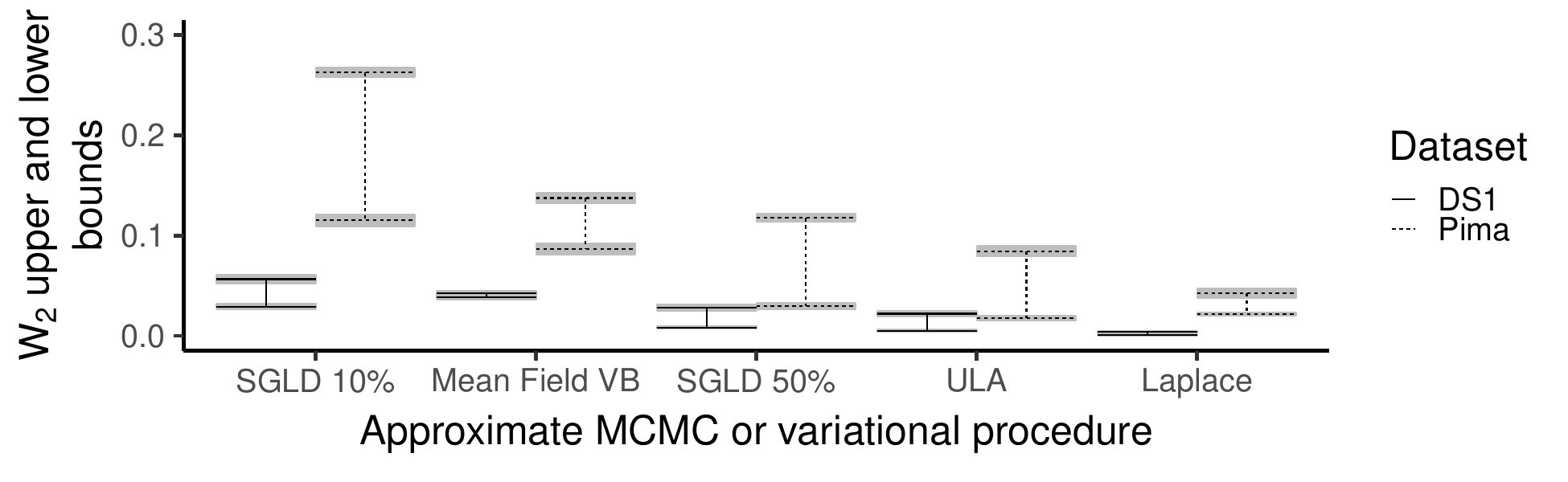}
    \caption{
    Bounds on the Euclidean $\mathcal{W}_2$ bias of approximate MCMC and variational inference procedures for Bayesian logistic regression.
    We consider the DS1 dataset ($n = 26732$ observations, $d = 10$ covariates)  and
    the Pima dataset ($n = 768$, $d = 8$).
    See Sec.~\ref{subsection:tall_data} for more details.}
    \label{fig:tall_data}
\end{figure}

\subsection{Approximate MCMC for high-dimensional linear regression} \label{subsection:half_t}
We now consider high-dimensional Bayesian linear regression, where the 
dimension $d$ is larger than the number of observations $n$.
The likelihood for the response vector $y \in \mathbb{R}^n$ is 
a Gaussian density 
with mean $X \beta$ and covariance matrix $\sigma^2 I_n$, where 
$X \in \mathbb{R}^{n \times d}$ is the design matrix, 
$\beta \in \mathbb{R}^d$ is an unknown signal vector, 
and $\sigma^2 > 0$ is the unknown noise variance.
We consider a class of global-local mixture priors, given by 
\begin{talign} \label{eq:half_t}
\xi^{-1/2} \sim \mathcal{C}_+(0,1), \ \eta_j^{-1/2} \overset{i.i.d.}{\sim} t_+(\nu), \ 
\sigma^{-2} \sim \mathrm{Gamma} \Big( \frac{a_0}{2}, \frac{b_0}{2} \Big), \
\beta_j | \eta, \xi, \sigma^2 \overset{ind.}{\sim} \mathcal{N} \Big( 0, \frac{\sigma^2}{\xi \eta_j} \Big)
\end{talign}
where $\mathcal{C}_+(0,1)$ is the half-Cauchy distribution on $[0, \infty)$ and
$t_+(\nu)$ is the half-t distribution on $[0, \infty)$ with $\nu$ degrees of freedom. 
When $\nu=1$, this corresponds to the popular Horseshoe prior 
\citep{carvalho2010theBIOMETRIKA}. 
This setting differs considerably from the log-concave tall data example 
of Sec.~\ref{subsection:tall_data}, 
as now the posterior distribution is multi-modal, has polynomial tails,
and has infinite density about the origin \citep{biswas2021coupled}. 
\citet{johndrow2020scalableJMLR} have developed exact and approximate Gibbs samplers
for the Horseshoe prior in this setting, which involves an approximation
parameter $\epsilon \geq 0$. \citet{biswas2021coupled} extended the 
sampler of \citeauthor{johndrow2020scalableJMLR} 
to all $\nu \geq 1$ and showed that using larger values of $\nu$ 
could improve mixing times in high dimensions.

In this section, we use couplings to assess the quality of such approximate
MCMC algorithms.
Following \citeauthor{biswas2021coupled}, we consider $\nu=2$ 
applied to a genome-wide 
association study (GWAS) dataset \citep{buhlmann2014highARSA} 
and a synthetic dataset.
We use a CRN coupling with the marginal chains 
corresponding to the exact and the approximate MCMC kernel. 
App.~\ref{appendices:half_t} contains details 
about the datasets, algorithms, and estimator parameters used. 

\begin{figure}[!thb]
\captionsetup[subfigure]{font=footnotesize, labelfont=footnotesize}
    \centering
\begin{subfigure}[b]{0.4\textwidth}
    \includegraphics[width=\textwidth]{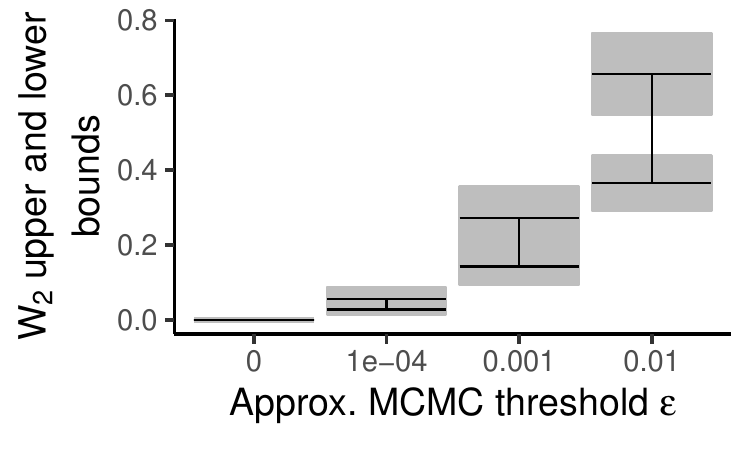}
    \caption{GWAS dataset ($n=71, d=4088$)}
\end{subfigure}
\begin{subfigure}[b]{0.4\textwidth}
    \includegraphics[width=\textwidth]{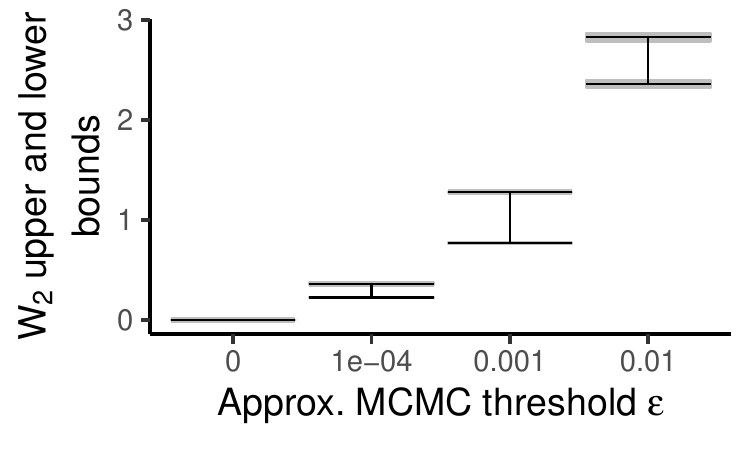}
    \caption{Synthetic dataset ($n=500, d=50000$)}
\end{subfigure}
	\caption{Bounds on the Euclidean $\mathcal{W}_2$ bias of an approximate MCMC Gibbs sampler for 
	high-dimensional Bayesian regression with half-t($2$) prior, $n$ observations, and $d$ covariates.  We consider both a bacteria GWAS dataset 
    and a synthetic dataset.
    See Sec.~\ref{subsection:half_t} for more details.}
    \label{fig:half_t}
\end{figure}

Fig.~\ref{fig:half_t} plots upper and lower bounds on the 
2-Wasserstein distance, illustrating how
asymptotic bias of the approximate Gibbs sampler varies with 
the approximation parameter $\epsilon \geq 0$.
The upper bounds are given by our estimator $\cub[2]$ 
\eqref{eq:W_ublimit1}, and the lower bounds are estimated using \eqref{eq:W2L2LB}. 
For developers of such high-dimensional approximate MCMC samplers, these bounds provide an empirical assessment of the trade-off between improved quality and higher computational cost. 
In particular, the bounds enable a developer to assess the computational cost of an approximation procedure as a function of the bias introduced (and vice-versa).
For example, for any maximum acceptable bias level, one can identify the largest approximation parameter $\epsilon$ with a \cubname interval below the acceptable level and assess the computational savings delivered 
relative to an exact sampler.

Often one will choose an acceptable level of Wasserstein bias based on the direct implications for downstream inferential tasks (e.g., based on tolerable discrepancies in predictive accuracy or numerical integration, as discussed in Section~\ref{subsection:wass}).
When it is otherwise difficult for a user to select an acceptable level of Wasserstein bias on an absolute scale, we would recommend normalizing each \cubname estimate based on the coupled chains $(X^{(i)}_{t}, Y^{(i)}_{t})_{t = 0}^T$ by a second, independent-coupling \cubname estimate based on the chains $(X^{(i)}_{t}, \tilde{X}^{(i)}_{t})_{t = 0}^T$, where $(\tilde{X}^{(i)}_{t})_{t = 0}^T$ is sampled independently of $(X^{(i)}_{t})_{t = 0}^T$ using the $P$-invariant $K_1$ kernel.
This enables Wasserstein bias to be assessed relative to a 
measure of the intrinsic variability or noise level in the target distribution $P$.

\subsection{Approximate MCMC for high-dimensional logistic regression} \label{subsection:skinny_gibbs}
We now consider high-dimensional Bayesian logistic regression with 
spike and slab priors, a popular choice 
for Bayesian variable selection
\citep{vannucci2021handbook}.
\citet{narisetty2019skinnyJASA} recently developed an approximate MCMC algorithm
called \textit{Skinny Gibbs}, to sample from  posteriors in this setting. 
Here, we assess the quality of the Skinny Gibbs algorithm
applied to a malware dataset \citep{Dua2019UCI} and a lymph node GWAS 
\citep{narisetty2019skinnyJASA} dataset using a CRN coupling between 
the exact MCMC kernel and the Skinny Gibbs kernel. 
App.~\ref{appendices:skinny_gibbs} contains further details 
about spike and slab priors and the datasets, algorithms, and estimator parameters used. 

Fig.~\ref{fig:spike_slab} displays $\cub[2]$ 
\eqref{eq:W_ublimit1} upper bounds and lower bounds estimated using \eqref{eq:W2L2LB} on the Euclidean $2$-Wasserstein distance between the limiting distributions of the exact and 
Skinny Gibbs chains for $\beta$.
We display these bounds not to draw comparisons across the datasets but rather to exemplify the level of precision provided by \cubname when applied to real high-dimensional logistic regression tasks.
For researchers developing 
approximate samplers, 
these bounds provide an empirical assessment of asymptotic 
bias for different datasets and posteriors under the spike and slab prior.

\begin{figure}[!t]
    \centering
\includegraphics[width=0.55\textwidth]{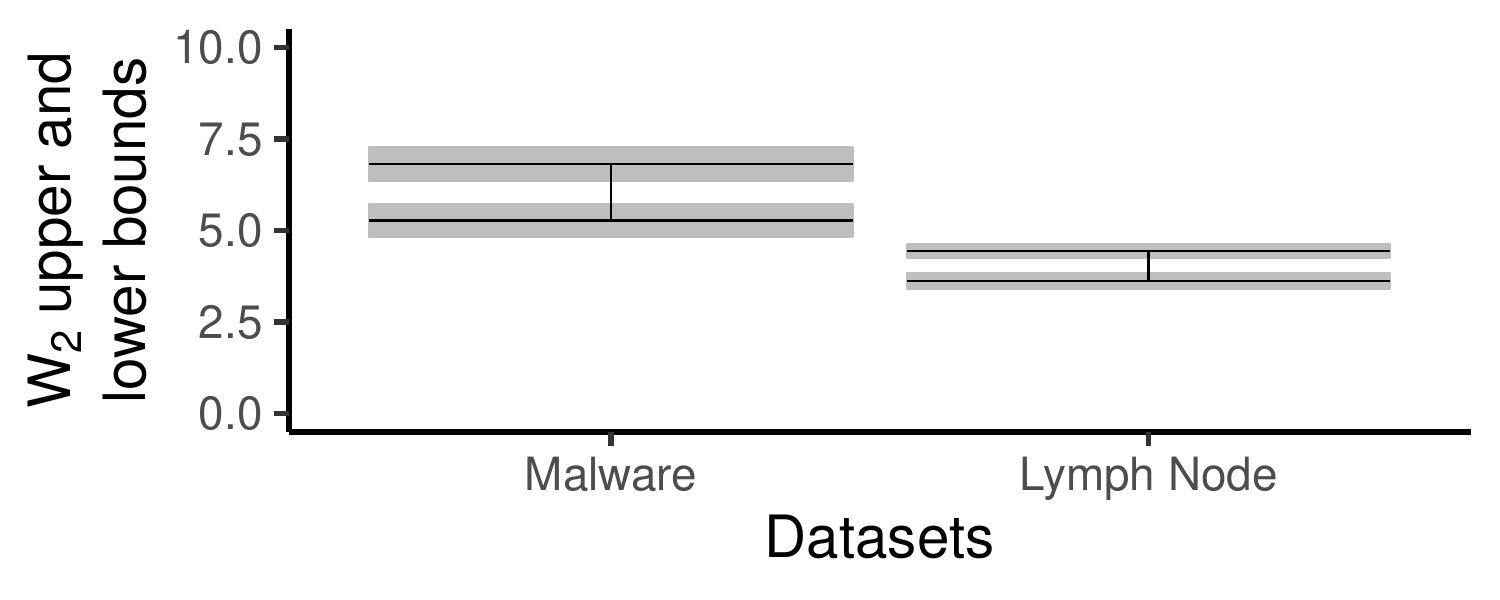}    
\caption{Bounds on the Euclidean $\mathcal{W}_2$ bias of the Skinny Gibbs sampler \citep{narisetty2019skinnyJASA}
for Bayesian logistic regression with a spike and slab prior; see Sec.~\ref{subsection:skinny_gibbs} for details.
We consider a malware dataset ($n = 373$ observations; $d = 503$ covariates) and a
 lymph node GWAS dataset ($n = 148$, $d = 4514$).}
\label{fig:spike_slab}
\end{figure}

\section{Discussion} \label{section:discuss}
We have introduced new estimators to assess the quality of 
approximate inference procedures. 
The estimators consistently bound the Wasserstein distance
between the limiting distribution of the approximation 
and the original target distribution of interest.
The proposed estimators can be applied to approximate MCMC and certain
variational inference methods in practical settings, including  Bayesian regression in $50000$ dimensions. 

The following questions arise from our work. 

\textit{Alternative coupling algorithms.} 
We have chosen CRN coupling as a practical default for our experiments due to its broad applicability, but a growing inventory of alternative coupling 
strategies is available \citep{heng2019unbiased, lee2020coupled,
xu2021couplings, oleary2021maximal, biswas2021coupled}, and, as evidenced in Sec.~\ref{subsection:algo_k_bar}, alternative couplings tailored to the problem can yield tighter upper bounds.
An important open question is how to best identify or construct a better coupling 
for a given problem at hand.

\textit{Avoiding sampling from
an asymptotically unbiased Markov chain.} 
Our proposed upper bounds 
require sampling from a $P$-invariant Markov chain $(X_t)_{t \geq 0}$. 
This raises the question: can one construct a Markov chain 
$(Y'_t, Y_t)_{t \geq 0}$ such that 
(i) $(Y'_t)_{t \geq 0}$ and $(Y_t)_{t \geq 0}$ are identically distributed 
according to the same asymptotically biased chain marginally and (ii)
$\mathbb{E}[c(X_t, Y'_t)^p] = \mathbb{E}[c(X_t, Y_t)^p] \leq \mathbb{E}[c(Y'_t, Y_t)^p]$ for all $t \geq 0$, where 
$(X_t)_{t \geq 0}$ is an asymptotically unbiased chain? Then we could sample from the computationally less expensive chain $(Y'_t, Y_t)_{t \geq 0}$ 
to obtain an upper bound of $\mathbb{E}[c(X_t, Y_t)^p]^{1/p}$ which is only loose
by a constant factor of 2, as 
$\mathbb{E}[c( Y'_t, Y_t)^p]^{1/p} \leq  \mathbb{E}[c(X_t, Y_t)^p]^{1/p} +
\mathbb{E}[c(X_t, Y'_t)^p]^{1/p} = 2\mathbb{E}[c(X_t, Y_t)^p]^{1/p}.$
We hope to investigate such coupling 
constructions in follow-up work. 

\textit{Upper bounds for total variation distance.} 
The $1$-Wasserstein distance with metric $c(x, y)=\mathrm{I}\{x \neq y\}$ gives
the popular total variation (TV) distance,  which always 
takes values in $[0,1]$ and is invariant to reparameterization. 
To obtain upper bounds of TV strictly less than 1 using
our estimators, we require couplings which allow exact meetings
between the two marginal chains. 
Our initial attempts at using
maximal couplings \citep{johnson1998coupling, 
jacob2020unbiasedJRSSB,oleary2021maximal} have not been effective 
in high dimensions and suggest a need for further methodological work. 

\textit{Spot checking.} 
Finally, an anonymous associate editor suggested the following additional application.
Often one is interested in approximating 
an entire family of target distributions $P_\eta$ with approximations $Q_\eta$ indexed by a parameter $\eta$ taking a large number of distinct values in $\mathbb{R}$.
When it is feasible to run a $P_\eta$-invariant Markov chain only for a small number of $\eta$ values 
but infeasible to run these exact chains for all target $\eta$ values, \cubname can be used to spot check Wasserstein quality at a small set of representative $\eta$ values and drive decision making around the 
degree or type of approximation 
used for the full collection of $\eta$ values.

\paragraph{Acknowledgments.}
We thank Juan Shen for sharing the Lymph Node dataset,
and Pierre E. Jacob, Xiao-Li Meng, the participants of the International Conference
on Monte Carlo Methods and Applications and the BayesComp workshop on ``Measuring the quality of MCMC output'' for helpful feedback.
We also thank the anonymous reviewers and associate editor for their valuable comments and suggestions.
NB was supported by the NSF grant DMS-1844695, a GSAS Merit Fellowship, and 
a Two Sigma Fellowship Award. 

\bibliography{references}

\appendix

\section{Additional figures and discussion} \label{appendices:calcs}
\subsection{Calculation of empirical Wasserstein bounds in Figure \ref{fig:stylized_example_mvn_combined}.}
\label{appendices:empirical_wass}
In this section we note how the empirical Wasserstein upper bounds and error bands 
in Figure \ref{fig:stylized_example_mvn_combined} are generated. Our upper bounds are based on Proposition \ref{prop:empirical_wasserstein}, which gives 
\begin{equation}
\mathcal{W}_p(P, Q)^p \leq \mathbb{E} \big[ \mathcal{W}_p(\hat{P}_T, \hat{Q}_T)^p \big]
\end{equation}
where $P$ and $Q$ are distributions on the metric space $(\mathcal{X}, c)$ with finite moments of order $p$,
and $\hat{P}_T$ and $\hat{Q}_T$ denote empirical distributions of the samples 
$(X_1, \mydots, X_T)$ and $(Y_1, \mydots, Y_T)$ where $X_i \sim P$ and $Y_i \sim Q$ for all $i=1, \mydots, T$.
For $p=2$ and $P \neq Q$, the dot-dashed lines in Figure \ref{fig:stylized_example_mvn_combined} plots the corresponding estimate of this upper bound, given by
\begin{equation} \label{eq:empirical_wass_estimate1}
    \Big( \frac{1}{I} \sum_{i=1}^I \mathcal{W}_2(\hat{P}^{(i)}_T,\hat{Q}^{(i)}_T)^2 \Big)^{1/2}
\end{equation}
where $\hat{P}^{(i)}_T$ and $\hat{Q}^{(i)}_T$ are empirical distribution
of $P$ and $Q$ respectively based on $T$ samples. For each $i=1,\mydots,I$,
such empirical distributions $\hat{P}^{(i)}_T$ and $\hat{Q}^{(i)}_T$ are generated independently and 
then $\mathcal{W}_2(\hat{P}^{(i)}_T,\hat{Q}^{(i)}_T)$ is calculated by solving a linear program. 
The error bands plot 95\% confidence intervals given by 
$\Big[ \frac{1}{I} \sum_{i=1}^I \mathcal{W}_2(\hat{P}^{(i)}_T,\hat{Q}^{(i)}_T)^2 \pm 1.96 \hat{\sigma}/\sqrt{I} \Big]^{1/2}$
where $\hat{\sigma}^2$ is the empirical 
variance of $\big( \mathcal{W}_2(\hat{P}^{(i)}_T,\hat{Q}^{(i)}_T)^2 \big)_{i=1}^I$.

Instead of \eqref{eq:empirical_wass_estimate1}, one could alternatively use the estimator $\mathcal{W}_2(\hat{P}_{IT},\hat{Q}_{IT})$
where $\hat{P}_{IT}$ and $\hat{Q}_{IT}$ are empirical distribution of $P$ and $Q$ respectively based on $IT$ samples.
Using $\mathcal{W}_2(\hat{P}_{IT},\hat{Q}_{IT})$ produces a tighter upper bound estimate compared to using \eqref{eq:empirical_wass_estimate1}, which is linked to consistency of empirical Wasserstein distance based estimates covered in Proposition \ref{prop:empirical_wasserstein} of Section \ref{subsection:comparison}. 
However, this numerical improvement is minor; for example in Figure \ref{fig:stylized_example_mvn_combined} (Left) with dimension $d=100$, a tighter empirical upper bound of $11.35$ is obtained using this estimator compared to the upper bound of $11.83$ using \eqref{eq:empirical_wass_estimate1} and both these upper bound estimates are looser than the coupling based upper bound estimate of $5.78$. Such minor numerical improvement is linked to the curse of dimensionality for empirical Wasserstein distances, as discussed in Sections \ref{subsection:wass} and \ref{subsection:comparison}. Furthermore, calculating $\mathcal{W}_2(\hat{P}_{IT},\hat{Q}_{IT})$ for this example requires approximately $10$ times greater numerical runtimes compared to calculating \eqref{eq:empirical_wass_estimate1}.

\subsection{Section {\ref{subsection:bias_ub}} calculations.}
\label{appendices:stylized_ula_mala}
As kernel $K_1$ is $P$ invariant, 
$X_t \sim P_t \overset{t \rightarrow \infty}{\Rightarrow} P$ for all $t \geq 0$ \citep[e.g.][]{roberts1996exponentialBERNOULLI}. 
The ULA chain $(Y_t)_{t \geq 0}$ corresponds to an auto-regressive $AR(1)$ model, where 
\begin{flalign}
Y_{t} &= ( I_d - (\sigma_Q^2/2) \Sigma^{-1} ) Y_{t-1} + \sigma_Q Z_{t} = B Y_{t-1} + \sigma_Q Z_{t}
\end{flalign}
for all $t \geq 0$, where $Y_{0} \sim \mathcal{N}(0, I_d)$, $Z_{t} \overset{i.i.d.}{\sim} \mathcal{N}(0, I_d)$ and
$Z_{0} \defeq Y_{0}$, and $B  = ( I_d - (\sigma_Q^2/2) \Sigma^{-1} )$. By induction,
\begin{flalign} \label{eq:ula_marginals_calc}
Y_{t} &= B^{t} Z_0 + \sigma_Q \Big( B^{t-1} Z_1 + B^{t-2} Z_2 + \mydots + Z_{t} \Big) \\
&= B^{t} Z_0 + \sigma_Q \sum_{j=0}^{t-1} B^{j} Z_{t-j} \\
&\sim \mathcal{N} \big(0, B^{2t} + \sigma_Q^2 \sum_{j=0}^
{t-1} B^{2j} \big) =:Q_t
\end{flalign}
as required. Finally, note that for $\sigma_Q = 0.5 d^{-1/6}$ sufficiently small such that 
$\opnorm{B}<1$ (where $\opnorm{\cdot}$ is the matrix operator norm), 
$\underset{t \rightarrow \infty}{\lim} \big( B^{2t} + \sum_{j=0}^{t-1} B^{2j} \big) = (I_d - B^2)^{-1}$
(see, e.g. \citet{shumstof2000time} for sufficient conditions for the convergence $AR(1)$ models). This gives
$Q_t \overset{t \rightarrow \infty}{\Rightarrow} \mathcal{N}(0, \sigma_Q^2 (I_d - B^2)^{-1})=:Q$.

\paragraph{ULA asymptotic bias upper bound calculation for Figure \ref{fig:stylized_example_ula_mala}.}
We recall a result of \citet{durmus2019highBERNOULLI} on the asymptotic bias of ULA.
\begin{Prop}\citep[Corollary 9]{durmus2019highBERNOULLI} 
Consider an ULA Markov chain targeting the distribution $\pi$ on $\mathbb{R}^d$ 
with un-normalized density $\exp(-U(x))$. For $\| \cdot \|_2$ the 
Euclidean norm on $\mathbb{R}^d$, assume:
\begin{enumerate}
\item $U$ is continuously differentiable and lipschitz: there exists 
some $L \geq 0$ such that for all $x,y \in \mathbb{R}^d$,
$$\| \nabla U(x) - \nabla U(y) \| \leq L \| x -y \|_2.$$
\item $U$ is $m$-strongly convex for some $m > 0$:
there exists some $m > 0$ such that for all $x,y \in \mathbb{R}^d$,
$$U(x) \leq U(y) + \langle  \nabla U(x), y-x \rangle + (m/2) \| x -y \|^2_2$$
\item $U$ is three times continuously differentiable and there exists
some $\tilde{L}>0$ such that for all $x,y \in \mathbb{R}^d$,
$$\| \nabla^2 U(x) - \nabla^2 U(y) \|_2 \leq \tilde{L} \| x -y \|_2.$$ 
\end{enumerate}
Let the step size $\sigma$ of the Markov chain be sufficiently small
such that $\gamma \defeq \sigma^2/2 < 1/(m+L)$. Then the ULA Markov chain 
converges to some distribution $\pi_{\gamma}$, and
\begin{equation} \label{eq:durmus2019}
\mathcal{W}_2(\pi, \pi_{\gamma})^2 \leq 
2 \kappa^{-1} \gamma^2 d \Big( 2L^2+ \gamma L^4 \big( \frac{\gamma}{6} + \frac{1}{m}\big) 
+ \kappa^{-1} \big( \frac{4 d \tilde{L}^2}{3} + \gamma L^4 + \frac{4 L^4}{3m}\big)\Big)
\end{equation}
where $\kappa = 2mL/(m+L)$. 
\end{Prop}
The dotted line in Figure \ref{fig:stylized_example_ula_mala} 
is plotted by applying \eqref{eq:durmus2019} for $\pi = \mathcal{N}(0, \Sigma)$, 
where $L = \lambda_{min}(\Sigma)^{-1}$, $m = \lambda_{max}(\Sigma)^{-1}$
and $\tilde{L}=0$. Here $\lambda_{max}(\Sigma)$ and $\lambda_{min}(\Sigma)$
are the largest and smallest eigenvalue of $\Sigma$ respectively.

\subsection{Non-asymptotic upper bounds using L-Lag coupling} \label{appendices:llag}
In this section, we discuss how to avoid burn-in removal and instead directly correct 
our bound for non-stationarity using the recent $L$-lag coupling approach of \citet{biswas2019estimating}
in the case of the 1-Wasserstein distance.

We first informally outline the approach of \citet{biswas2019estimating}. 
Consider a Markov chain on $(\mathcal{X},c)$ with transition kernel $K_1$, marginal distributions $(P_t)_{t\geq0}$
and a unique stationary distribution $P$. 
Consider a joint kernel $\bar{K}_1$ on $\mathcal{X} \times \mathcal{X}$ such that 
$\bar{K}_1((x,y), (\cdot, \mathcal{X})) = K_1(x, \cdot)$ and  
$\bar{K}_1((x,y), (\mathcal{X}, \cdot)) = K_1(y, \cdot)$ for all 
$x,y \in \mathcal{X}$. Then the \emph{$L$-lag coupling} chain
$(\tilde{X}_{t-L}, X_t)_{t \geq L}$ is generated by sampling 
$X_0$ and $\tilde{X}_0$ independently from a common initial distribution $P_0$, sampling 
$X_t | X_{t-1} \sim K_1(X_{t-1}, \cdot)$ for $t=1,\mydots,L$, and generating 
$(\tilde{X}_{t-L}, X_t)|\tilde{X}_{t-L-1}, X_{t-1} \sim \bar{K}_1((\tilde{X}_{t-L-1}, X_{t-1}), \cdot)$
for $t > L$. Crucially, the joint kernel $\bar{K}_1$ is 
designed such that: (i) the marginal chains $(\tilde{X}_{t-L})_{t \geq L}$
and $(X_t)_{t \geq 0}$ exactly meet such that the random meeting time 
$\tau \defeq \inf\{t>L: \tilde{X}_{t-L} = X_{t} \}$ is almost surely finite and 
(ii) the chains remain faithful after meeting such that $\tilde{X}_{t-L} = X_{t}$
for all $t \geq \tau$. 
Suppose the coupled chain $(\tilde{X}_{t-L}, X_t)_{t \geq L}$
satisfies Assumptions \ref{asm:llag_asm_1}, \ref{asm:llag_asm_2} and 
\ref{asm:llag_asm_3} \citep{biswas2019estimating, jacob2020unbiasedJRSSB}
(see \citet{middleton2020unbiasedEJS} for the use of polynomially-tailed
meeting times). 

\begin{Asm}[Marginal convergence and uniformly bounded moments] \label{asm:llag_asm_1} 
Marginal distributions $(P_t)_{t\geq0}$ converge to $P$ in 1-Wasserstein distance, and 
for all $t \geq L$, $\mathbb{E}[c(\tilde{X}_{t-L}, X_{t})^{2+\eta}] \leq D$
for some constants $ \eta > 0$ and $D < \infty$.
\end{Asm}
\begin{Asm}[Sub-exponentially tailed meeting times] \label{asm:llag_asm_2} 
The meeting times $\tau \defeq \inf \{ t > L : X_t = \tilde{X}_{t-L} \}$ 
satisfies $\mathbb{P}( \frac{\tau-L}{L} >t) \leq C \delta^t$
for some constants $C<\infty$ and $\delta \in (0,1)$ and all $t\geq 0$.
\end{Asm}
\begin{Asm}[Faithfulness after meeting] \label{asm:llag_asm_3} 
$X_t = \tilde{X}_{t-L}$ for all $t \geq \tau$.
\end{Asm}

Under Assumptions \ref{asm:llag_asm_1}, \ref{asm:llag_asm_2} and 
\ref{asm:llag_asm_3}, \citet{biswas2019estimating} obtain
\begin{flalign}
    \mathcal{W}_1(P_t, P) &\leq \sum_{j=1}^\infty \mathcal{W}_1(P_{t+jL-L}, P_{t+jL}) \label{eq:llag_apply1} \\
    &\leq \sum_{j=1}^\infty \mathbb{E}[ c(\tilde{X}_{t+jL-L}, X_{t+jL}) ] \label{eq:llag_apply2} \\
    &= \mathbb{E}\Big[\sum_{j=1}^{\infty} c(\tilde{X}_{t+jL-L}, X_{t+jL})\Big] \label{eq:llag_apply3} \\
    &= \mathbb{E}\Big[\sum_{j=1}^{\lceil(\tau-L-t)/L\rceil} c(\tilde{X}_{t+jL-L}, X_{t+jL})\Big] \label{eq:llag_apply4},
\end{flalign}
where \eqref{eq:llag_apply1} follows from the triangle inequality using Assumption \ref{asm:llag_asm_1}, 
\eqref{eq:llag_apply2} follows from the coupling representation of the 
Wasserstein distance, and \eqref{eq:llag_apply3} follows from 
interchanging the summation and expectation using the dominated
convergence theorem under Assumptions \ref{asm:llag_asm_1} and \ref{asm:llag_asm_2}, 
and \eqref{eq:llag_apply4} follows as $c(\tilde{X}_{t+jL-L}, X_{t+jL})=0$ for
all $j > \lceil(\tau-L-t)/L\rceil$ under Assumption \ref{asm:llag_asm_3}. 
Note that $\tau$ has finite expectation under Assumption \ref{asm:llag_asm_2}, which means 
the upper bound in \eqref{eq:llag_apply4} can be estimated in finite time. 
We can estimate this upper bound by simulating
multiple $L$-lag coupled chains $(\tilde{X}_{t-L}, X_t)_{\tau \geq t \geq L}$ 
independently and using the empirical average
$$ \frac{1}{I} \sum_{i=1}^I \sum_{j=1}^{\lceil(\tau^{(i)}-L-t)/L\rceil} c(\tilde{X}^{(i)}_{t+jL-L}, X^{(i)}_{t+jL})$$
where $I \geq 1$ is the  number of independent coupled chains.

The following Proposition employs this upper bound alongside $\cub[1]$ \eqref{eq:W_ublimit1}
to obtain a non-asymptotic upper bound on $\mathcal{W}_1(P,Q)$. 

\begin{Prop}[Non-asymptotic upper bound] \label{prop:Llag_bound_non_asym}
For any lag $L \geq 1$, consider the coupled chain $(\tilde{X}_{t-L}, X_t, Y_t, \tilde{Y}_{t-L})_{t \geq L}$
such that $(\tilde{X}_{t-L}, X_t)_{t \geq L}$ is an $L$-lag coupling chain for the kernel $K_1$, 
$(\tilde{Y}_{t-L}, Y_t)_{t \geq L}$ is an $L$-lag coupling chain for the kernel $K_2$, and $(X_t, Y_t)_{t \geq L}$
is a coupled chain sampled using Algorithm \ref{algo:coupled_chain_general}. Under 
Assumption \ref{asm:marginal_conv_Wp} with $p=1$ and Assumptions 
\ref{asm:llag_asm_1}, \ref{asm:llag_asm_2} and \ref{asm:llag_asm_3} for the coupled
chains $(\tilde{X}_{t-L}, X_t)_{t \geq L}$ and $(\tilde{Y}_{t-L}, Y_t)_{t \geq L}$, 
\begin{equation} \label{eq:Llag_bound_non_asym}
\mathcal{W}_1(P, Q) \leq \mathbb{E}[\cub[1,t]] + 
\mathbb{E}[\sum_{j=1}^{ \lceil (\tau_P - L -t)/L \rceil } c(\tilde{X}_{t+(j-1)L}, X_{t+jL})] + 
\mathbb{E}[\sum_{j=1}^{ \lceil (\tau_Q - L -t)/L \rceil } c(\tilde{Y}_{t+(j-1)L}, Y_{t+jL})]
\end{equation}
for all $t \geq 0$, where $\tau_P \defeq \inf\{t>L: \tilde{X}_{t-L} = X_{t} \}$ and 
$\tau_Q \defeq \inf\{t>L: \tilde{Y}_{t-L} = Y_{t} \}$.
\end{Prop}

\subsection{Sinkhorn algorithm simulations for Section \ref{subsection:comparison}} \label{appendices:sinkhorn}
In this section we consider the impact of the regularization parameter 
of the Sinkhorn algorithm. Figure \ref{fig:sinkhorn_distance} of this section plots the Wasserstein
distance upper bounds for the stylized example in Section \ref{subsection:wass_ub}. 
In particular, we consider the $2$-Wasserstein distance with Euclidean norm on $\mathbb{R}^d$,
and the distributions $P = \mathcal{N}(0, \Sigma)$ where $\Sigma_{i,j} = 0.5^{|i-j|}$
for $1 \leq i,j \leq d$ and $Q= \mathcal{N}(0, I_d)$ in the case of dimension $d=10$. 

The $\cub[2]$ \eqref{eq:W_ublimit1} estimate
(black line) in Figure \ref{fig:sinkhorn_distance} is based a CRN coupling of 
marginal MALA kernels, with $I = 10$ independent 
coupling chains and trajectories of length $T=500$ with a burn-in of 
$S = 100$ for each chain. The true Wasserstein (black dot-dashed line)
distance and the upper bound from indepdendent coupling (black dotted line) 
are analytically tractable, as given in Section \ref{subsection:wass_ub}. 
For different values of the entropic regularization parameter $\lambda$, 
the grey solid line plots the induced distance of the 
optimal matching obtained from the Sinkhorn algorithm. For each $\lambda$, 
we implement the Sinkhorn algorithm on empirical distributions
with $IT=5000$ sample points from $P$ and $Q$. Figure \ref{fig:sinkhorn_distance}
shows that we require a small entropic regularization parameter $\lambda$
to obtain informative upper bounds using the Sinkhorn algorithm. 
On the other hand, Figure \ref{fig:sinkhorn_time} shows that the runtime
for the Sinkhorn algorithm increases dramatically for smaller values of
$\lambda$. This example illustrates that the Sinkhorn algorithm has
expensive runtime precisely for the smaller values of $\lambda$
that give tighter upper bounds to the Wasserstein distance.
In comparison, the $\cub[2]$ \eqref{eq:W_ublimit1} estimate
does not require solving any expensive optimization problem.

\begin{figure}[!t]
\captionsetup[subfigure]{font=footnotesize,labelfont=footnotesize}
    \centering
    \begin{subfigure}[b]{0.43\textwidth}
        \includegraphics[width=\textwidth]{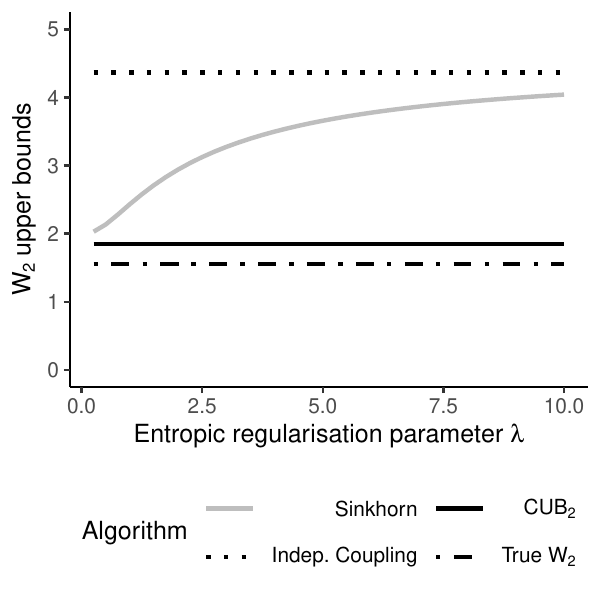}
        \caption{$\mathcal{W}_2$ upper bounds with varying $\lambda$.}
        \label{fig:sinkhorn_distance}
    \end{subfigure}
    \begin{subfigure}[b]{0.43\textwidth}
        \includegraphics[width=\textwidth]{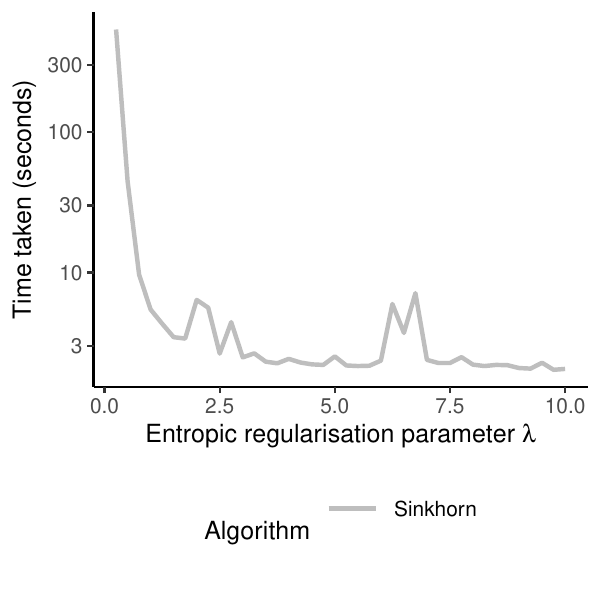} 
        \caption{Sinkhorn runtime with varying $\lambda$.}
        \label{fig:sinkhorn_time}
    \end{subfigure}
    \caption{Figure \ref{fig:sinkhorn_distance} plots upper bound estimates for $\mathcal{W}_2(P,Q)$ with $P = \mathcal{N}(0, \Sigma) \text{ where } \Sigma_{i,j} = 0.5^{|i-j|} \text{ for } 1 \leq i,j \leq d$, $Q= \mathcal{N}(0, I_d)$, metric $c(x,y) = \twonorm{x-y}$ and dimension $d=10$. Figure \ref{fig:sinkhorn_time} plots the runtime of the Sinkhorn algorithm.
    }
    \label{fig:sinkhorn}
\end{figure}

\section{Proofs} \label{appendices:proofs}
\subsection{Consistency proofs}
\paragraph{Technical Results. } We first collect some technical results for reference. 
\begin{Lemma} \label{lemma:series_limit}
Let $(a_j)_{j \geq 0}$ be a real sequence with 
$a_j \overset{j \rightarrow \infty}{\rightarrow} 0$, and let $\rho \in (0,1)$. 
Then $\sum_{j=1}^t \rho^{t-j} a_j \overset{t \rightarrow \infty}{\rightarrow}  0$.
\end{Lemma}
\begin{proof}[Proof of \ref{lemma:series_limit}]
As $a_j \overset{j \rightarrow \infty}{\rightarrow} 0$, the sequence $(a_j)_{j \geq 0}$
is bounded by some $M \in (0, \infty)$. Also for all $\epsilon >0$, there
exists some $j_0 \geq 1$ such that $|a_j| < \epsilon$ for all $j \geq j_0$. 
For all $t > j_0$, this gives
\begin{flalign}
\Big| \sum_{j=1}^t \rho^{t-j} a_j \Big| &\leq \sum_{j=1}^{j_0} \rho^{t-j} |a_j| + \sum_{j=j_0+1}^t \rho^{t-j} |a_j| 
\leq M \rho^{t-j_0} \frac{1 - \rho^{j_0}}{1-\rho} + \epsilon \frac{1-\rho^{t-j_0}}{1-\rho}.
\end{flalign}
Taking $t \rightarrow \infty$, we obtain 
$\lim_{t \rightarrow \infty} \Big| \sum_{j=1}^t \rho^{t-j} a_j \Big| \leq \epsilon/(1-\rho)$,
where $\epsilon/(1-\rho)$ can be made arbitrarily small.
\end{proof}

\begin{Lemma}\label{lemma:slln}
Let $(\xi_i)_{i \geq 0}$ be independent and identically distributed non-negative 
random variables with $\mathbb{E}[\xi_1] < \infty$, and let $S_n = \sum_{i=1}^n \xi_i$. Then
as $n\to\infty$, $S_n/n \toasL \mathbb{E}[ \xi_1 ]$ and for any
$p \geq 1$, $(S_n/n)^{1/p} \toasL \mathbb{E}[\xi_1]^{1/p}$. 
\end{Lemma}
\begin{proof}[Proof of \ref{lemma:slln}]
As $n$ tends to infinity, $S_n/n \toasL \mathbb{E}[ \xi_1 ]$ follows from the proof of 
the Strong law of large numbers using backwards martingales
(see, e.g., \citet[Theorem 4.7.1 and Example 4.7.4]{durrett2019probability}).
$(S_n/n)^{1/p} \toas (\mathbb{E}[ \xi_1 ])^{1/p}$ follows from 
$S_n/n \toas \mathbb{E}[ \xi_1 ]$ by continuous mapping theorem on $[0, \infty)$. 
Finally, for $p \geq 1$, 
\begin{equation}
    \mathbb{E}[|(S_n/n)^{1/p} - (\mathbb{E}[ \xi_1 ])^{1/p}|] \leq 
    \mathbb{E}[|(S_n/n) - \mathbb{E}[\xi_1]|^{1/p}] \leq
    \mathbb{E}[|(S_n/n) - \mathbb{E}[\xi_1]|]^{1/p} \stackrel{n\to\infty}{\to} 0
\end{equation}
where the first inequality follows as $|a^{1/p}-b^{1/p}| \leq |a-b|^{1/p}$
for all $a, b \geq 0$ and $p \geq 1$, the
second inequality follows from Jensen's inequality
and the limit follows as $S_n/n \toL \mathbb{E}[ \xi_1 ]$. Therefore, 
$(S_n/n)^{1/p} \toL \mathbb{E}[ \xi_1 ]^{1/p}$.
\end{proof}

\begin{proof}[Proof of Proposition \ref{prop:Wp_UB_t}: Consistency of instantaneous \cubname]
Note that $\mathcal{W}_p(P_t, Q_t)$ is well-defined and $\mathbb{E}[c(X_t,Y_t)^p]$ is finite as
distributions $P$ and $Q$ have finite moments of order $p$. We obtain
\begin{flalign}
\mathcal{W}_p(P_t, Q_t)^p \leq \mathbb{E}[c(X_t,Y_t)^p] = \mathbb{E}[ \cub[p,t]^p ],
\end{flalign}
where the inequality follows from the coupling representation of Wasserstein distance,
and the equality follows from the definition of $\cub[p,t]$. As $\mathbb{E}[ \cub[p,t]^p ]<\infty$, 
by Lemma \ref{lemma:slln}, $\cub[p,t]^p \toasL \mathbb{E}[ \cub[p,t]^p ]$ as $I\to\infty$. 
\end{proof}

\begin{proof}[Proof of Corollary \ref{cor:Wp_UB_average_t}:
Consistency of CUB for time-averaged marginals]
We first show that 
$$\mathcal{W}_p( \frac{1}{T-S} \sum_{t=S+1}^T P_t, \frac{1}{T-S} \sum_{t=S+1}^T Q_t)^p \leq 
\frac{1}{T-S} \sum_{t=S+1}^T \mathcal{W}_p( P_t, Q_t)^p.$$
Let $\gamma_t$ denote the $p$-Wasserstein optimal coupling between distributions
$P_t$ and $Q_t$ for $t=S+1,\mydots,T$. Sample the coupling $(X^*,Y^*)$ such that  
$(X^*,Y^*)| U^*=t \sim \gamma_t$ for $U^* \sim \Uniform( \{S+1,\mydots, T\})$. Then 
$X^* \sim \frac{1}{T-S} \sum_{t=S+1}^T P_t$ and $Y^* \sim \frac{1}{T-S} \sum_{t=S+1}^T Q_t$
marginally, and 
\begin{flalign}
\mathcal{W}_p \Big( \frac{1}{T-S} \sum_{t=S+1}^T P_t, \frac{1}{T-S} \sum_{t=S+1}^T Q_t \Big)^p &\leq 
\mathbb{E}[c(X^*,Y^*)^p] \text{ by the coupling representation of } \mathcal{W}_p \\
&= \frac{1}{T-S} \sum_{t=S+1}^T \mathbb{E}[ c(X^*,Y^*)^p | U^*=t] \\
&= \frac{1}{T-S} \sum_{t=S+1}^T \mathcal{W}_p( P_t, Q_t)^p.
\end{flalign}
Now by Proposition \ref{prop:Wp_UB_t} and definition \eqref{eq:W_ublimit1}, 
\begin{flalign}
\frac{1}{T-S} \sum_{t=S+1}^T \mathcal{W}_p( P_t, Q_t)^p
\leq \mathbb{E} \Big[ \frac{1}{T-S} \sum_{t=S+1}^T \cub[p,t]^p \Big] 
&= \mathbb{E}[ \cub[p]^p ].
\end{flalign}
As $\mathbb{E}[ \cub[p]^p ] < \infty$, by Lemma \ref{lemma:slln}
$\cub[p]^p \toasL \mathbb{E}[ \cub[p]^p ]$
as $I\to\infty$.
\end{proof}

\begin{proof}[Proof of Corollary \ref{cor:Wp_UB}:
Consistency of \cubname with stationary initialization]
Note that $\mathcal{W}_p(P, Q)$ is well-defined and $\sum_{t=S+1}^T \mathbb{E}[c(X_t,Y_t)^p]/(T-S)$ 
is finite as distributions $P_t$ and $Q_t$ have finite moments of order $p$. We obtain,
\begin{flalign}
\mathcal{W}_p(P, Q)^p = \frac{1}{T-S} \sum_{t=S+1}^T \mathcal{W}_p(P_t, Q_t)^p 
\leq \frac{1}{T-S} \sum_{t=S+1}^T \mathbb{E}[c(X_t,Y_t)^p] = \mathbb{E}[ \cub^p ].
\end{flalign}
where the first equality follows as $P_t=P$ and $Q_t=Q$ for all $t \geq 0$, the
inequality follows Proposition \ref{prop:Wp_UB_t},
and the last equality follows from the definition of $\cub$.
As $\mathbb{E}[ \cub[p]^p ] < \infty$, by Lemma \ref{lemma:slln} 
$\cub[p]^p \toasL \mathbb{E}[ \cub[p]^p ]$
as $I\to\infty$.
\end{proof}

\begin{proof}[Proof of Proposition \ref{prop:Wp_UB_marginal_conv}:
Consistency when chain marginals converge]
Let $(P_t)_{t \geq 0}$ and $(Q_t)_{t \geq 0}$ denote the marginal distributions
of Markov chains $(X_t)_{t \geq 0}$ and $(Y_t)_{t \geq 0}$ respectively.
By Assumption \ref{asm:marginal_conv_Wp}, distributions 
$(P_t)_{t \geq 0}$, $(Q_t)_{t \geq 0}$, $P$ and $Q$ all have finite moments of
order $p$. Then for all $t \geq 1$, 
\begin{flalign}
\mathcal{W}_p(P,Q) &\leq \mathcal{W}_p(P,P_t) + \mathcal{W}_p(P_t,Q_t) + \mathcal{W}_p(Q_t,Q) \label{eq:apply_wp_metric}\\
&\leq \mathcal{W}_p(P,P_t) + \mathbb{E}[c(X_t, Y_t)^p]^{1/p} + \mathcal{W}_p(Q_t,Q), \label{eq:apply_coupling_rep_wass}
\end{flalign}
where \eqref{eq:apply_wp_metric} follows by the triangle inequality as
$\mathcal{W}_p$ is a metric on the space of measure on $\mathcal{X}$ with 
finite moments of order $p$ , and \eqref{eq:apply_coupling_rep_wass}
follows from the coupling representation of $\mathcal{W}_p$. By Assumption
\ref{asm:marginal_conv_Wp}, $\lim_{t \rightarrow \infty} \mathcal{W}_p(P,P_t) = 0$
and $\lim_{t \rightarrow \infty} \mathcal{W}_p(Q_t, Q) = 0$. 
Taking the limit infimum in \eqref{eq:apply_coupling_rep_wass} and raising to
the $p^{th}$ exponent gives 
$ \mathcal{W}_p(P,Q)^p \leq \liminf_{t \rightarrow \infty} \mathbb{E}[c(X_t, Y_t)^p] $.
Therefore for all
$\epsilon > 0$, there exists $S \geq 1$ such that for all $t \geq S$, 
$\mathcal{W}_p(P,Q)^p \leq \epsilon + \mathbb{E}[c(X_t, Y_t)^p]$, and
\begin{flalign}
\mathcal{W}_p(P,Q)^p &\leq  \epsilon + \frac{1}{T-S} \sum_{t=S+1}^T \mathbb{E}[c(X_t, Y_t)^p]
= \epsilon + \mathbb{E}[ \cub^p]
\end{flalign}
for all $T \geq S$.
As $\mathbb{E}[ \cub[p]^p ]^p < \infty$, by Lemma \ref{lemma:slln} 
$\cub[p]^p \toasL \mathbb{E}[ \cub[p]^p ]$
as $I\to\infty$.
\end{proof}

\begin{proof}[Proof of Proposition \ref{prop:Llag_bound_non_asym}:
Non-asymptotic upper bound]
By the triangle inequality, 
\begin{flalign}
\mathcal{W}_1(P, Q) \leq \mathcal{W}_1(P_t, Q_t) + \mathcal{W}_1(P_t,P) + \mathcal{W}_1(P_t,P).
\end{flalign}
By Proposition \ref{prop:Wp_UB_t}, $\mathcal{W}_1(P_t, Q_t) \leq \mathbb{E} [
\cub[1,t] ]$. Under assumptions \ref{asm:llag_asm_1}, 
\ref{asm:llag_asm_2} and \ref{asm:llag_asm_3}, by \citet[Theorem 2.5]{biswas2019estimating} 
\begin{flalign}
\mathcal{W}_1(P_t,P) &\leq \mathbb{E} \big[
\sum_{j=1}^{ \lceil (\tau_P - L -t)/L \rceil } c(\tilde{X}_{t+(j-1)L}, X_{t+jL}) \big] \text{ and } \\
\mathcal{W}_1(Q_t,Q) &\leq \mathbb{E} \big[
\sum_{j=1}^{ \lceil (\tau_Q - L -t)/L \rceil } c(\tilde{Y}_{t+(j-1)L}, Y_{t+jL}) \big].
\end{flalign}
Equation \eqref{eq:Llag_bound_non_asym} now directly follows. As the meeting times 
$\tau_P$ and $\tau_Q$ have sub-exponential tails by Assumption \ref{asm:llag_asm_2}, 
the $L$-lag upper bounds can be estimated in finite time.
\end{proof}

\subsection{Wasserstein upper bound proofs}
\begin{proof}[Proof of Theorem \ref{thm:Wp_UB_single_step}:
\cubname upper bound]
Under the coupled kernel $\bar{K}$ from Algorithm \ref{algo:coupled_kernel_single_step}, 
for each $t \geq 1$ we have the coupling $(X_{t},Z_{t},Y_{t})$
where $(X_{t},Z_{t}) | X_{t-1}, Y_{t-1} \sim \Gamma_1(X_{t-1}, Y_{t-1})$ and
$(Z_{t}, Y_{t}) | X_{t-1}, Y_{t-1} \sim \Gamma_\Delta (Y_{t-1})$. This gives
\begin{flalign}
\mathbb{E} [ c(X_{t}, Y_{t})^p ]^{1/p} &= \mathbb{E}[\mathbb{E}[c(X_{t}, Y_{t})^p |X_{t-1}, Y_{t-1}]]^{1/p}  \\
&\leq \mathbb{E}[ \mathbb{E}[ \big( c(X_{t}, Z_{t}) + c(Z_{t}, Y_{t}) \big)^p |X_{t-1}, Y_{t-1} ] ]^{1/p} \label{eq:Wp_UB_dim_free_apply1} \\
&\leq \mathbb{E}[ \mathbb{E}[ c(X_{t}, Z_{t})^p |X_{t-1}, Y_{t-1} ] ]^{1/p} + 
\mathbb{E}[ \mathbb{E}[ c(Z_{t}, Y_{t})^p |X_{t-1}, Y_{t-1} ] ]^{1/p} \label{eq:Wp_UB_dim_free_apply2} \\
&\leq \rho \mathbb{E}[ c(X_{t-1}, Y_{t-1})^p ]^{1/p}  + \mathbb{E}[ \Delta_p(Y_{t-1}) ]^{1/p}\label{eq:Wp_UB_dim_free_apply3}
\end{flalign}
where \eqref{eq:Wp_UB_dim_free_apply1} follows as $c$ is a metric,
\eqref{eq:Wp_UB_dim_free_apply2} follows by Minkowski's inequality, and
\eqref{eq:Wp_UB_dim_free_apply3} follows by Assumption \ref{asm:gammaP_uniform_contract} 
with $\Delta_p(z) \defeq \mathbb{E}[c(X,Y)^p|z]$ for $(X,Y) \sim \Gamma_\Delta(z)$. 
By induction, \eqref{eq:Wp_UB_dim_free_apply3} implies
\begin{flalign}
\mathbb{E} [ c(X_{t}, Y_{t})^p ]^{1/p} 
&\leq \rho^t \mathbb{E} [ c(X_0, Y_0)^p ]^{1/p} + \sum_{i=1}^t \rho^{t-i} \mathbb{E} [ \Delta_p(Y_{i-1}) ]^{1/p}.
\end{flalign}
\end{proof}

\begin{proof}[Proof of Corollary \ref{cor:Wp_UB_single_step}:
\cubname upper bound under marginal convergence]
Denote $a\defeq \mathbb{E}[ \Delta_p(Y^*)]^{1/p}$ for $Y^* \sim Q$ and 
$a_k \defeq \mathbb{E} [\Delta_p(Y_{k})]^{1/p}$ for $k \geq 0$. 
Then $a_k \overset{k \rightarrow \infty}{\rightarrow} a$, because 
$Q_t$ converges in $p$-Wasserstein distance to $Q$ as $t \rightarrow \infty$. 
By Lemma \ref{lemma:series_limit}, this implies 
\begin{flalign}
\sum_{i=1}^t \rho^{t-i} a_{i-1}
\overset{t \rightarrow \infty}{\rightarrow} 
\sum_{i=1}^t \rho^{t-i} a = \frac{1-\rho^t}{1-\rho} a.
\end{flalign} 
Therefore, for all $\epsilon>0$ there exists $S \geq 1$ such that for all $t \geq S$,
$\sum_{i=1}^t \rho^{t-i} |a_{i}-a| < \epsilon$. 
By Theorem \ref{thm:Wp_UB_single_step}, 
\begin{flalign}
\mathbb{E} [ c(X_{t}, Y_{t})^p ]^{1/p} 
&\leq \rho^t \mathbb{E} [ c(X_0, Y_0)^p ]^{1/p} + \sum_{i=1}^t \rho^{t-i} a_{i-1} \\
&\leq \rho^t \mathbb{E} [ c(X_0, Y_0)^p ]^{1/p} 
+ \sum_{i=1}^t \rho^{t-i} a
+ \sum_{i=1}^t \rho^{t-i} |a_{i-1}-a|   \\
&= \rho^t \mathbb{E} [ c(X_0, Y_0)^p ]^{1/p} 
+ \frac{1-\rho^t}{1-\rho} a + \epsilon .
\end{flalign}
\end{proof}

\begin{proof}[Proof of Proposition \ref{prop:Wp_UB_single_step_lyapunov}:
\cubname upper bound weighted by a Lyapunov function]
As $V$ is a  a $p^{th}$-order Lyapunov function of $K_2$, by induction 
\begin{flalign} \label{eq:lyapunov_KQ_apply}
\mathbb{E}[V(Y_{i})^p] \leq \gamma^i \mathbb{E}[V(Y_{0})^p] + (1-\gamma^i)\frac{L}{1-\gamma}
\text{ for all } i \geq 0.
\end{flalign}
for all $i \geq 0$. Therefore,
\begin{alignat*}{2}
\mathbb{E} [ \Delta_p(Y_{i}) ] &\leq \delta \mathbb{E}[1+ V(Y_{i-1})^p] \leq 
\delta^p \Big( 1 + \gamma^{i-1} \mathbb{E}[V(Y_{0})^p] + (1-\gamma^{i-1})\frac{L}{1-\gamma} \Big) 
\leq  \delta^p \kappa^p
\end{alignat*}
for all $i \geq 1$, where the first inequality follows from the definition of $\delta$, second 
inequality from \eqref{eq:lyapunov_KQ_apply}, and 
the second inequality from the definition of $\kappa$.
By Theorem \ref{thm:Wp_UB_single_step}, we obtain
\begin{flalign}
\mathbb{E} [ c(X_{t}, Y_{t})^p ]^{1/p} 
&\leq \rho^t \mathbb{E} [ c(X_0, Y_0)^p ]^{1/p} + \sum_{i=1}^t \rho^{t-i} \mathbb{E} \Big[ \Delta_p(Y_{i-1}) \Big]^{1/p} 
\\
&\leq \rho^t \mathbb{E} [ c(X_0, Y_0)^p ]^{1/p} + \delta \kappa \sum_{i=1}^t \rho^{t-i} \\
&= \rho^t \mathbb{E} [ c(X_0, Y_0)^p ]^{1/p} + (1-\rho^t) \frac{\delta \kappa}{1- \rho}. 
\end{flalign}
\end{proof}

\subsection{Wasserstein distances of empirical distributions proofs}\label{sec:empirical_wasserstein_proof}
To prove Proposition \ref{prop:empirical_wasserstein}, we first record a technical result.
\begin{Lemma} \label{lemma:empirical_wass}
Suppose $S$ and $T$ are distributions on the metric space $(\mathcal{X}, c)$ with finite moments of order $p$, and $n \geq 1$ is an integer. 
Given $U_i \sim S$ for $i=1, \mydots, n$, let $\hat{S}_n$ denote the empirical distribution of $(U_1, \mydots, U_n)$. 
Then, 
\begin{equation}
\mathcal{W}_p(S,T)^p \leq \mathbb{E}[ \mathcal{W}_p( \hat{S}_n, T )^p ].
\end{equation}
\end{Lemma}

\begin{proof}
Our proof follows a coupling construction. Define random variables $V \sim T$ and $U_i \sim S$ 
for $i=1, \mydots, n$ such that $V$ and $(U_1, \mydots, U_n)$ are independent. 
Then $V | U_1, \mydots U_n \sim V \sim T$ by independence.
Let $\hat{S}_n$ denote the empirical distribution of $(U_1, \mydots, U_n)$. 
Define a random variable $U$ such that 
$U | U_1, \mydots U_n \sim \hat{S}_n$ and $(U,V)| U_1, \mydots U_n$ is a Wasserstein optimal coupling of $\hat{S}_n$ and $T$. Note that unconditionally $V \sim T$ and $U \sim S$ as $U_i \sim S$ for all $i=1, \mydots, n$. Therefore $(U,V)$ is a coupling of $S$ and $T$. We obtain, 
\begin{flalign}
\mathcal{W}_p(S,T)^p &\leq \mathbb{E}[c(U, V)^p] \text{ by the coupling representation of
Wasserstein distance} \\
&= \mathbb{E}[\mathbb{E}[c(U, V)^p| U_1, \mydots U_n]] \\
&= \mathbb{E}[\mathcal{W}_p(\hat{S}_n,T)^p].
\end{flalign}
\end{proof}

\begin{proof}[Proof of Proposition \ref{prop:empirical_wasserstein}:
Empirical Wasserstein distance bounds]
\item \paragraph{Upper bound.} 
Let $\hat{P}_n$ and $\hat{Q}_n$ denote the empirical distributions 
of the samples $(X_1, \mydots, X_n)$ and $(Y_1, \mydots, Y_n)$ respectively, where 
$X_i \sim P$, $Y_i \sim Q$ for all $i=1, \mydots, n$, and 
$(X_1, \mydots, X_n)$ and $(Y_1, \mydots, Y_n)$ are independent. By 
Lemma \ref{lemma:empirical_wass} with $S=P$, $U_i = X_i$ and $T=Q$, 
$$\mathcal{W}_p(P,Q)^p \leq \mathbb{E}[\mathcal{W}_p(\hat{P}_n,Q)^p].$$ 

As $(X_1, \mydots, X_n)$ and $(Y_1, \mydots, Y_n)$ are independent, 
$Y_i | (X_1, \mydots, X_n) \sim Y_i \sim Q$ for all $i=1, \mydots, n$. 
We can therefore apply Lemma \ref{lemma:empirical_wass} conditional on
$(X_1, \mydots, X_n)$ now with $S=Q$, $U_i = Y_i$ and $T=\hat{P}_n$ to
obtain 
$$ \mathcal{W}_p(\hat{P}_n,Q)^p \leq 
\mathbb{E}[\mathcal{W}_p(\hat{P}_n,\hat{Q}_n)^p | X_1, \mydots, X_n ]$$
almost surely for all $X_1, \mydots, X_n$. Overall, this gives 
$$ \mathcal{W}_p(P,Q)^p \leq \mathbb{E}[\mathcal{W}_p(\hat{P}_n,Q)^p] \leq 
\mathbb{E}[\mathbb{E}[\mathcal{W}_p(\hat{P}_n,\hat{Q}_n)^p | X_1, \mydots, X_n ]] 
= \mathbb{E}[\mathcal{W}_p(\hat{P}_n,\hat{Q}_n)^p]$$
as required.
\item \paragraph{Lower bound.} 
Let $\hat{P}_n$ and $\hat{Q}_n$ denote empirical distributions of the samples
$(X_1, \mydots, X_n)$ and $(Y_1, \mydots, Y_n)$ respectively, where $X_i \sim P$, 
$Y_i \sim Q$ for all $i=1, \mydots, n$. Given $(X_1, \mydots, X_n)$ and $(Y_1, \mydots, Y_n)$,
by the triangle inequality we obtain
\begin{equation}
\mathcal{W}_p(\hat{P}_n, \hat{Q}_n) \leq 
\mathcal{W}_p(\hat{P}_n, P) + \mathcal{W}_p(P, Q) + \mathcal{W}_p(Q, \hat{Q}_n).
\end{equation}
By Minkowski's inequality, this gives 
\begin{flalign}
\mathbb{E}[\mathcal{W}_p(\hat{P}_n, \hat{Q}_n)^p]^{1/p}
&\leq 
\mathbb{E} \Big[\Big( \mathcal{W}_p(\hat{P}_n, P) + \mathcal{W}_p(P, Q) + \mathcal{W}_p(Q, \hat{Q}_n) \Big)^p \Big]^{1/p} \\
&\leq \mathbb{E}[\mathcal{W}_p(\hat{P}_n, P)^p]^{1/p} + 
\mathbb{E}[\mathcal{W}_p(P, Q)^p]^{1/p} + 
\mathbb{E}[\mathcal{W}_p(Q, \hat{Q}_n)^p]^{1/p} \\
&= \mathbb{E}[\mathcal{W}_p(\hat{P}_n, P)^p]^{1/p} + 
\mathcal{W}_p(P, Q) + \mathbb{E}[\mathcal{W}_p(Q, \hat{Q}_n)^p]^{1/p} \label{eq:wass_midstep1}
\end{flalign}

Let $\tilde{P}_n$ denote empirical distributions of the samples
$(\tilde{X}_1, \mydots, \tilde{X}_n)$, where $\tilde{X}_i \sim P$ for all $i=1, \mydots, n$
and $(\tilde{X}_1, \mydots, \tilde{X}_n)$ and $({X}_1, \mydots, {X}_n)$ are independent. 
Independence implies $\tilde{X}_i | ({X}_1, \mydots, {X}_n) \sim \tilde{X}_i \sim P$ for all $i=1, \mydots, n$.
We can therefore apply Lemma \ref{lemma:empirical_wass} conditional on $({X}_1, \mydots, {X}_n)$,
with $S=P$, $T=\hat{P}_n$ and $\tilde{X}_i=U_i$ to obtain 
\begin{equation} \label{eq:empirical_wass_bound1}
\mathcal{W}_p(\hat{P}_n, P)^p \leq \mathbb{E}[\mathcal{W}_p(\hat{P}_n, \tilde{P}_n)^p | {X}_1, \mydots, {X}_n].
\end{equation}
Similarly, 
\begin{equation} \label{eq:empirical_wass_bound2}
\mathcal{W}_p(Q, \hat{Q}_n)^p \leq \mathbb{E}[\mathcal{W}_p(\tilde{Q}_n, \hat{Q}_n)^p | {Y}_1, \mydots, {Y}_n]
\end{equation}
where $\tilde{Q}_n$ denotes empirical distributions of the samples
$(\tilde{Y}_1, \mydots, \tilde{Y}_n)$, where $\tilde{Y}_i \sim Q$ for all $i=1, \mydots, n$
and $(\tilde{Y}_1, \mydots, \tilde{Y}_n)$ and $({Y}_1, \mydots, {Y}_n)$ are independent. 
By \eqref{eq:wass_midstep1}, we obtain 
\begin{flalign}
\mathbb{E}[\mathcal{W}_p(\hat{P}_n, \hat{Q}_n)^p]^{1/p} \leq &
\mathbb{E}[\mathcal{W}_p(\hat{P}_n, P)^p]^{1/p} + 
\mathcal{W}_p(P, Q) + \mathbb{E}[\mathcal{W}_p(Q, \hat{Q}_n)^p]^{1/p} \\
= & \mathbb{E}[ \mathbb{E}[ \mathcal{W}_p(\hat{P}_n, P)^p | {X}_1, \mydots, {X}_n ]  ]^{1/p} + 
\mathcal{W}_p(P, Q) + \\
&\mathbb{E}[ \mathbb{E}[ \mathcal{W}_p(Q, \hat{Q}_n)^p| {Y}_1, \mydots, {Y}_n ] ]^{1/p} \\
\leq & \mathbb{E}[ \mathbb{E}[ \mathcal{W}_p(\hat{P}_n, \tilde{P}_n)^p | {X}_1, \mydots, {X}_n ]  ]^{1/p} + 
\mathcal{W}_p(P, Q) + \\
&\mathbb{E}[ \mathbb{E}[ \mathcal{W}_p(\tilde{Q}_n, \hat{Q}_n)^p| {Y}_1, \mydots, {Y}_n ] ]^{1/p} \\
= & \mathbb{E}[ \mathcal{W}_p(\hat{P}_n, {P}_n)^p ]^{1/p} + 
\mathcal{W}_p(P, Q) + \mathbb{E}[ \mathcal{W}_p(\tilde{Q}_n, \hat{Q}_n)^p ]^{1/p}
\end{flalign}
as required. 
\item \paragraph{Consistency.} 
By triangle inequality, 
\begin{equation}
    | \mathcal{W}_p(\hat{P}_n, \hat{Q}_n) - \mathcal{W}_p(P, Q) | \leq \mathcal{W}_p(\hat{P}_n, P) + \mathcal{W}_p(Q, \hat{Q}_n).
\end{equation}
Note that $P$, $Q$, $(\hat{P}_n)_{n \geq 0}$ and $(\hat{Q}_n)_{n \geq 0}$ 
all have finite moments of order $p$, and that $\hat{P}_n \Rightarrow P$ and $\hat{Q}_n \Rightarrow Q$
almost surely by the Glivenko–Cantelli theorem, 
where the empirical distribution moments of order $p$ also converge weakly. 
By completeness of the $p$-Wasserstein distance on the space of probability
measures with finite moments of order $p$ \citep[][Theorem 6.9]{villani2008optimal},
$\mathcal{W}_p(\hat{P}_n, P) \toas 0$ and $\mathcal{W}_p(Q, \hat{Q}_n) \toas 0$ 
as $n \rightarrow \infty$.
\end{proof}

\subsection{Proofs for comparison with the approach of Dobson et al.} \label{appendices:dobson}
To prove Proposition \ref{prop:dobson_bound}, we first outline the setup of 
\citet{dobson2021usingSIAM}.  Consider a continuous time diffusion with a 
unique stationary distribution $P$ on $\mathbb{R}^d$.  
Let $K_1$ and $K_2$ denote the Markov chain
transition kernels corresponding to a discretization of this diffusion with and without an 
accept-reject bias correction step respectively.  For example, $K_1$ and $K_2$ 
can be the (single or multiple step) transition kernels of an MALA and an ULA Markov chain respectively. 
Suppose the marginal Markov chains with kernels $K_1$ and $K_2$ converge
in distribution to the unique invariant distributions $P$ and $Q$ respectively. 

For some small $\epsilon>0$, suppose there is a compact subset
$\Omega$ of $\mathbb{R}^d$ such that $P(\Omega^c)<\epsilon$ and 
$Q(\Omega^c)<\epsilon$.
For the capped metric $c(x,y) = \min \{1, \|x-y\|_2 \}$ on $\mathbb{R}^d$,
suppose there exists a Markovian coupling $\Gamma_1$ of the kernel $K_1$ such that 
for some constant $\alpha_{\Omega} \in (0,1)$ and all $X_t, X'_t \in \Omega$,
$\mathbb{E}[ c(X_{t+1}, X'_{t+1}) | X_t, X'_t ] \leq \alpha_{\Omega} c(X_t, X'_t)$ for
$(X_{t+1}, X'_{t+1}) | (X_t, X'_t) \sim \Gamma_2(X_t, X'_t)$. Under such assumptions, 
\citet{dobson2021usingSIAM} show
\begin{equation} \label{ref:dobson_ub}
\mathcal{W}_1(P,Q) \leq \frac{\mathbb{E}[\mathbb{E}[c(X_1,Y_1) | Y^*]] + 
2\epsilon}{1 - \alpha_{\Omega}}
\end{equation}
where $Y^* \sim Q$ and $(X_1, Y_1)|Y^* \sim \Gamma_{\Delta}(Y^*)$ for some fixed coupling 
$\Gamma_{\Delta}(Y^*)$ such that $X_1|Y^* \sim K_1(Y^*, \cdot)$ and $Y_1|Y^* \sim K_2(Y^*, \cdot)$ marginally.
\citet{dobson2021usingSIAM} then estimate the quantities $\mathbb{E}[\mathbb{E}[c(X_1,Y_1) | X^*]]$
and $\alpha_{\Omega}$ separately using couplings to obtain a final upper bound estimate.

Given this setup,  we can show that our upper bound estimator 
$\cub[1]$ \eqref{eq:W_ublimit1} constructed using such
couplings $\Gamma_1$ and $\Gamma_{\Delta}$
has a smaller expected value than the upper bound of \eqref{ref:dobson_ub}.  

\begin{proof}[Proof of Proposition \ref{prop:dobson_bound}:
\cubname lower bounds Dobson et al.]
We proceed as in the proofs of Theorem \ref{thm:Wp_UB_single_step} and
Corollary \ref{cor:Wp_UB_single_step}.  Consider the coupling based estimator in 
\eqref{eq:W_ublimit1} for the $1$-Wasserstein distance with metric 
$c(x,y) = \min \{1, \|x-y\|_2 \}$ on $\mathbb{R}^d$.
Under the coupled kernel $\bar{K}$ from Algorithm \ref{algo:coupled_kernel_single_step}, 
for each $t \geq 1$ we have the coupling $(X_{t},Z_{t},Y_{t})$
where $(X_{t},Z_{t}) | X_{t-1}, Y_{t-1} \sim \Gamma_1(X_{t-1}, Y_{t-1})$ and
$(Z_{t}, Y_{t}) | X_{t-1}, Y_{t-1} \sim \Gamma_\Delta (Y_{t-1})$. This gives
\begin{flalign}
\mathbb{E} [c(X_{t}, Y_{t})] \leq & \mathbb{E} [ c(X_{t}, Z_{t})] + \mathbb{E} [c(Z_{t}, Y_{t})] \label{eq:dobson_apply1} \\ 
= & \mathbb{E} [c(X_{t}, Z_{t}) \mathrm{I}_{\{ X_{t-1} \in \Omega, Y_{t-1} \in \Omega \}^c} ] + 
\mathbb{E} [c(X_{t}, Z_{t}) \mathrm{I}_{\{ X_{t-1} \in \Omega, Y_{t-1} \in \Omega \}} ] +
\mathbb{E} [ c(Z_{t}, Y_{t})] \\
\leq & \mathbb{P}( \{ X_{t-1} \in \Omega^c \} \cup \{ Y_{t-1} \in \Omega^c \}) + 
\mathbb{E} [c(X_{t}, Z_{t}) \mathrm{I}_{\{ X_{t-1} \in \Omega, Y_{t-1} \in \Omega \}} ] + \\
&\mathbb{E} [c(Z_{t}, Y_{t})] \label{eq:dobson_apply2} \\
\leq & \mathbb{P}(X_{t-1} \in \Omega^c) + \mathbb{P}( Y_{t-1} \in \Omega^c ) + 
\alpha_\Omega \mathbb{E} [c(X_{t-1}, Y_{t-1}) ] + \mathbb{E} [c(Z_{t}, Y_{t})] \label{eq:dobson_apply3} 
\end{flalign}
where \eqref{eq:dobson_apply1} follows by the triangle inequality,
\eqref{eq:dobson_apply2} follows as $c$ is bounded by $1$, and
\eqref{eq:dobson_apply3} follows by the union bound and the definition of 
$\alpha_{\Omega}$. Denote $\Delta(z) \defeq \mathbb{E}[c(X,Y)^p|z]$ for $(X,Y) \sim \Gamma_\Delta(z)$,
such that $\mathbb{E} [c(Z_{t}, Y_{t})] = \mathbb{E} [ \mathbb{E} [c(Z_{t}, Y_{t})|Y_{t-1}]] = \mathbb{E}[\Delta(Y_{t-1})] $. 
Then by induction, \eqref{eq:dobson_apply3} implies
\begin{flalign}
\mathbb{E} [ c(X_{t}, Y_{t}) ]
&\leq \alpha_\Omega^t \mathbb{E} [ c(X_0, Y_0) ] + 
\sum_{i=1}^t \alpha_\Omega^{t-i} \Big(  \mathbb{P}(X_{t-1} \in \Omega^c) + \mathbb{P}( Y_{t-1} \in \Omega^c ) + \mathbb{E} [ \Delta(Y_{i-1}) ]\Big).
\end{flalign}
As $X_{t-1}$ and $Y_{t-1}$ converges to $P$ and $Q$ respectively in distribution, 
$ \mathbb{P}(X_{t-1} \in \Omega^c) \overset{t \rightarrow \infty}{\rightarrow} P(\Omega^c) < \epsilon $,
$ \mathbb{P}(Y_{t-1} \in \Omega^c) \overset{t \rightarrow \infty}{\rightarrow} Q(\Omega^c)  < \epsilon $ and 
$ \mathbb{E} [ \Delta(Y_{t}) ]  \overset{t \rightarrow \infty}{\rightarrow} \mathbb{E} [ \Delta(Y^*) ]$
for $Y^* \sim Q$.  
Following the argument in Corollary \ref{cor:Wp_UB_single_step} we obtain that for all
$\epsilon'>0$, there exists some $S \geq 1$ such that for all $t \geq S$, 
\begin{flalign}
\mathbb{E} [ c(X_{t}, Y_{t}) ]
&\leq \alpha_\Omega^t \mathbb{E} [ c(X_0, Y_0) ] + 
\sum_{i=1}^t \alpha_\Omega^{t-i} \Big( \mathbb{E} [ \Delta(Y^*)] + 2 \epsilon \Big) + \epsilon'.
\end{flalign}
Therefore as $\alpha_\Omega \in (0,1)$,  
$\liminf_{t \rightarrow \infty} \mathbb{E} [ c(X_{t}, Y_{t}) ] \leq \frac{\mathbb{E}[\Delta(Y^*)]+2\epsilon}{1 - \alpha_\Omega}$
where $ \Delta(Y^*)=\mathbb{E}[c(X_1,Y_1) | Y^*]$  and $Y^* \sim Q$ from \eqref{ref:dobson_ub}
as required.
\end{proof}

\section{Example applications of theoretical results} \label{appendices:theory_examples}
In this section we consider the theoretical results of Section \ref{subsection:theory} applied to 
three simple examples, working with the metric $c(x,y) = \| x - y \|_2$.  
\paragraph{MALA and ULA.} Consider a MALA chain and an ULA chain with a 
common step size $\sigma$ both targeting a distribution $P$. 
Assume the negative log density of $P$ is gradient Lipschitz and strongly convex. 
In this setting, let $(X_t, Y_t)_{t \geq 0}$ be a CRN coupling of ULA and MALA simulated using 
Algorithm \ref{algo:coupled_chain_general}, such that the Markov chains $(X_t)_{t \geq 0}$ and 
$(Y_t)_{t \geq 0}$ marginally correspond to ULA and MALA respectively. 
For $\sigma$ sufficiently small, the marginal ULA chain $(X_t)_{t \geq 0}$
converges to some distribution $P_\sigma$ and
satisfies Assumption \ref{asm:gammaP_uniform_contract} for $p=2$
under a CRN coupling \cite[Proposition 3]{durmus2019highBERNOULLI},
giving a contraction rate $\rho$ such that $1-\rho = C \sigma^2/2$ 
for some constant $C$ which depends on the
gradient Lipschitz constant and convexity of the negative log density of $P$ rather 
than depending explicitly on the dimension of the state space. 
By Corollary \ref{cor:Wp_UB_single_step},
\begin{equation} \label{eq:mala_ula_bound}
\mathcal{W}_2(P_\sigma,P) \leq 
\liminf_{t \rightarrow \infty} \mathbb{E}[ \cub[2,t]^2 ]^{1/2} \leq 
\frac{\mathbb{E} \big[ \| Y - Y' \|^2 \big( 1-\alpha_{\sigma} \big(Y, Y' \big) \big) \big]^{1/2}}{C \sigma^2/2},
\end{equation}
where $Y \sim P$ is the limiting distribution of the MALA chain, 
$Y' | Y \sim \mathcal{N}( Y + \frac{\sigma^2}{2} \nabla \log P (Y), \sigma^2 I_d)$
corresponds to the Euler–Maruyama discretization based proposal, 
and $\alpha_{\sigma} \big(Y, Y' \big) \in [0,1]$ is the Metropolis--Hastings acceptance probability. 
As the step size $\sigma$ tends to zero, the upper bound in \eqref{eq:mala_ula_bound} 
require further analysis of the MALA acceptance probabilities 
\citep{bou-rabeenonasymptotic2012IMAJNA, eberle2014errorAoAP} and could degenerate. 
Recently, discrete sticky couplings \citep{durmus2021discrete} have been developed 
for perturbed functional autoregressive processes, which produce stable upper
bounds on total variation and the Wasserstein distance in such limiting regimes.

\paragraph{ULA and ULA.} We can similarly consider two ULA chains 
with a common step size $\sigma$ targeting different distributions $P$ and $Q$. 
As above, assume both $\log P$ and $\log Q$ are gradient Lipschitz and strongly convex. 
In this setting, let $(X_t, Y_t)_{t \geq 0}$ be a CRN coupling of two ULA chains simulated using 
Algorithm \ref{algo:coupled_chain_general}, such that the Markov chains $(X_t)_{t \geq 0}$ and 
$(Y_t)_{t \geq 0}$ marginally correspond to ULA targeting distributions $P$ and $Q$ respectively. 
For $\sigma$ sufficiently small, the marginal chains $(X_t)_{t \geq 0}$ and $(Y_t)_{t \geq 0}$ 
converge to some distributions $P_\sigma$ and $Q_\sigma$ respectively. Both marginal chains 
also satisfy Assumption \ref{asm:gammaP_uniform_contract} for $p=2$ under a CRN coupling, with 
contraction rates $\rho_P$ and $\rho_Q$ such that $1- \rho_{P} = C_P \sigma^2/2$ and 
$1- \rho_{Q} = C_Q \sigma^2/2$ respectively for
some constants $C_P$ and $C_Q$ that do not explicitly depend on the dimension. 
By Corollary \ref{cor:Wp_UB_single_step}, this gives 
\begin{flalign} \label{eq:ula_ula_bound}
\mathcal{W}_2(P_\sigma, Q_\sigma) \leq 
\liminf_{t \rightarrow \infty} \mathbb{E}[ \cub[2](P_{t},Q_{t})^2 ]^{1/2} \leq 
\frac{\mathbb{E} \big[ \| \nabla \log P(Y_\sigma) - \nabla \log Q(Y_\sigma) \|^2 \big]^{1/2}}{C_P}
\end{flalign}
where $Y \sim Q_\sigma$. By symmetry, we can obtain a similar bound in terms of some random variable 
$X \sim P_\sigma$ and $C_Q$. 
As $\sigma$ approaches  zero, the numerator in \eqref{eq:ula_ula_bound} approaches the
square root of the Fisher divergence between distributions $Q$ and $P$, given by
$F(Q,P)\defeq\mathbb{E} [ \| \nabla \log P(Y) - \nabla \log Q(Y) \|^2 ]$ for $Y \sim Q$. Such link
between the Fisher divergence and the Wasserstein distance has been noted previously by considering
continuous-time Langevin diffusions (e.g., \citet{huggins2019scableAISTATS}). Finally, note that the
upper bound in \eqref{eq:ula_ula_bound} does not explicitly depend on dimension, 
highlighting that estimators based on our coupled chains may give upper bounds that remain 
informative in high dimensions.

\paragraph{ULA and SGLD.} 
Consider an ULA chain and a Stochastic gradient Langevin dynamics (SGLD) \citep{welling2011bayesianICML}
chain with a common step size $\sigma$ and both targeting a distribution $P$. The
SGLD chain is based on unbiased estimates of the gradient of the log density of $P$, such that 
$\widehat{\nabla \log P}_{SGLD}(z) = \nabla \log P(z) + e_{SGLD}(z)$ for all $z \in \mathcal{X}$, 
where $e_{SGLD}(z)$ is mean zero error. 
We assume this error is bounded such that 
$\delta^2 \defeq \sup_{z \in \mathcal{X}}  e_{SGLD}(z)/(1+V(z)^2) < \infty $, 
for some $2^{nd}$-order Lyapunov function $V$ as in Proposition \ref{prop:Wp_UB_single_step_lyapunov}
and that the negative log density of $P$ is gradient Lipschitz and strongly convex. 
In this setting, let $(X_t, Y_t)_{t \geq 0}$ be a CRN coupling of ULA and SGLD simulated using 
Algorithm \ref{algo:coupled_chain_general}, such that the Markov chains $(X_t)_{t \geq 0}$ and 
$(Y_t)_{t \geq 0}$ marginally correspond to ULA and SGLD with marginal distributions 
$(P^{(ULA)}_t)_{t \geq 0}$ and $(P^{(SGLD)}_t)_{t \geq 0}$ respectively. 
For $\sigma$ sufficiently small, the marginal ULA chain $(X_t)_{t \geq 0}$ 
satisfies Assumption \ref{asm:gammaP_uniform_contract} for $p=2$
under a CRN coupling, giving a contraction rate $\rho$ such that $1-\rho = C \sigma^2/2$ for
constants $C$ that does not explicitly depend on the dimension. Then by 
Proposition \ref{prop:Wp_UB_single_step_lyapunov}, 
\begin{flalign} \label{eq:ula_sgld_bound}
\limsup_{t \rightarrow \infty} \mathcal{W}_2 \big( P^{(ULA)}_t, P^{(SGLD)}_t \big) \leq 
\liminf_{t \rightarrow \infty} \mathbb{E} \Big[ \cub[2] \big(P^{(ULA)}_{t},Q^{(ULA)}_{t} \big)^2 \Big]^{1/2} 
\leq \frac{\delta \kappa}{C}.
\end{flalign}
Note that the upper bound in \eqref{eq:ula_sgld_bound} does not explicitly depend on dimension,
and approaches zero as $\delta$ approaches zero.
This shows that estimators based on our coupled chains give upper bounds which 
may remain informative in high dimensions and are tight with respect to
the error from the stochastic gradients.
This example also highlights the stability of our upper bounds even when one of the marginal chains
(SGLD) may not converge to a limiting distribution. 

\section{Multi-step couplings} \label{appendices:multi_step}
In this section, we consider coupling algorithms for multi-step kernels 
and investigate their theoretical properties.

\subsection{Coupling algorithms for multi-step kernels} \label{subsection:multi_step_algo}
Consider the $L$-step Markov chains $(X_{Lt})_{t \geq 0}$ and $(Y_{Lt})_{t \geq 0}$
for $L \geq 1$, corresponding to marginal multi-step Markov kernels $K^L_P$ and $K^L_Q$ respectively. 
Following \eqref{eq:joint_kernel} and Section \ref{subsection:algo_k_bar}, 
we now construct a kernel $\bar{K}_{L-step}$ on the joint space 
$\mathcal{X} \times \mathcal{X}$ such that for all $x,y \in \mathcal{X}$ and all $A \in \mathcal{B}(\mathcal{X})$,
\begin{equation} \label{eq:joint_kernel_multi_step}
\bar{K}_{L-step}\big( (x, y), (A, \mathcal{X}) \big) = K^L_P(x, A) \text{ and } \bar{K}_{L-step}\big( (x, y), (\mathcal{X}, A) \big) = K^L_Q(y, A).
\end{equation}

Given coupled kernels $\Gamma_1$ and $\Gamma_\Delta$, 
Figure \ref{fig:multiple_step_coupling} illustrates how to sample from the joint kernel 
$\bar{K}_{L-step}$. By construction, this gives the marginal distributions 
$X_{s} | X_{0}, Y_{0} \sim K^{s}_P(X_{0}, \cdot)$ and
$Y_{s} | X_{0}, Y_{0} \sim K^{s}_Q(Y_{0}, \cdot)$ for all $s=1,\mydots,L$, such that Equation 
\eqref{eq:joint_kernel_multi_step} is satisfied. 
Algorithm \ref{algo:coupled_kernel_multi_step} samples from this coupled kernel $\bar{K}_{L-step}$. 
It characterizes the dependency between $X_{Lt}$ and $Y_{Lt}$ such that 
\begin{flalign}
X_{Lt} | X_{L(t-1)}, Y_{L(t-1)} &\sim K^L_P(X_{L(t-1)}, \cdot) \\ 
Z^{(j)}_{L} | Y_{L(t-1)+(j-1)} &\sim K^{L-(j-1)}_P( Y_{L(t-1)+(j-1)}, \cdot) \\ 
Y_{Lt} | X_{L(t-1)}, Y_{L(t-1)} &\sim K^L_Q(Y_{L(t-1)}, \cdot)
\end{flalign}
for $s=1,\mydots,L-1$. When $L=1$, we obtain $\bar{K}_{L-step} = \bar{K}$ from 
Algorithm \ref{algo:coupled_kernel_single_step}. 
Note that $\bar{K}_{1-step}$ is the single-step kernel $\bar{K}$ from Algorithm \ref{algo:coupled_kernel_single_step}, 
but $\bar{K}_{L-step}$ and $\bar{K}^L$ are not equivalent in general. 

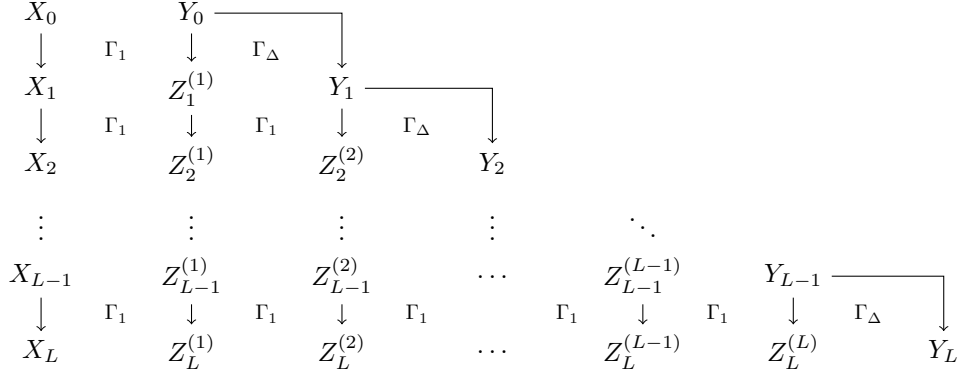
\begin{figure}
\centering
\begin{tikzpicture}
\node (X0) at (0, 0) {$X_{0}$};
\node (Y0) at (2,0) {$Y_{0}$};

\node (X1) at (0,-1) {$X_{1}$};
\node (Z11) at (2, -1) {$Z_{1}^{(1)}$};
\node (Y1) at (4, -1) {$Y_{1}$};

\node (X2) at (0,-2) {$X_{2}$};
\node (Z21) at (2, -2) {$Z_{2}^{(1)}$};
\node (Z22) at (4, -2) {$Z_{2}^{(2)}$};
\node (Y2) at (6, -2) {$Y_{2}$};

\node (Xdot) at (0,-2.75) {$\vdots$};
\node (Zdot1) at (2, -2.75) {$\vdots$};
\node (Zdot2) at (4, -2.75) {$\vdots$};
\node (Zdot3) at (6, -2.75) {$\vdots$};
\node (Ydot) at (8, -2.75) {$\ddots$};

\node (XL_minus_one) at (0,-3.5) {$X_{L-1}$};
\node (ZL_minus_one1) at (2, -3.5) {$Z_{L-1}^{(1)}$};
\node (ZL_minus_one2) at (4, -3.5) {$Z_{L-1}^{(2)}$};
\node (ZL_minus_onedot) at (6, -3.5) {$\ldots$};
\node (ZL_minus_oneL_minus_one) at (8, -3.5) {$Z_{L-1}^{(L-1)}$};
\node (YL_minus_one) at (10, -3.5) {$Y_{L-1}$};

\node (XL) at (0,-4.5) {$X_{L}$};
\node (ZL1) at (2, -4.5) {$Z_{L}^{(1)}$};
\node (ZL2) at (4, -4.5) {$Z_{L}^{(2)}$};
\node (ZLdot) at (6, -4.5) {$\ldots$};
\node (ZLL_minus_one) at (8, -4.5) {$Z_{L}^{(L-1)}$};
\node (ZLL) at (10, -4.5) {$Z_{L}^{(L)}$};
\node (YL) at (12, -4.5) {$Y_{L}$};

\draw[->]
  (X0) edge (X1) (Y0) edge (Z11);
\draw[->, to path={-| (\tikztotarget)}]
 (Y0) edge (Y1);
\node[scale=0.8] at (1, -0.5) {$\Gamma_1$};
\node[scale=0.8] at (3, -0.5) {$\Gamma_\Delta$};

\draw[->]
  (X1) edge (X2) (Z11) edge (Z21) (Y1) edge (Z22);
\draw[->, to path={-| (\tikztotarget)}]
 (Y1) edge (Y2);
\node[scale=0.8] at (1, -1.5) {$\Gamma_1$};
\node[scale=0.8] at (3, -1.5) {$\Gamma_1$};
\node[scale=0.8] at (5, -1.5) {$\Gamma_\Delta$};

\draw[->]
  (XL_minus_one) edge (XL) (ZL_minus_one1) edge (ZL1) 
  (ZL_minus_one2) edge (ZL2) (ZL_minus_oneL_minus_one) edge (ZLL_minus_one)
  (YL_minus_one) edge (ZLL);
\draw[->, to path={-| (\tikztotarget)}]
 (YL_minus_one) edge (YL);
\node[scale=0.8] at (1, -4) {$\Gamma_1$};
\node[scale=0.8] at (3, -4) {$\Gamma_1$};
\node[scale=0.8] at (5, -4) {$\Gamma_1$};
\node[scale=0.8] at (7, -4) {$\Gamma_1$};
\node[scale=0.8] at (9, -4) {$\Gamma_1$};
\node[scale=0.8] at (11, -4) {$\Gamma_\Delta$};
\end{tikzpicture}
\caption{Joint kernel $\bar{K}_{L-step}$ on $\mathcal{X} \times \mathcal{X}$, which couples the marginal kernels $K^L_P$ and $K^L_Q$}
\label{fig:multiple_step_coupling}
\end{figure}

\begin{algorithm}
\DontPrintSemicolon
\KwIn{chain states $X_{0}$ and $Y_{0}$, kernels $K_1$ and $K_2$, coupled kernels $\Gamma_1$ and $\Gamma_\Delta$} 
\For{s=1,\mydots,L}{
Sample 
\begin{equation}
(X_{s},Z^{(1)}_s,\mydots,Z^{(s)}_{s},Y_{s}) | (X_{s-1},Z^{(1)}_{s-1},\mydots,Z^{(s-1)}_{s-1},Y_{s-1}) 
\end{equation}
jointly such that
\begin{flalign}
(X_{s},Z^{(1)}_s) &\sim \Gamma_1(X_{s-1},Z^{(1)}_{s-1}) \label{eq:multi_step_algo_1} \\
(Z^{(j)}_s, Z^{(j+1)}_s) &\sim \Gamma_1(Z^{(j)}_{s-1}, Z^{(j+1)}_{s-1}) \text{ for } j=1,\mydots,s-1 \label{eq:multi_step_algo_2} \\
(Z^{(s)}_{s},Y_{s}) &\sim \Gamma_\Delta(Y_{s-1}) \label{eq:multi_step_algo_3} 
\end{flalign}
}
\Return $(X_{L(t-1)+s}, Y_{L(t-1)+s})$ for $s=1,\mydots,L$.
\caption{Joint kernel $\bar{K}_{L-step}$ on $\mathcal{X} \times \mathcal{X}$, which couples the marginal kernels $K^L_P$ and $K^L_Q$}
\label{algo:coupled_kernel_multi_step}
\end{algorithm}

We give concrete implementations of Algorithm \ref{algo:coupled_kernel_multi_step} for the ULA
and MALA Markov chain with common random numbers and reflection couplings. These are based on 
common random numbers and reflection couplings for the single-step coupling kernels
included in Appendices \ref{appendices:algos}. 

\paragraph{ULA with common random numbers coupling.}  For each $s=1,\ldots,L$ in Algorithm \ref{algo:coupled_kernel_multi_step}, sample $\epsilon_{s} \sim \mathcal{N}(0,I_d)$ and
\begin{itemize}
    \item Sample $(X_{s},Z^{(1)}_s) \sim \Gamma_1(X_{s-1},Z^{(1)}_{s-1})$ in \eqref{eq:multi_step_algo_1}
        such that 
        $X_s = X_{s-1} + \frac{1}{2} \sigma_P^2 \nabla \log p(X_{s-1}) + \sigma_P \epsilon_{s}$ and 
        $Z^{(1)}_s = Z^{(1)}_{s-1} + \frac{1}{2} \sigma_P^2 \nabla \log q(Z^{(1)}_{s-1}) + \sigma_P \epsilon_{s}$.
    \item Sample $(Z^{(j)}_s, Z^{(j+1)}_s) \sim \Gamma_1(Z^{(j)}_{s-1}, Z^{(j+1)}_{s-1})$ for each $j=1,\mydots,s-1$ in \eqref{eq:multi_step_algo_2} such that 
    $Z^{(j)}_s = Z^{(j)}_{s-1} + \frac{1}{2} \sigma_P^2 \nabla \log q(Z^{(j)}_{s-1}) + \sigma_P \epsilon_{s}$ and 
    $Z^{(j+1)}_s = Z^{(j+1)}_{s-1} + \frac{1}{2} \sigma_P^2 \nabla \log q(Z^{(j+1)}_{s-1}) + \sigma_P \epsilon_{s}$.
    \item Sample $(Z^{(s)}_{s},Y_{s}) \sim \Gamma_\Delta(Y_{s-1})$ in \eqref{eq:multi_step_algo_3} such that 
    $Z^{(s)}_s = Z^{(s)}_{s-1} + \frac{1}{2} \sigma_P^2 \nabla \log q(Z^{(s)}_{s-1}) + \sigma_P \epsilon_{s}$ and 
    $Y_s = Y_{s-1} + \frac{1}{2} \sigma_Q^2 \nabla \log p(Y_{s-1}) + \sigma_Q \epsilon_{s}$.
\end{itemize}

\paragraph{MALA with common random numbers coupling.} 
For $s=1,\ldots,L$ in Algorithm \ref{algo:coupled_kernel_multi_step},
sample $\epsilon_{s} \sim \mathcal{N}(0,I_d)$ and generate proposals $X^*_s,Z^{(1),*}_s,\ldots, Z^{(s),*}_s, Y^*_s$ using the steps for ULA with common random numbers coupling given above. Then 
sample $U^{(s)} \sim \Uniform ([0,1]) $ and 
accept each of these proposals if $U^{(s)}$ is less than the respective Metropolis-Hastings acceptance probabilities.

\paragraph{ULA with reflection coupling.}  
For each $s=1,\ldots,L$ in Algorithm \ref{algo:coupled_kernel_multi_step}, sample $\epsilon_{s} \sim \mathcal{N}(0,I_d)$ and 
\begin{itemize}
    \item Sample $(X_{s},Z^{(1)}_s) \sim \Gamma_1(X_{s-1},Z^{(1)}_{s-1})$ in \eqref{eq:multi_step_algo_1}
        such that 
        $X_s = X_{s-1} + \frac{1}{2} \sigma_P^2 \nabla \log p(X_{s-1}) + \sigma_P \epsilon_{s}$ and 
        $Z^{(1)}_s = Z^{(1)}_{s-1} + \frac{1}{2} \sigma_P^2 \nabla \log q(Z^{(1)}_{s-1}) + \sigma_P ( I_d - e^{(1)} e^{(1)\top} )  \epsilon_{s}$
        for $e^{(1)}=\frac{X_{s-1}-Z^{(1)}_{s-1} }{\| X_{s-1} - Z^{(1)}_{s-1} \|}$.
    \item Sample $(Z^{(j)}_s, Z^{(j+1)}_s) \sim \Gamma_1(Z^{(j)}_{s-1}, Z^{(j+1)}_{s-1})$ for each $j=1,\mydots,s-1$ in \eqref{eq:multi_step_algo_2} such that 
    $Z^{(j)}_s = Z^{(j)}_{s-1} + \frac{1}{2} \sigma_P^2 \nabla \log q(Z^{(j)}_{s-1}) + \sigma_P \epsilon_{s}$ and 
    $Z^{(j+1)}_s = Z^{(j+1)}_{s-1} + \frac{1}{2} \sigma_P^2 \nabla \log q(Z^{(j+1)}_{s-1}) + \sigma_P ( I_d - e^{(j+1)} e^{(j+1)\top} ) \epsilon_{s}$
    for $e^{(j+1)}=\frac{Z^{(j)}_{s-1}-Z^{(j+1)}_{s-1} }{\| Z^{(j)}_{s-1} - Z^{(j+1)}_{s-1} \|}$.
    \item Sample $(Z^{(s)}_{s},Y_{s}) \sim \Gamma_\Delta(Y_{s-1})$ in \eqref{eq:multi_step_algo_3} such that 
    $Z^{(s)}_s = Z^{(s)}_{s-1} + \frac{1}{2} \sigma_P^2 \nabla \log q(Z^{(s)}_{s-1}) + \sigma_P \epsilon_{s}$ and 
    $Y_s = Y_{s-1} + \frac{1}{2} \sigma_Q^2 \nabla \log p(Y_{s-1}) + \sigma_Q ( I_d - e^{(s+1)} e^{(s+1)\top} ) \epsilon_{s}$
    for $e^{(s+1)}=\frac{Z^{(s)}_{s-1}-Y_{s-1}}{\| Z^{(s)}_{s-1}-Y_{s-1} \|}$.
\end{itemize}

\paragraph{MALA with reflection coupling.} For $s=1,\ldots,L$ in Algorithm \ref{algo:coupled_kernel_multi_step}, sample sample $\epsilon_{s} \sim \mathcal{N}(0,I_d)$ and
generate proposals $X^*_s,Z^{(1),*}_s,\ldots, Z^{(s),*}_s, Y^*_s$ using the steps for ULA with reflection coupling given above. Then 
sample $U^{(s)} \sim \Uniform ([0,1]) $ and 
accept each of these proposals if $U^{(s)}$ is less than the respective Metropolis-Hastings acceptance probabilities.

\paragraph{} 
Having developed algorithms to sample from the coupled kernels $\bar{K}$
and $\bar{K}_{L-step}$, we now investigate theoretical properties our upper bounds. 

\subsection{Theoretical properties of couplings of multi-step kernels} \label{subsection:multi_step_theory}
To establish theoretical guarantees of coupled Markov chains based on the coupled kernel $\bar{K}_{L-step}$,  
we assume the Markovian coupling $\Gamma_1$ in Algorithm \ref{algo:coupled_kernel_multi_step} 
satisfies a geometric ergodicity condition. 

\begin{Asm} \label{asm:gammaP_geo_erg} 
There exists constants $C \in [1,\infty)$ and $\rho \in (0,1)$ such that for all $L \geq 1$,
\begin{equation} \label{eq:gammaP_geo_erg}
\mathbb{E}[ c(X_{t+L}, Y_{t+L})^p | X_t, Y_t ]^{1/p} \leq C \rho^L c(X_t, Y_t) \text{ for }
(X_{t+L}, Y_{t+L}) | (X_t, Y_t) \sim \Gamma^L_P(X_t, Y_t).
\end{equation}
\end{Asm}

Assumption \ref{asm:gammaP_geo_erg} is weaker than uniform contraction in Wasserstein’s distance as in 
Assumption \ref{asm:gammaP_uniform_contract}. 
Under Assumption \ref{asm:gammaP_geo_erg}, we now characterize the distance from our coupled 
chains based on the coupled kernel $\bar{K}_{L-step}$
explicitly in terms of the initial distribution $\bar{I_0}$ and the coupled kernel $\Gamma_\Delta$ 
corresponding to perturbations between the marginal kernels $K_1$ and $K_2$. 
At the heart of our analysis is the construction of the coupled kernel $\bar{K}_{L-step}$ given in
Figure \ref{fig:multiple_step_coupling} and Algorithm \ref{algo:coupled_kernel_multi_step}. 
When the coupled kernel $\Gamma_\Delta$ characterizing the perturbation
between the marginal kernels $K_1$ and $K_2$ is Wasserstein optimal, our analysis is linked to 
\citet{rudolf2018}, which only considers the $1$-Wasserstein distance and establishes similar 
results using analytic rather than probabilistic arguments. 

\begin{Th} \label{thm:Wp_UB_multi_step}
Let $(X_t, Y_t)_{t \geq 0}$ denote a coupled Markov chain
generated using Algorithm \ref{algo:coupled_chain_general} with 
initial distribution $\bar{I}_0$ and joint kernel $\bar{K}$ 
on $\mathcal{X} \times \mathcal{X}$ from Algorithm 
\ref{algo:coupled_kernel_single_step}. Suppose the coupled kernel $\Gamma_1$ satisfies 
Assumption \ref{asm:gammaP_geo_erg} for some $C \geq 1$ and $\rho<1$. Fix some $L \geq 1$
such that $\tilde{\rho} = C \rho^L < 1$, and consider the coupled chain $(X_{t}, Y_{t})_{t \geq 0}$ generated using 
Algorithm \ref{algo:coupled_kernel_multi_step} with the $L$-step coupled kernel $\bar{K}_{L-step}$. 
Then for all $t \geq 0$,
\begin{flalign}
\mathbb{E} [ c(X_{Lt}, Y_{Lt})^p ]^{1/p} 
&\leq \tilde{\rho}^t \mathbb{E} [ c(X_0, Y_0)^p ]^{1/p} + \sum_{i=1}^t \tilde{\rho}^{t-i} \Big( \sum_{j=1}^{L} C \rho^{L-j} \mathbb{E} \Big[ \Delta_p(Y_{L(i-1)+j}) \Big]^{1/p} \Big)
\end{flalign}
where $(X_0, Y_0) \sim \bar{I}_0$ and $\Delta_p(z) : = \mathbb{E}[ c(X,Y)^p]$ 
for $(X,Y)|z \sim \Gamma_\Delta(z)$.
\end{Th}

\begin{Cor} \label{cor:Wp_UB_multi_step}
Under the setup and assumptions of Theorem \ref{thm:Wp_UB_multi_step}, 
consider when the marginal distributions $Q_t$ converge in $p$-Wasserstein 
distance to some distribution $Q$ with finite moments of order $p$ as 
$t \rightarrow\infty$. 
Then for all $\epsilon >0$, there exists some $S \geq 1$ such that for all $t \geq S$,
\begin{equation}
\mathbb{E} [ c(X_{Lt}, Y_{Lt})^p ]^{1/p} \leq 
(C\rho^L)^t \mathbb{E} [ c(X_0, Y_0)^p ]^{1/p} + C \Big( \frac{1 - (C\rho^L)^t}{1-C\rho^L} \Big) \Big( \frac{1-\rho^L}{1-\rho} \Big) \mathbb{E}[\Delta_p(Y^*)]^{1/p} + \epsilon.
\end{equation}
where $(X_0, Y_0) \sim \bar{I}_0$, $\Delta_p(z) \defeq \mathbb{E}[ c(X,Y)^p|z]$ for $(X,Y) \sim \Gamma_\Delta(z)$ and $Y^* \sim Q$.
\end{Cor}

As in Section \ref{subsection:theory}, we can also upper bound the limiting distance from 
our coupled chains in terms of the perturbations between the marginal kernels weighted by a Lyapunov function of $K_2$.

\begin{Prop} \label{prop:Wp_UB_multi_step_lyapunov}
Under the setup and assumptions of Theorem \ref{thm:Wp_UB_multi_step}, 
let $V: \mathcal{X} \rightarrow [0, \infty)$ 
be a $p^{th}$-order Lyapunov function of $K_2$ such that 
\begin{equation} \label{eq:lyapunov_KQ}
\mathbb{E}[V(Y_{t+1})^p|Y_t=z] \leq \gamma V(z)^p + L 
\end{equation}
for all $z \in \mathcal{X}$, where $\gamma \in [0,1)$ and $L \in [0,\infty)$ are 
constants. Define
\begin{flalign} \label{eq:lyapunov_KQ_def2}
\delta \defeq \sup_{z \in \mathcal{X}} \bigg( \frac{\Delta_p(z)
}{1 + V(z)^p} \bigg)^{1/p} \qquad
\kappa \defeq 1 + \max \Big\{ \mathbb{E}[V(Y_{0})^p]^{1/p}, \Big(\frac{L}{1-\gamma}\Big)^{1/p} \Big\}  \Big).
\end{flalign}
where $\Delta_p(z) \defeq \mathbb{E}[ c(X,Y)^p|z]$ for $(X,Y) \sim \Gamma_\Delta(z)$. Then for all $t \geq 0$,
\begin{equation}
\mathbb{E}[ \cub[p,t]^p ]^{1/p} = 
\mathbb{E} [ c(X_t, Y_t)^p ]^{1/p} \leq 
(C\rho^L)^t \mathbb{E} [ c(X_0, Y_0)^p ]^{1/p} + C \Big( \frac{1 - (C\rho^L)^t}{1-C\rho^L} \Big) \Big( \frac{1-\rho^L}{1-\rho} \Big) \delta \kappa.
\end{equation}
\end{Prop}

\subsection{Proofs} \label{subsection:multi_step_proofs}
\begin{proof}[Proof of Theorem \ref{thm:Wp_UB_multi_step}]
Under the coupled kernel $\bar{K}_{L-step}$ from Algorithm \ref{algo:coupled_kernel_single_step}, 
for each $t \geq 1$ we obtain
\begin{equation}
(X_{Lt},Z^{(1)}_L,\mydots,Z^{(L)}_{L},Y_{Lt})
\end{equation}
where
\begin{flalign}
(X_{Lt},Z^{(1)}_L) | X_{L(t-1)}, Y_{L(t-1)} &\sim \Gamma^L_P(X_{L(t-1)}, Y_{L(t-1)}) \\
(Z^{(j)}_L, Z^{(j+1)}_L) | Y_{L(t-1)+j-1} &\sim \Gamma_\Delta(Y_{L(t-1)+j-1}) \Gamma_1^{L-j} \text{ for } j=1,\mydots,L-1 \\
(Z^{(L)}_L, Y_{Lt}) | Y_{L(t-1)+L-1} &\sim \Gamma_\Delta(Y_{L(t-1)+L-1}).
\end{flalign}
As $(X_{Lt}, Z^{(0)}_{t}) | X_{L(t-1)}, Y_{L(t-1)} \sim \Gamma_1^L(X_{L(t-1)}, Y_{L(t-1)})$, we obtain
\begin{flalign}
\mathbb{E} [ c(X_{Lt}, Y_{Lt})^p ]^{1/p} = & \mathbb{E}[\mathbb{E}[c(X_{Lt}, Y_{Lt})^p |X_{L(t-1)}, Y_{L(t-1)}]]^{1/p}  \\
\leq & \mathbb{E}[ \mathbb{E}[ \big( c(X_{Lt}, Z^{(1)}_{L}) + c(Z^{(1)}_{L}, Y_{Lt}) \big)^p |X_{L(t-1)}, Y_{L(t-1)} ] ]^{1/p} \label{eq:Wp_UB_dim_free_apply1multi} \\
\leq & \mathbb{E}[ \mathbb{E}[ c(X_{Lt}, Z^{(1)}_{L})^p |X_{L(t-1)}, Y_{L(t-1)} ] ]^{1/p} + \\
& \mathbb{E}[ \mathbb{E}[ c(Z^{(1)}_{L}, Y_{Lt})^p |X_{L(t-1)}, Y_{L(t-1)} ] ]^{1/p} \label{eq:Wp_UB_dim_free_apply2multi} \\
\leq & \tilde{\rho} \mathbb{E}[ c(X_{L(t-1)}, Y_{L(t-1)})^p ]^{1/p}  + \mathbb{E}[ c(Z^{(1)}_{L}, Y_{Lt})^p]^{1/p}\label{eq:Wp_UB_dim_free_apply3multi}
\end{flalign}
where \eqref{eq:Wp_UB_dim_free_apply1multi} follows as $c$ is a metric,
\eqref{eq:Wp_UB_dim_free_apply2multi} follows by Minkowski's inequality, and
\eqref{eq:Wp_UB_dim_free_apply3multi} follows by Assumption \ref{asm:gammaP_geo_erg}. 
Denote $\Delta_p(z) \defeq \mathbb{E}[c(X,Y)^p|z]$ for $(X,Y) \sim \Gamma_\Delta(z)$. Then, 
\begin{flalign}
\mathbb{E}[c(Z^{(1)}_{L}, Y_{Lt})^p]^{1/p} &\leq 
\mathbb{E} \Big[ \Big( c(Z^{(L)}_{L}, Y_{Lt}) + \sum_{j=1}^{L-1} c(Z^{(j)}_{L}, Z^{(j+1)}_{L}) \Big)^p \Big]^{1/p}
\text{ as $c$ is a metric}\\
&\leq \mathbb{E} \Big[ c(Z^{(L)}_{L}, Y_{Lt})^p \Big]^{1/p}  + 
\sum_{j=1}^{L-1} \mathbb{E} \Big[ c(Z^{(j)}_{L}, Z^{(j+1)}_{L})^p \Big]^{1/p} 
\text{ by Minkowski's inequality} \\
&= \mathbb{E} \Big[ \mathbb{E} \Big[ c(Z^{(L)}_{L}, Y_{Lt})^p |Y_{L(t-1)+L-1} \Big] \Big]^{1/p}  + 
\sum_{j=1}^{L-1} \mathbb{E} \Big[ \mathbb{E} \Big[ c(Z^{(j)}_{L}, Z^{(j+1)}_{L})^p | Y_{L(t-1)+j-1} \Big] \Big]^{1/p} \\
&= \mathbb{E}[\Delta_p(Y_{L(t-1)+(L-1)}) ]^{1/p}  + 
\sum_{j=1}^{L-1} \mathbb{E} \Big[ \mathbb{E} \Big[ c(Z^{(j)}_{L}, Z^{(j+1)}_{L})^p | Y_{L(t-1)+j-1} \Big] \Big]^{1/p} \\
&\leq \mathbb{E}[\Delta_p(Y_{L(t-1)+(L-1)}) ]^{1/p}  + 
\sum_{j=1}^{L-1} C \rho^{L-j} \mathbb{E} \Big[ \Delta_p(Y_{L(t-1)+j-1}) \Big]^{1/p} 
\text{ by Assumption \ref{asm:gammaP_geo_erg}} \\
&\leq \sum_{j=1}^{L} C \rho^{L-j} \mathbb{E} \Big[ \Delta_p(Y_{L(t-1)+j}) \Big]^{1/p} 
\text{as $C \geq 1 $}.
\end{flalign}
Equation \eqref{eq:Wp_UB_dim_free_apply3multi} now gives 
\begin{flalign}
\mathbb{E} [ c(X_{Lt}, Y_{Lt})^p ]^{1/p} &\leq 
\tilde{\rho} \mathbb{E}[ c(X_{L(t-1)}, Y_{L(t-1)})^p ]^{1/p}  + \sum_{j=1}^{L} C \rho^{L-j} \mathbb{E} \Big[ \Delta_p(Y_{L(t-1)+j}) \Big]^{1/p}
\label{eq:Wp_UB_dim_free_apply3multi2}
\end{flalign}
By induction, \eqref{eq:Wp_UB_dim_free_apply3multi2} implies
\begin{flalign}
\mathbb{E} [ c(X_{Lt}, Y_{Lt})^p ]^{1/p} 
&\leq \tilde{\rho}^t \mathbb{E} [ c(X_0, Y_0)^p ]^{1/p} + \sum_{i=1}^t \tilde{\rho}^{t-i} \Big( \sum_{j=1}^{L} C \rho^{L-j} \mathbb{E} \Big[ \Delta_p(Y_{L(i-1)+j}) \Big]^{1/p} \Big)
\end{flalign}
as required.
\end{proof}

\begin{proof}[Proof of Corollary \ref{cor:Wp_UB_multi_step}]
Denote $a\defeq \mathbb{E}[ \Delta_p(Y^*)]^{1/p}$ for $Y^* \sim Q$ and 
$a_k \defeq \mathbb{E} [\Delta_p(Y_{k})]^{1/p}$ for $k \geq 0$. 
Then $a_k \overset{k \rightarrow \infty}{\rightarrow} a$, because 
$Q_t$ converges in $p$-Wasserstein distance to $Q$ as $t \rightarrow \infty$. 
This implies 
\begin{flalign}
\sum_{i=1}^t \tilde{\rho}^{t-i} \Big( \sum_{j=1}^{L} C \rho^{L-j} a_{L(i-1)+j} \Big)
\overset{t \rightarrow \infty}{\rightarrow} 
\sum_{i=1}^t \tilde{\rho}^{t-i} \Big( \sum_{j=1}^{L} C \rho^{L-j} a \Big).
\end{flalign} 
Therefore, for all $\epsilon>0$ there exists $S \geq 1$ such that for all $t \geq S$,
$\sum_{i=1}^t \tilde{\rho}^{t-i} \sum_{j=1}^{L} C \rho^{L-j} |a_{L(i-1)+j}-a| < \epsilon$. 
By Theorem \ref{thm:Wp_UB_multi_step}, 
\begin{flalign}
\mathbb{E} [ c(X_{Lt}, Y_{Lt})^p ]^{1/p} 
&\leq \tilde{\rho}^t \mathbb{E} [ c(X_0, Y_0)^p ]^{1/p} + \sum_{i=1}^t \tilde{\rho}^{t-i} \Big( \sum_{j=1}^{L} C \rho^{L-j} a_{L(i-1)+j} \Big) \\
&\leq \tilde{\rho}^t \mathbb{E} [ c(X_0, Y_0)^p ]^{1/p} 
+ \sum_{i=1}^t \tilde{\rho}^{t-i} \sum_{j=1}^{L} C \rho^{L-j} a
+ \sum_{i=1}^t \tilde{\rho}^{t-i} \Big( \sum_{j=1}^{L} C \rho^{L-j} |a_{L(i-1)+j}-a| \Big)  \\
&\leq \tilde{\rho}^t \mathbb{E} [ c(X_0, Y_0)^p ]^{1/p} 
+ \sum_{i=1}^t \tilde{\rho}^{t-i} \sum_{j=1}^{L} C \rho^{L-j} a + \epsilon \\
&= (C\rho^L)^t \mathbb{E} [ c(X_0, Y_0)^p ]^{1/p} 
+ C \Big( \frac{1 - (C\rho^L)^t}{1-C\rho^L} \Big) \Big( \frac{1-\rho^L}{1-\rho} \Big) a + \epsilon 
\end{flalign}
as required.
\end{proof}

\begin{proof}[Proof of Proposition \ref{prop:Wp_UB_multi_step_lyapunov}]
As $V$ is a  a $p^{th}$-order Lyapunov function of $K_2$, by induction 
\begin{flalign}
\mathbb{E}[V(Y_{j})^p] \leq \gamma^j \mathbb{E}[V(Y_{0})^p] + (1-\gamma^t)\frac{L}{1-\gamma}
\end{flalign}
for all $j \geq 0$. This gives
\begin{flalign}
\mathbb{E} \Big[ \Delta_p(Y_{j}) \Big]^{1/p} &\leq \delta \mathbb{E}[1+ V(Y_{j-1})^p]^{1/p} \\
&\leq \delta ( 1 + \mathbb{E}[V(Y_{j-1})^p]^{1/p} ) \\
&\leq \delta \bigg( 1 + \Big( \gamma^{t-1}\mathbb{E}[V(Y_{0})^p] + (1- \gamma^{t-1})\frac{L}{1 - \gamma} \Big)^{1/p} \bigg) \\
&\leq \delta \bigg( 1 + \max \Big\{ \mathbb{E}[V(Y_{0})^p]^{1/p}, \Big(\frac{L}{1-\gamma}\Big)^{1/p} \Big\} \bigg) \\
&= \delta	\kappa
\end{flalign}
for all $j \geq 0$. By Theorem \ref{thm:Wp_UB_multi_step}, we obtain
\begin{flalign}
\mathbb{E} [ c(X_{Lt}, Y_{Lt})^p ]^{1/p} 
&\leq \tilde{\rho}^t \mathbb{E} [ c(X_0, Y_0)^p ]^{1/p} + \sum_{i=1}^t \tilde{\rho}^{t-i} \Big( \sum_{j=1}^{L} C \rho^{L-j} \mathbb{E} \Big[ \Delta_p(Y_{L(i-1)+j}) \Big]^{1/p} \Big)
\\
&\leq \tilde{\rho}^t \mathbb{E} [ c(X_0, Y_0)^p ]^{1/p} + \sum_{i=1}^t \tilde{\rho}^{t-i} \Big( \sum_{j=1}^{L} C \rho^{L-j} \delta \kappa \Big)
\\
&\leq \tilde{\rho}^t \mathbb{E} [ c(X_0, Y_0)^p ]^{1/p} + C \Big( \frac{1 - (C\rho^L)^t}{1-C\rho^L} \Big) \Big( \frac{1-\rho^L}{1-\rho} \Big) \delta \kappa
\end{flalign}
\end{proof}

\section{Details for the practical applications in Section \ref{section:applications}} \label{appendices:applications}
In this section, we provide details of the datasets, algorithms and parameters used 
for the three practical applications in Section \ref{section:applications}.
Open-source R code \citep{Rsoftware} recreating all experiments in this paper can be found at
\if1\blind
{\url{github.com/niloyb/BoundWasserstein}.} \fi
\if0\blind 
{\url{github.com/AnonPaperCode/BoundWasserstein}.} \fi

\subsection{Approximate MCMC and variational inference for tall data} \label{appendices:tall_data}
Section \ref{subsection:tall_data} considers Bayesian logistic regression
with a Gaussian prior applied to the Pima Diabetes dataset \citep{smith1988using} and the 
DS1 life sciences dataset \citep{ds1dataset}. 
The Pima Diabetes dataset has $n=768$ binary observations
(corresponding to the presence of diabetes), and 
$d=8$ covariates (containing information such as body mass index, 
insulin level and age), and is publicly available on 
\url{kaggle.com/uciml/pima-indians-diabetes-database}. 
The DS1 life sciences dataset has $n=26732$ binary observations
(corresponding to reactivity of the compound observed in a life 
sciences experiment), and $d=10$ covariates (containing information about the 
inputs to the life sciences experiment), and is publicly available on 
\url{komarix.org/ac/ds/} (ds1.10 file). 

In Figure \ref{fig:tall_data}, the upper bounds are given by 
our estimator $\cub[2]$ \eqref{eq:W_ublimit1} with
$S=1000,  T=2000,$ and $I=100$ for the Pima dataset and 
$S=500,T=100,$ and $I=40$ for the DS1 dataset, 
where these values were chosen based on initial runs.
The lower bounds are estimated using \eqref{eq:W2L2LB}
based on the same samples from the coupled chains used to 
calculate the upper bound estimate. For all the cases considered in 
Figure \ref{fig:tall_data}, we use a CRN coupling of the 
marginal kernels with a common step-size of $0.05$ for the 
Pima dataset and a common step-size of $0.05$ for the DS1 dataset.
We also considered switching between CRN 
and reflection couplings based on the Euclidean norm between the two chains. 
This did not produce tighter upper bounds than CRN in our experiments, 
but it may be effective in other examples, so  
we have included this option in our released code. 

\subsection{Approximate MCMC for high-dimensional linear regression} \label{appendices:half_t}
Section \ref{subsection:half_t} considers Bayesian linear regression
with the half-t global-local shrinkage prior applied to a bacteria genome-wide 
association study (GWAS) dataset \citep{buhlmann2014highARSA} 
and a synthetically generated dataset. The GWAS dataset
has $n = 71$ observations (corresponding to production of the vitamin riboflavin) and 
$d = 4088$ covariates (corresponding to single nucleotide 
polymorphisms (SNPs) in the genome) and is publicly available.
The synthetically generated dataset
has $n = 500$ observations and $d = 50000$ covariates. 
For the synthetic dataset, we generate 
$[X]_{i,j} \stackrel{\textup{i.i.d.}}{\sim} \mathcal{N}(0,1)$ 
and $y \sim \mathcal{N}(X \beta_{*} , \sigma^2_{*}I_n)$, where $\sigma_{*}=2$ and 
$\beta_{*} \in \mathbb{R}^d$ is chosen to be 
sparse such that $\beta_{*,j}=2^{(9-j)/4}$ for $1 \leq j \leq 20$ and $\beta_{*,j}=0$ for
all $j > 20$.

The state-of-the-art exact MCMC algorithms to sample from posteriors 
corresponding to the half-t prior are Gibbs samplers which cost 
$\mathcal{O}(n^2 d)$ per iteration.
This computation cost arises from a weighted matrix product 
calculation of the form $X \Diag(\eta_t)^{-1} X^\top $ where $\eta_t \in [0,\infty)^p$ 
corresponds to the local scale parameters which take different values at each iteration $t$. 
For the Horseshoe prior (degrees of freedom $\nu$=1), approximate MCMC methods have been 
developed by 
\citet{johndrow2020scalableJMLR} based on approximations of the form
\begin{equation} \label{eq:half_t_approx}
X \Diag(\xi \eta_t)^{-1} X^\top  \approx X\ \text{Diag}( ( \xi^{-1} \eta_j^{-1} \mathrm{I}_{ \{ \xi^{-1} \eta_j^{-1} > \epsilon \} })_{j=1}^p )\ X^\top 
\end{equation}
for some small threshold $\epsilon>0$.
\citet{biswas2021coupled} extended the exact marginal chain of \citep{johndrow2020scalableJMLR}
to all degrees of freedom $\nu \geq 1$.

In Section \ref{subsection:half_t}, we use couplings to assess the 
quality of the approximate MCMC algorithm characterized by the approximation in \eqref{eq:half_t_approx} for $\nu=2$. 
The upper bounds in Figure \ref{fig:half_t} are given by our estimator $\cub[2]$ 
\eqref{eq:W_ublimit1}. 
We take $S=1000,  T=3000,$ and $I=100$ for both datasets, where these values were chosen 
based on initial runs and the coupling-based
convergence assessment of the exact chain from \citet{biswas2021coupled}. 
The lower bounds in Figure \ref{fig:half_t} are estimated using \eqref{eq:W2L2LB}
based on same samples from the coupled chains used to 
calculate the upper bound estimate. 
We consider a CRN coupling with one marginal chain 
corresponding to the exact MCMC kernel and the other chain corresponding 
to the approximate MCMC kernel. 
The CRN coupled kernel is given in Algorithm \ref{algo:Half_t_crn_coupling}.

\begin{algorithm}
\LinesNumbered
\SetKwBlock{Begin}{Begin}{}
\SetAlgoLined
\SetKwProg{Loop}{LOOP}{}{}
\DontPrintSemicolon
\KwIn{exact chain current state $ C_t \defeq (\beta_t, \eta_t, \sigma_t^2, \xi_t) \in \mathbb{R}^d \times \mathbb{R}^d_{> 0} \times \mathbb{R}_{> 0} \times \mathbb{R}_{> 0}$, approximate chain current state $ \tilde{C}_t \defeq (\tilde{\beta}_t, \tilde{\eta}_t, \tilde{\sigma}_t^2, \tilde{\xi}_t) \in \mathbb{R}^d \times \mathbb{R}^d_{> 0} \times \mathbb{R}_{> 0} \times \mathbb{R}_{> 0}$ and 
approximation threshold $\epsilon>0$.} 
\begin{enumerate}
\item Sample $(\eta_{t+1}, \tilde{\eta}_{t+1}) \big| \xi_{t},\tilde{\xi}_{t}, \sigma^2_{t}, \tilde{\sigma}^2_{t},\beta_{t}, \tilde{\beta}_{t}$ component-wise, for each component $j$ targeting 
\begin{flalign} \label{eq:eta_full_conditionals}
\pi(\eta_{t+1,j} | \mydots ) &\propto \frac{e^{-m_{t,j} \eta_{t+1,j}}}{\eta_{t+1,j}^{\frac{1-\nu}{2}}(1+\nu \eta_{t+1,j})^{\frac{\nu + 1}{2}}} \text{ and }
\pi(\tilde{\eta}_{t+1,j} | \mydots ) \propto \frac{e^{-\tilde{m}_{t,j} \eta_{t+1,j}}}{\eta_{t+1,j}^{\frac{1-\nu}{2}}(1+\nu \eta_{t+1,j})^{\frac{\nu + 1}{2}}}
\end{flalign} 
for $m_{t,j} \defeq \big(\xi_{t} \beta_{t,j}^2\big)/\big(2  \sigma^2_{t}\big)$ and $\tilde{m}_{t,j} \defeq \big(\tilde{\xi}_{t} \tilde{\beta}_{t,j}^2\big)/\big(2 \tilde{\sigma}^2_{t}\big)$ respectively using common random numbers. This can be done using the slice sampler of \citet{biswas2021coupled}.

\item Sample $(\xi_{t+1},\tilde{\xi}_{t+1}, \sigma^2_{t+1}, \tilde{\sigma}^2_{t+1},\beta_{t+1}, \tilde{\beta}_{t+1})$ given $\eta_{t+1}$ and 
$\tilde{\eta}_{t+1}$ as follows:
\begin{enumerate}
\item Sample $(\xi_{t+1},\tilde{\xi}_{t+1})$ via Metropolis-Hastings with step size $\sigma_{\text{MH}}=0.8$: \;
Propose $\log(\xi^*) = \log(\xi_{t}) + \sigma_{\text{MH}} Z^*$ and $\log(\tilde{\xi}^*) = \log(\tilde{\xi}_{t}) + \sigma_{\text{MH}} Z^*$ for $Z^* \sim \mathcal{N}(0, 1)$. \;
Calculate acceptance probabilities
\begin{equation}
q =  \frac{L(y | \xi_*, \eta_{t+1} ) \pi_{\xi}(\xi_*)  }{L(y | \xi_{t}, \eta_{t+1} ) \pi_{\xi}(\xi_{t})} \frac{\xi^*}{\xi_{t}} \text{ and } 
\tilde{q} =  \frac{L(y | \tilde{\xi}_*, \tilde{\eta}_{t+1} ) \pi_{\xi}(\tilde{\xi}_*)  }{L(y | \tilde{\xi}_{t}, \tilde{\eta}_{t+1} ) \pi_{\xi}(\tilde{\xi}_{t})} \frac{\tilde{\xi}^*}{\tilde{\xi}_{t}}
\end{equation}
where $\pi_{\xi}(\cdot) $ is the prior density of $\xi$, 
$ M \defeq I_n + \xi_t^{-1} X\, \text{Diag}(\eta_{j,t}^{-1})\, X^\top $,
$ \tilde{M} \defeq I_n + X\, \text{Diag}( ( \tilde{\xi}_t^{-1} \tilde{\eta}_{j,t}^{-1} \mathrm{I}_{ \{ \tilde{\xi}_{\operatorname{max}}^{-1} \tilde{\eta}_{j,t}^{-1} > \epsilon \} })_{j=1}^p )\, X^\top $ for $\tilde{\xi}_{\operatorname{max}} = \max \{ \tilde{\xi}_t, \tilde{\xi}^* \}$, 
\begin{flalign} \log(L(y | \xi, \eta )) &=  - \frac{1}{2} \log (|M|) - \frac{a_0 + n}{2} \log( b_0 + y^\top M^{-1} y ) \text{ and} \\
\log(L(y | \xi, \eta )) &=  - \frac{1}{2} \log (|\tilde{M}|) - \frac{a_0 + n}{2} \log( b_0 + y^\top \tilde{M}^{-1} y ).
\label{eq:Half_t_Exact_algo_ll}
\end{flalign}
Sample $U^* \sim \Uniform([0,1])$. 
Set $\xi_{t+1} \defeq \xi^*$ if $U^* \leq \min(1,q)$, else set $\xi_{t+1} \defeq \xi_{t}$.
Set $\tilde{\xi}_{t+1} \defeq \tilde{\xi}^*$ if $U^* \leq \min(1,\tilde{q})$, else set $\tilde{\xi}_{t+1} \defeq \tilde{\xi}_{t}$.
\end{enumerate}
\end{enumerate}
\caption{Common random numbers coupling of an exact and an approximate
Markov chain for Bayesian regression with half-t priors.}
\label{algo:Half_t_crn_coupling}
\end{algorithm}

\begin{algorithm}
\setcounter{algocf}{1}
\caption{continued}
  \LinesNumbered
\setcounter{AlgoLine}{12}
  \SetAlgoVlined
\SetKwBlock{Begin}{}{end}
\SetKwProg{Loop}{LOOP}{}{}

\begin{enumerate}
\setcounter{enumi}{1}
    \item 
\begin{enumerate}[(a)]
\setcounter{enumii}{1}
\item Sample $(\sigma^2_{t+1}, \tilde{\sigma}^2_{t+1}) \big| \xi_{t+1},\tilde{\xi}_{t+1},\eta_{t+1}, \tilde{\eta}_{t+1}$ using 
common random numbers, marginally targeting
\begin{flalign}
\sigma^2_{t+1} &| \xi_{t+1}, \eta_{t+1} \sim \InvGamma\Big(\frac{a_0+n}{2}, \frac{y^\top M_{\xi_{t+1},\eta_{t+1}}^{-1} y + b_0}{2}\Big) \text{ and} \\
\tilde{\sigma}^2_{t+1} &| \tilde{\xi}_{t+1}, \tilde{\eta}_{t+1} \sim \InvGamma\Big(\frac{a_0+n}{2}, \frac{y^\top M_{\tilde{\xi}_{t+1},\tilde{\eta}_{t+1}}^{-1} y + b_0}{2}\Big).
\end{flalign}
\item Sample $(\beta_{t+1}, \tilde{\beta}_{t+1}) | \sigma^2_{t+1}, \tilde{\sigma}^2_{t+1}, \xi_{t+1}, \tilde{\xi}_{t+1},  \eta_{t+1}, \tilde{\eta}_{t+1}$ 
with common random numbers and the fast sampling algorithms of \citet{bhattacharya2016fastBIOMETRIKA}, marginally targeting
\begin{flalign}
\beta_{t+1} | \sigma^2_{t+1}, \xi_{t+1}, \eta_{t+1} \sim \mathcal{N} \big( \Sigma^{-1} X^\top y,  \sigma^2_{t+1} \Sigma^{-1} \big)
\text{ for } \Sigma = X^\top X + \xi_{t+1} \text{Diag}(\eta_{t+1} ) \\
\tilde{\beta}_{t+1} | \tilde{\sigma}^2_{t+1}, \tilde{\xi}_{t+1}, \tilde{\eta}_{t+1} \sim \mathcal{N} \big( \tilde{\Sigma}^{-1} X^\top y,  \tilde{\sigma}^2_{t+1} \tilde{\Sigma}^{-1} \big)
\text{ for } \tilde{\Sigma} = X^\top X + \tilde{\xi}_{t+1} \text{Diag}( \tilde{\eta}_{t+1} )
\end{flalign}
\end{enumerate}
\end{enumerate}
\Return 
$C_{t+1} \defeq (\beta_{t+1}, \eta_{t+1}, \sigma^2_{t+1}, \xi_{t+1})$ and 
$\tilde{C}_{t+1} \defeq (\tilde{\beta}_{t+1}, \tilde{\eta}_{t+1}, \tilde{\sigma}_{t+1}^2, \tilde{\xi}_{t+1})$.
\end{algorithm}

\subsection{Approximate MCMC for high-dimensional logistic regression} \label{appendices:skinny_gibbs}

Section \ref{subsection:skinny_gibbs} considers Bayesian logistic regression
with spike and slab priors applied to 
a malware detection dataset and a lymph node GWAS dataset.
The Malware detection dataset from the UCI machine learning repository \citep{Dua2019UCI} has $n=373$ observations (corresponding to a binary response vector indicating whether a file is malicious or non-malicious) and $d=503$ covariates (corresponding to features of the files), and 
is publicly available on \url{kaggle.com/piyushrumao/malware-executable-detection}. 
The lymph node GWAS dataset \citep{hans2007shotgunJASA, liang2013bayesianJASA, narisetty2019skinnyJASA} has $n=148$ observations (corresponding to a binary response 
vector indicating high or low risk status of the lymph node that is related to
breast cancer) and $d=4514$ covariates (corresponding to SNPs in the genome)
is not publicly available.  

The logistic regression likelihood is given by 
$L(\beta; y, X) = \prod_{i=1}^n (1+\exp(-y_i x_i^\top \beta))^{-1}$ where 
$y \in \{-1,1\}^n$ is the response vector, 
$X \in \mathbb{R}^{n \times d}$ is the scaled design matrix with rows $x_i^\top$, 
and $\beta \in \mathbb{R}^d$ is an unknown signal vector. 
The spike and slab prior is given by 
\begin{equation} \label{eq:spike_slab}
Z_j \overset{i.i.d.}{\sim} \Bernoulli(q), \qquad \beta_j | Z_j=0 \sim \mathcal{N}(0, \tau^2_{0}),
\qquad \beta_j | Z_j=1 \sim \mathcal{N}(0, \tau^2_{1})
\end{equation}
for $j=1, \mydots, d$ where $q \in (0,1), \tau_0>0,$ and $\tau_1>0$ are hyper-parameters
with $\tau_0 \ll \tau_1$ such that $Z_i=0$ and $Z_i=1$ correspond to null and 
non-null components of $\beta_j$ respectively.
By considering the posterior distribution of each variable $Z_j$ on 
$\{0,1\}$, spike and slab priors provide an interpretable method for 
Bayesian variable selection \citep[e.g.][]{george1993variableJASA, ishwaran2005spikeAOS, 
narisetty2014bayesianAOS}. 

The state-of-the-art exact MCMC algorithms to sample from posteriors
corresponding to the prior in \eqref{eq:spike_slab}
are Gibbs samplers which cost $\mathcal{O}(n^2 d)$ per iteration 
\citep{bhattacharya2016fastBIOMETRIKA}. 
\citet{narisetty2019skinnyJASA} have recently developed approximate MCMC methods for 
this setting. Their approximate MCMC algorithm, called \textit{Skinny Gibbs},
is based on matrix approximations of the form
\begin{equation}
\begin{pmatrix}
X_A^\top X_A + \tau_{1}^{-2}I & X_A^\top X_{A^c} \\
X_{A^c}^\top X_A & X_{A^c}^\top X_{A^c} + \tau_{0}^{-2}I \\
\end{pmatrix} \approx \begin{pmatrix}
X_A^\top X_A + \tau_{1}^{-2}I & 0 \\
0 & ((n-1) + \tau_{0}^{-2})I \\
\end{pmatrix}
\end{equation}
where $A= \{j: Z_{j}=1\}$, $X_A$ is an $n \times |A|$ matrix corresponding to the active columns 
$j \in A$ of the design matrix, and $X_{A^c}$ is an $n \times (d-|A|)$ matrix 
corresponding to the inactive columns $j \notin A$. 
This gives an overall computation cost of $\mathcal{O}(n \min \{d, |A|^2 \})$ per iteration. 

In Section \ref{subsection:skinny_gibbs}, we use couplings to assess the 
quality of the Skinny Gibbs algorithm.
The upper bounds in Figure \ref{fig:half_t} are given by our estimator $\cub[2]$ 
\eqref{eq:W_ublimit1} with $S=1000,  T=3000,$ and $I=100$ for both the malware and lymph node GWAS datasets, where these values were chosen based on initial runs. 
The lower bounds in Figure \ref{fig:half_t} are estimated using \eqref{eq:W2L2LB}
based on the same samples from the coupled chains used to 
calculate the upper bound estimate. 
We consider a CRN coupling between one marginal chain corresponding to the exact MCMC kernel 
and another corresponding to the Skinny Gibbs kernel. 
The CRN coupled kernel is given in Algorithm \ref{algo:skinny_gibbs_crn_coupling}.

\begin{algorithm}%
\DontPrintSemicolon
\KwIn{exact chain current state $C_t \defeq (\beta_t, z_t, e_t, w_t) \in \mathbb{R}^d \times \{0,1\}^d \times \mathbb{R}^n \times \mathbb{R}^n$ and \\ 
approximate chain current state $\tilde{C}_t \defeq (\tilde{\beta}_t, \tilde{z}_t, \tilde{e}_t, \tilde{w}_t) \in \mathbb{R}^d \times \{0,1\}^d \times \mathbb{R}^n \times \mathbb{R}^n$. } 
\begin{enumerate}
\item Sample $(\beta_{t+1}, \tilde{\beta}_{t+1}) | z_t, e_t, w_t, \tilde{z}_t, \tilde{e}_t, \tilde{w}_t$ 
with common random numbers and the fast sampling algorithms of \citet{bhattacharya2016fastBIOMETRIKA}, marginally targeting
\begin{flalign}
(\beta_{A,t+1}, \beta_{A^c,t+1}) | z_t,  e_t, w_t \sim \mathcal{N} \big( \Sigma^{-1} X^\top  W y,  \Sigma^{-1} \big) &\text{ for } 
\Sigma = 
\scriptsize 
\begin{pmatrix}
X_A^\top  W X_A + \tau_{1}^{-2}I & X_A^\top  W X_{A^c} \\
X_{A^c}^\top  W X_A & X_{A^c}^\top  W X_{A^c} + \tau_{0}^{-2}I \\
\end{pmatrix}
\normalsize
, \\
(\tilde{\beta}_{\tilde{A}, t+1}, \tilde{\beta}_{\tilde{I}, t+1}) | \tilde{z}_t,  \tilde{e}_t, \tilde{w}_t \sim \mathcal{N} \big( \tilde{\Sigma}^{-1} X^\top  \tilde{W} y,  \tilde{\Sigma}^{-1} \big)
&\text{ for } 
\scriptsize
\tilde{\Sigma}= \begin{pmatrix}
X_{\tilde{A}}^\top  \tilde{W} X_{\tilde{A}} + \tau_{1}^{-2} I & 0 \\
0 & ((n-1) + \tau_{0}^{-2}) I \\
\end{pmatrix}
\end{flalign}
\normalsize
where $W = \Diag( w_t )$ and $\tilde{W} = \Diag( \tilde{w}_t )$,  $A= \{j: z_{j,t}=1\}$ and $\tilde{A}= \{j: \tilde{z}_{j,t}=1\}$ are the index sets of active components,
$X_A$ and $X_{\tilde{A}}$ are matrices corresponding to the active (or inactive) columns of $X$ 
with columns $j \in A$ and $j \in \tilde{A}$ respectively,  
$\beta_{A, t+1}$ and $\tilde{\beta}_{\tilde{A}, t+1}$ are vectors of active components of $\beta_{t+1}$ and $\tilde{\beta}_{t+1}$ respectively.

\item Sample $(z_{t+1}, \tilde{z}_{t+1})$ given $\beta_{t+1}, \tilde{\beta}_{t+1}, e_t, \tilde{e}_t, w_t,  \tilde{w}_t$ with common random numbers sequentially in order for $j=1, \mydots, p$ such that each $z_{j,t+1}$ and $\tilde{z}_{j,t+1}$ are Bernoulli random variables with odds
\begin{flalign}
& \frac{q \mathcal{N}(\beta_{j,t+1}, 0, \tau_1^2)}{(1-q) \mathcal{N}(\beta_{j,t+1}, 0, \tau_0^2)} \text{ and } \\
&\frac{q \mathcal{N}(\tilde{\beta}_{j,t+1}, 0, \tau_1^2)}{(1-q) \mathcal{N}(\tilde{\beta}_{j,t+1}, 0, \tau_0^2)} 
\exp \Big( \tilde{\beta}_{j,t+1} X_j^\top  \tilde{W} ( Y - X_{C_j} \beta_{C_j, t+1} ) + \frac{1}{2}X_j^\top  (I-\tilde{W}) X_j \beta_{j,t+1}^2 \Big)& & &
\end{flalign}
respectively where $\mathcal{N}(\cdot; \mu, \Sigma)$ is the probability density of the normal distribution with mean $\mu$ and variance $\Sigma$,  $C_j \defeq \{k:  \tilde{z}_{k,t+1}=1 \text{ for } k <j \text{ or } \tilde{z}_{k,t}=1 \text{ for } k > j \}$ is the index set of active components in $\{1, \mydots, p\} \backslash \{j\}$, 
$X_{C_j}$ is a matrix of the columns of $X$ which correspond to indices in $C_j$,  and $\tilde{\beta}_{C_j, t+1}$ is a vector of the components of $\tilde{\beta}_{t+1}$
which correspond to indices in $C_j$. 
\end{enumerate}
\caption{Common random numbers coupling of an exact and an approximate
Markov chain for Bayesian logistic regression with spike and slab priors.}
 \label{algo:skinny_gibbs_crn_coupling}
\end{algorithm}

\begin{algorithm}
\setcounter{algocf}{2}
\caption{continued}
  \LinesNumbered
\setcounter{AlgoLine}{12}
  \SetAlgoVlined
\SetKwBlock{Begin}{}{end}
\SetKwProg{Loop}{LOOP}{}{}

\begin{enumerate}
\setcounter{enumi}{2}
\item Sample $(e_{t+1}, \tilde{e}_{t+1}) | \beta_{t+1}, \tilde{\beta}_{t+1}, z_{t+1}, \tilde{z}_{t+1}, w_t, \tilde{w}_t$ with common random numbers component-wise independently such that for each $i=1,\mydots, n$
\begin{flalign}
e_{i,t+1} &\sim
\left\{\begin{matrix}
 \mathcal{N}( x_i^\top  \beta_{t+1},  w_{i,t}^{-1}) \mathrm{I}_{[0, \infty)} & \text{if } y_i=1 \\
 \mathcal{N}( x_i^\top  \beta_{t+1},  w_{i,t}^{-1}) \mathrm{I}_{(-\infty, 0)} & \text{if } y_i=0  \\
\end{matrix}\right.
\text{ and}  \\
\tilde{e}_{i,t+1} &\sim
\left\{\begin{matrix}
 \mathcal{N}( x_{\tilde{A},i}^\top  \tilde{\beta}_{\tilde{A},t+1},  \tilde{w}_{i,t}^{-1}) \mathrm{I}_{[0, \infty)} & \text{if } y_i=1 \\
 \mathcal{N}( x_{\tilde{A},i}^\top  \tilde{\beta}_{\tilde{A},t+1},  \tilde{w}_{i,t}^{-1}) \mathrm{I}_{(-\infty, 0)} & \text{if } y_i=0  \\
\end{matrix}\right.
\end{flalign}
where $x_i^\top $ and $x_{\tilde{A}, i}^\top $ are the $i^{th}$ row of the $X$ and $X_{\tilde{A}}$ respectively.
\item Sample $(w_{t+1}, \tilde{w}_{t+1}) | \beta_{t+1}, \tilde{\beta}_{t+1}, z_{t+1}, \tilde{z}_{t+1}, e_{t+1}, \tilde{e}_{t+1}$.  We take this variable to be fixed,  and set $w_{i,t} = \tilde{w}_{i,t} = 3/\pi^2$ for all $i=1, \mydots, n$ and $t \geq 0$, where $3/\pi^2$ is the precision of the logistic distribution.  In the case this variable can vary,  they can be sampled using common random numbers such that for each $i=1, \mydots, n$, 
\begin{flalign}
w_{i,t+1} \sim  \Gamma \bigg( \frac{\nu+1}{2},  \frac{K (y_i- x_i^\top  \beta_{t+1})^2}{2} \bigg) \text{ and }  
\tilde{w}_{i,t+1} \sim  \Gamma \bigg( \frac{\nu+1}{2},  \frac{K(y_i- x_{\tilde{A},i}^\top  \tilde{\beta}_{\tilde{A},t+1})^2}{2} \bigg)
\end{flalign}
where $\nu=7.3$, $K \defeq (\pi^2 (\nu-2)/3)$ are fixed constants as given in \citet{narisetty2019skinnyJASA}. 
\end{enumerate}
\normalsize
\Return $C_t \defeq (\beta_t, z_t, e_t, w_t)$ and $C_t \defeq (\tilde{\beta}_t, \tilde{z}_t, \tilde{e}_t, \tilde{w}_t)$.
\end{algorithm}

\newpage

\section{Additional Algorithms} \label{appendices:algos}
\begin{algorithm}[!htb]
\DontPrintSemicolon
\KwIn{$(X_t, Y_t)$, unnormalized densities $p$ and $q$
of $P$ and $Q$ respectively, step sizes $\sigma_P$ and $\sigma_Q$} 
Sample $\epsilon_{CRN} \sim \mathcal{N}(0,I_d)$. Calculate proposals
$$X^* \defeq X_t + \frac{1}{2} \sigma_P^2 \nabla \log p(X_{t}) + \sigma_P \epsilon_{CRN} \text{ and } Y^* \defeq Y_t + \frac{1}{2} \sigma_Q^2 \nabla \log q(Y_{t}) + \sigma_Q \epsilon_{CRN}$$

Sample $U_{CRN} \sim \Uniform([0,1])$ \;
\textbf{if} $U_{CRN} \leq \frac{p(X^*) \mathcal{N}( X^* ; X_t + \frac{1}{2} \sigma_P^2 \nabla \log p(X_{t}), \sigma_P^2I_d )}{p(X_t) \mathcal{N}( X_t ; X^* + \frac{1}{2} \sigma_P^2 \nabla \log p(X^*), \sigma_P^2I_d )}  $, \textbf{then} set \( X_{t+1} = X^* \) ; \textbf{else} set \( X_{t+1} = X_{t} \) \;
\textbf{if} $U_{CRN} \leq \frac{q(Y^*) \mathcal{N}( Y^* ; Y_t + \frac{1}{2} \sigma_Q^2 \nabla \log q(Y_{t}), \sigma_Q^2I_d )}{q(Y_t) \mathcal{N}( Y_t ; Y^* + \frac{1}{2} \sigma_Q^2 \nabla \log q(Y^*), \sigma_Q^2I_d )}  $, \textbf{then} set \( Y_{t+1} = Y^* \) ; \textbf{else} set \( Y_{t+1} = Y_{t} \) \;
\Return $(X_{t+1}, Y_{t+1})$
\caption{Common random numbers coupling of two MALA kernels marginally 
 targetting distributions $P$ and $Q$ respectively}
 \label{algo:mala_crn}
\end{algorithm}

\begin{algorithm}[!htb]
\DontPrintSemicolon
\KwIn{$(X_t, Y_t)$, unnormalized densities $p$ and $q$
of $P$ and $Q$ respectively, step sizes $\sigma_P$ and $\sigma_Q$} 
Sample $\epsilon_{CRN} \sim \mathcal{N}(0,I_d)$. Calculate proposals
$$X^* \defeq X_t + \frac{1}{2} \sigma_P^2 \nabla \log p(X_{t}) + \sigma_P \epsilon_{CRN} \text{ and } Y^* \defeq Y_t + \frac{1}{2} \sigma_Q^2 \nabla \log q(Y_{t}) + \sigma_Q \epsilon_{CRN}. $$

Sample $U \sim \Uniform([0,1])$ \;
\textbf{if} $U \leq \frac{p(X^*) \mathcal{N}( X^* ; X_t + \frac{1}{2} \sigma_P^2 \nabla \log p(X_{t}), \sigma_P^2I_d )}{p(X_t) \mathcal{N}( X_t ; X^* + \frac{1}{2} \sigma_P^2 \nabla \log p(X^*), \sigma_P^2I_d )}  $, \textbf{then} set \( X_{t+1} = X^* \) ; \textbf{else} set \( X_{t+1} = X_{t} \) \;
Set \( Y_{t+1} = Y^* \) \;
\Return $(X_{t+1}, Y_{t+1})$
\caption{Common random numbers coupling of a MALA kernel and an ULA kernel 
marginally targeting distributions $P$ and $Q$ respectively}
 \label{algo:mala_ula_crn}
\end{algorithm}

\begin{algorithm}[!htb]
\DontPrintSemicolon
\KwIn{$(X_t, Y_t)$, unnormalized densities $p$ and $q$
of $P$ and $Q$ respectively, step sizes $\sigma_P$ and $\sigma_Q$} 
Sample $\epsilon \sim \mathcal{N}(0,I_d)$. Calculate proposals
\begin{flalign}
X^* &\defeq X_t + \frac{1}{2} \sigma_P^2 \nabla \log p(X_{t}) + \sigma_P \epsilon \\
Y^* &\defeq Y_t + \frac{1}{2} \sigma_Q^2 \nabla \log q(Y_{t}) + \sigma_Q ( I_d - e e^\top ) \epsilon
\text{  for } e = \frac{X_t-Y_t}{\|X_t-Y_t\|_2}.
\end{flalign}

Sample $U_{CRN} \sim \Uniform([0,1])$. \;
\textbf{if} $U_{CRN} \leq \frac{p(X^*) \mathcal{N}( X^* ; X_t + \frac{1}{2} \sigma_P^2 \nabla \log p(X_{t}), \sigma_P^2I_d )}{p(X_t) \mathcal{N}( X_t ; X^* + \frac{1}{2} \sigma_P^2 \nabla \log p(X^*), \sigma_P^2I_d )}  $, \textbf{then} set \( X_{t+1} = X^* \) ; \textbf{else} set \( X_{t+1} = X_{t} \). \;
\textbf{if} $U_{CRN} \leq \frac{q(Y^*) \mathcal{N}( Y^* ; Y_t + \frac{1}{2} \sigma_Q^2 \nabla \log q(Y_{t}), \sigma_Q^2I_d )}{q(Y_t) \mathcal{N}( Y_t ; Y^* + \frac{1}{2} \sigma_Q^2 \nabla \log q(Y^*), \sigma_Q^2I_d )}  $, \textbf{then} set \( Y_{t+1} = Y^* \) ; \textbf{else} set \( Y_{t+1} = Y_{t} \). \;
\Return $(X_{t+1}, Y_{t+1})$
\caption{Reflection coupling of two MALA kernels marginally 
 targetting distributions $P$ and $Q$ respectively \citep[see, e.g.][]{bou-rabee2020couplingAOAP}.}
 \label{algo:mala_reflection}
\end{algorithm}

\begin{algorithm}[!htb]
\DontPrintSemicolon
\KwIn{$(X_t, Y_t)$, unnormalized densities $p$ and $q$
of $P$ and $Q$ respectively, step sizes $\sigma_P$ and $\sigma_Q$} 
Sample $\epsilon \sim \mathcal{N}(0,I_d)$, $U^* \sim \Uniform([0,1])$. 
Calculate proposals 
\begin{flalign}
X^* &\defeq X_t + \frac{1}{2} \sigma_P^2 \nabla \log p(X_{t}) + \sigma_P \epsilon \\
Y^* &\defeq
\left\{\begin{matrix}
 X^* & \text{if } U^* \leq \frac{\mathcal{N}( X^* ; Y_t + \frac{1}{2} \sigma_P^2 \nabla \log p(Y_{t}), \sigma_P^2I_d )}{\mathcal{N}( X^* ; X_t + \frac{1}{2} \sigma_P^2 \nabla \log p(X_{t}), \sigma_P^2I_d )} \\
 Y_t + \frac{1}{2} \sigma_Q^2 \nabla \log q(Y_{t}) + \sigma_Q ( I_d - e e^\top ) \epsilon & \text{ otherwise} \text{ for } e = \frac{X_t-Y_t}{\|X_t-Y_t\|_2}. \\
\end{matrix}\right.
\end{flalign}

Sample $U_{CRN} \sim \Uniform([0,1])$. \;
\textbf{if} $U_{CRN} \leq \frac{p(X^*) \mathcal{N}( X^* ; X_t + \frac{1}{2} \sigma_P^2 \nabla \log p(X_{t}), \sigma_P^2I_d )}{p(X_t) \mathcal{N}( X_t ; X^* + \frac{1}{2} \sigma_P^2 \nabla \log p(X^*), \sigma_P^2I_d )}  $, \textbf{then} set \( X_{t+1} = X^* \) ; \textbf{else} set \( X_{t+1} = X_{t} \). \;
\textbf{if} $U_{CRN} \leq \frac{q(Y^*) \mathcal{N}( Y^* ; Y_t + \frac{1}{2} \sigma_Q^2 \nabla \log q(Y_{t}), \sigma_Q^2I_d )}{q(Y_t) \mathcal{N}( Y_t ; Y^* + \frac{1}{2} \sigma_Q^2 \nabla \log q(Y^*), \sigma_Q^2I_d )}  $, \textbf{then} set \( Y_{t+1} = Y^* \) ; \textbf{else} set \( Y_{t+1} = Y_{t} \). \;
\Return $(X_{t+1}, Y_{t+1})$
\caption{Reflection maximal coupling of two MALA kernels marginally 
 targetting distributions $P$ and $Q$ respectively \citep[see, e.g.][]{bou-rabee2020couplingAOAP}.}
 \label{algo:mala_reflection_maximal}
\end{algorithm}

\end{document}